\documentclass[11pt]{article}
\usepackage{graphicx} 
\usepackage[margin=1in]{geometry}
\usepackage{subcaption}
\usepackage{titletoc}
\usepackage{enumitem}
\usepackage{times}
\usepackage[compact]{titlesec}

\usepackage{amsmath}
\usepackage{amssymb}
\usepackage{mathtools}
\usepackage{amsthm}

\newtheoremstyle{tightexample}
  {1pt}                
  {1pt}                
  {\normalfont}        
  {}                   
  {\bfseries}          
  {.}                  
  {.5em}               
  {}                   

\usepackage{booktabs} 
\usepackage[noend,ruled]{algorithm2e} 
\usepackage{comment}
\usepackage[colorlinks=true,citecolor=black,linkcolor=blue,urlcolor=blue]{hyperref}

\usepackage{mathtools}
\usepackage{color}
\usepackage{natbib}

\usepackage{hyperref}
\usepackage{enumitem}

\usepackage{soul}

\newtheorem{theorem}{Theorem}
\newtheorem{corollary}{Corollary}
\newtheorem{lemma}{Lemma}
\newtheorem{proposition}{Proposition}
\theoremstyle{definition}

\theoremstyle{tightexample}
\newtheorem{remark}{Remark}
\newtheorem{example}{Example}

\DeclareMathOperator*{\argmax}{arg\,max}
\DeclareMathOperator*{\argmin}{arg\,min}

\title{Data-Driven Games with Coherent Risk Measures}
 \author{
   Bharat Gangwani\\
   Independent Researcher\\
   \texttt{bharatg.2020@economics.smu.edu.sg}
   \and
   Arunesh Sinha\\
   Management Science and Information Systems\\
   Rutgers University\\
   \texttt{arunesh.sinha@rutgers.edu}
 }

\date{}

\begin{document}

\begin{titlepage}
\maketitle
\begin{abstract}

    We introduce Coherent Utility Measure Games (CUMGs) in which players' uncertainty about the distribution of payoffs is modeled using coherent utility (risk) measures. Such measures, including mean semideviation risk and conditional value-at-risk, allow for interpretable notions of players' risk aversion while retaining formal equivalence to distributionally robust games. While CUMGs, which are a subclass of distributionally robust games, are continuous games in general, they can be viewed as finite games ``lifted'' to the mixed strategy space, which illustrates computational challenges. Prior results extend to guarantee equilibrium existence in data-driven CUMGs. For CUMGs parameterized by several popular risk measures, we show that the computation of exact equilibria lies in FIXP, even for two-player games, and approximate equilibria lie in PPAD. Separately, we derive direct complementarity formulations for exact equilibrium computation for these games, which grow with $K$, the number of data samples. Unlike standard games, these programs are not linear in a two-player setting. Next, we establish the existence of approximate equilibria in finite data-driven CUMGs with small supports in the players' pure actions, yielding a quasi-polynomial time approximation scheme (QPTAS); this, together with a sparse data subsample result, guides the search for such equilibria. We also develop a stochastic first-order approach for smoothed CUMGs using data mini-batches, with bounds linking first-order error to approximate equilibrium. We include numerical experiments exploring the structure of equilibrium in CUMGs and comparing the various approaches in this work.

\end{abstract}

\thispagestyle{empty}
\end{titlepage}

\section{Introduction}
In recent years, data-driven decision-making has become a central paradigm across economics, operations research (OR), and artificial intelligence (AI). As decisions increasingly rely on empirical data rather than fully known models, understanding and mitigating the impact of distributional uncertainty has emerged as a crucial challenge. In single-agent decision problems, a rich line of research has developed robust and risk-sensitive formulations, starting from classical Markowitz mean-variance
model~\citep{markowitz1952portfolio} to the unifying framework of coherent risk measures~\citep{artzner1999coherent}, which includes risk measures such as Conditional Value-at-Risk (CVaR) or Mean-semideviation (MSD). These formulations have played a foundational role in capturing aversion to uncertainty in finance~\citep{acharya2017measuring}, OR~\citep{dentcheva2024risk}, and also AI~\citep{lam2022risk}.

However, the extension of such distributional robustness principles \emph{based on risk measures} to multi-agent or strategic settings remains limited. Although robust equilibria and uncertainty-aware games have been studied (see related work), the incorporation of coherent risk measures into game-theoretic formulations remains nascent. This gap is consequential: coherent risk measures offer axiomatic rigor, interpretability, and a principled parameterization of risk attitudes such as tail or downside sensitivity. Thus, even when equilibrium responses are non-monotone in risk aversion, outcome changes can be traced directly to shifts in players’ risk preferences.
 By contrast, probability distance-based ambiguity sets (e.g., Wasserstein) entangle preference effects with statistical uncertainty, complicating comparative statics in strategic settings. We also note that, under standard duality results, optimization with a coherent risk measure is equivalent to a distributionally robust formulation over a suitable ambiguity set around a nominal distribution~\citep{ruszczynski2006optimization}. 

In this work, we explore data-driven games with a focus on risk-measure based games. In a reward setting, risk measures are called coherent utility measures. In the data-driven setting, we assume $K$ samples of the payoff matrix and treat the empirical distribution as the nominal distribution for various ambiguity sets. 
In this work, we use Mean-semideviation (MSD), Mean-deviation (MD), and CVaR as canonical examples of coherent utility measures. 
These risk measures are widely studied in the risk-management literature and offer easy interpretation, with mean–semideviation capturing downside risk and conditional value-at-risk quantifying tail risk. The contributions in this work are listed below:
\begin{enumerate}[noitemsep]
    \item We show that data-driven Coherent Utility Measure Games (CUMGs) are continuous concave games, and prior results~\citep{qu2017distributionally} establish the existence of a Distributionally Robust Equilibrium (DRE).

    \item We characterize the complexity of DRE computation: approximate DREs for MSD, MD, and CVaR are PPAD-complete, while exact DREs lie in FIXP and are FIXP-complete for three or more players.

    \item We derive complementarity formulations for general data-driven CUMGs and multilinear complementarity programs (MLCP) for MSD, MD, and CVaR.

    \item We prove a small-support result for general data-driven CUMGs, yielding a QPTAS for MSD, MD, and CVaR and motivating a practical heuristic search over small action and data-sample supports.

    \item We develop a stochastic first-order framework using data mini-batches, with unbiased gradient estimators and approximation guarantees.
\end{enumerate}
We complement these results with numerical experiments comparing various approaches proposed in this work. The contributions above suggest that this class of games shares characteristics of both finite and continuous games, placing CUMGs in an intermediate regime between them. Overall, our framework \emph{unifies coherent risk modeling and strategic interaction}, providing an interpretable and analytically structured approach to uncertainty-aware equilibrium analysis in data-driven games beyond probability-distance-based ambiguity sets. The next section introduces the formulation, and Table~\ref{tab:paperOrganization} provides a summary of the results and organization of the rest of the paper.

\begin{table}[t]
\centering
\caption{Organization and main results of the paper.}
\label{tab:paperOrganization}
\begin{tabular}{p{0.1\linewidth} p{0.77\linewidth}}
\toprule
\textbf{Section} & \textbf{Main results} \\
\midrule

Section~\ref{subsec:existence}
&
CUMGs are continuous rather than matrix games; DRE existence
(Theorem~\ref{thm::dreSMExist}). \\

Section~\ref{subsec:maincomplexity}
&
Exact DRE: FIXP results (Theorem~\ref{thm:exactcomplexity}), in
contrast to PPAD-completeness of exact NE in standard 2-player games; approximate DRE: PPAD results
(Theorem~\ref{thm:complexity}). \\

Section~\ref{sec:complementarity}
&
General complementarity characterization and MLCPs for MSD, MD, and CVaR.
Unlike LCPs for standard 2-player games, the 2-player formulations remain
nonlinear. \\

Section~\ref{subsec:scalable-action-cumg}
&
Small-support $\epsilon$-DREs over actions and data samples
(Theorem~\ref{thm:small_support_lifted_drg}); QPTAS for MSD, MD, and CVaR
(Corollary~\ref{cor:qptas}); heuristic small-support search. \\
Section~\ref{subsec:scalable-samples-cumg}
&
Stochastic first-order method for large $K$: data mini-batching yields unbiased
stochastic gradients and approximate-DRE guarantees. \\

Section~\ref{sec:experiments}
&
Empirical demonstration of the structure of DRE and a comparison of the scalability of approaches proposed in this work with the number of actions and data samples. \\
Appendix
&
Missing proofs and extra experimental results. \\

\bottomrule
\end{tabular}
\end{table}

\section{Preliminaries and Problem Formulation}

\paragraph{Preliminaries of coherent utility (risk) measures} Here we describe coherent risk measures. We define them in terms of gain or utility (more common in game theory), which are sometimes referred to as coherent utility measures~\citep{cheridito2006coherent}, unlike the more popular cost setting in the risk-measure literature. Let $\bar{\mathbb{R}}$ denote extended reals.
Let $\mathcal{X}: \Omega \rightarrow \mathbb{R}$ be the space of measurable functions on some sample space $\Omega$ with a given $\sigma$-algebra; intuitively, $\mathcal{X}$ can be thought of as the space of random variables. A utility functional $\rho:\mathcal{X}\rightarrow \bar{\mathbb{R}}$ is coherent iff it satisfies the following properties for any $X,Y \in \mathcal{X}$ where $X,Y$ represent random variables (utilities):
\begin{enumerate}[itemsep=0.0mm,parsep=0.2mm]
    \item[(A1)] \textbf{Concavity:} $\rho(\alpha X + (1-\alpha)Y) \geq \alpha\rho(X) + (1-\alpha)\rho(Y)$, and $\alpha\in[0,1]$
    \item[(A2)] \textbf{Monotonicity:} If $Y(\omega) \geq X(\omega), \forall \omega\in\Omega$, then $\rho(Y) \geq \rho(X)$
    \item[(A3)] \textbf{Translation equivariance:} If $a\in\mathbb{R}$, $\rho(X+a) = \rho(X) + a$
    \item[(A4)] \textbf{Positive homogeneity:} If $t>0$ and $X\in\mathcal{X}$, then $\rho(tX) = t\rho(X)$
\end{enumerate}
Coherent utility measures also admit a dual form, which connects them to distributional robustness. Take any \emph{nominal} probability measure $\mathbb{P}$ and consider the space of measurable functions $\mathcal{L}^p(\Omega, \Sigma, \mathbb{P})$ for $p \in [1, \infty)$, where $X \in \mathcal{L}^p(\Omega, \Sigma, \mathbb{P})$ means that the p-absolute moment of $X$ is bounded: $\mathbb{E}^{\mathbb{P}}[|X|^p] < \infty$. 
A concave function $f$ is proper iff $\text{dom}(f)\neq\emptyset$ and $f(\cdot)<+\infty$ for all points in its domain. Here, $\text{dom}(f)=\{x\mid f(x)> -\infty\}$ refers to the effective domain of $f$. 
The following result follows from prior work:
\begin{theorem} \cite[Theorem 2.2]{ruszczynski2006optimization} \label{thm:rs06} Let $\mathcal{P}$ be the set of probability measures over $(\Omega, \Sigma)$. Let $\mathcal{Y} = \left\{\mu \mid \frac{d\mu}{d\mathbb{P}} \in \mathcal{L}^q(\Omega, \Sigma, \mathbb{P}) \right\}$ be a set of measures, where $\frac{d\mu}{d\mathbb{P}}$ is the Radon-Nikodym derivative (probability density in non-measure theoretic sense). A proper, upper-semicontinuous, and concave utility functional $\rho$ is coherent if and only if it can be represented as
$$
\rho(X) = \inf_{\mathbb{Q} \in U}\mathbb{E}^\mathbb{Q}[X], \ \forall X\in \mathcal{L}^p(\Omega, \Sigma, \mathbb{P}), 
$$
where $U = \left\{\mathbb{Q} \in \mathcal{P} \cap \mathcal{Y} \mid  \mathbb{E}^\mathbb{Q}[X] \geq \rho(X),\forall X \in \mathcal{L}^p(\Omega, \Sigma, \mathbb{P} ) \right\}$ is a subset of probability measures.
\end{theorem}
The result above means that coherent utility (risk) measures are equivalent to a distributionally robust utility (cost) for a given uncertainty set $U$ defined w.r.t. a nominal distribution $\mathbb{P}$. Notably, the Radon-Nikodym derivative (or density) $\frac{d\mathbb{Q}}{d\mathbb{P}}$ exists and $\frac{d\textcolor{black}{\mathbb{Q}}}{d\mathbb{P}} \in \mathcal{L}^q(\Omega, \Sigma, \mathbb{P})$. In particular, a consequence of this existence, which we will use in the sequel, is that $\mathbb{Q}$ must be \emph{absolutely continuous} w.r.t. $\mathbb{P}$ ($\mathbb{Q}\ll\mathbb{P}$), meaning the support of $\mathbb{Q}$ is a subset of the support of $\mathbb{P}$.
\begin{example} \label{example:coherent}
Some popular examples of coherent utility measures are listed below. Note that all definitions are in terms of rewards; that is, the random variable $X$ specifies utility.
\begin{itemize}[noitemsep]
    \item (Mean-semideviation) Consider the following form with $p=1$ from~\citet{ruszczynski2006optimization}
  \[
      \rho_{\text{MSD},p}(X) =
\mathbb{E}^{\mathbb{P}}[X] - \gamma_s
\left\|
\max(0,\mathbb{E}^{\mathbb{P}}[X]-X)
\right\|_p  \text { where } \|Y\|_p=(\mathbb{E}^{\mathbb{P}}[|Y|^p])^{1/p}, \quad \gamma_s \in [0, 1].
  \]

  \item (Mean–deviation) Consider the following form with $p=1$ from~\citet{ruszczynski2006optimization}
  \[
    \rho_{\text{MD}}(X)
    =
    \mathbb{E}^{\mathbb{P}}[X] - 
    \gamma_d \left\|
 \mathbb{E}^{\mathbb{P}}[X] - X
\right\|_p  \text { where } \|Y\|_p=(\mathbb{E}^{\mathbb{P}}[|Y|^p])^{1/p}, \quad \gamma_d \in [0, 1/2^{1/p}].
\]
  \item (Conditional Value-at-Risk (CVaR) at level $\alpha \in(0,1)$) The well-known CVaR is defined for reward $X$ as $
    \mathrm{CVaR}_\alpha [X]=\sup_{z\in\mathbb{R}} \left\{
      z + \frac{1}{\alpha}\,\mathbb{E}^{\mathbb{P}}\!\big[\min(0,X-z)\big]
    \right\} 
  $, which measures the average of the worst 
$100\alpha$\% of rewards. A CVaR-based coherent utility functional can be written as
  \[
    \rho_{\text{CVaR}}(X)
    =
    (1-\gamma_c)\,\mathbb{E}^{\mathbb{P}}[X] + \gamma_c\,\mathrm{CVaR}_\alpha[X] = \mathbb{E}^{\mathbb{P}}[X] - \gamma_c \big(\mathbb{E}^{\mathbb{P}}[X] - \mathrm{CVaR}_\alpha[X] \big),
    \; \gamma_c \in [0,1].
  \]
  Note that with $X$ as reward, $\alpha$ is a small value which defines the lower tail of the distribution of $X$. The above form is from~\citet{ruszczynski2006optimization}.
\end{itemize}
\end{example}
Hereafter, MSD and MD denote the order-one case ($p=1$) analyzed in this paper; we suppress the order parameter, writing, e.g., $\rho_{\mathrm{MSD}}$ for $\rho_{\mathrm{MSD},1}$. The order-$p$ versions are mentioned only to note that, for $p>1$, the deviation (second) term is nonlinear in the sample payoffs.

\paragraph{Preliminaries of one-shot matrix games} 
We denote vectors with a boldface $\mathbf{v}$ and the $i$-th component as $v_i$.
Consider a game with $m$ players, with $n_i > 1$ actions available to player $i$ where $i \in \{1, \dots , m\}$. $m$ and $n_i$ are finite. Let $\mathbf{A}_i$ denote the set of $n_i$ actions available to player $i$ and $\mathbf{a} = (a_{j_1}\dots a_{j_m})$ is a pure strategy profile chosen by the players where $j_i \in \{1,\dots ,n_i\}$ indexes the action taken by player $i$. Then $u_i: \mathbf{A}\rightarrow\mathbb{R}$ denotes the utility of player $i$ under pure strategy profile $\mathbf{a} \in \mathbf{A} = \prod_{i=1}^m \mathbf{A}_i$. We denote the mixed-strategy tuple of player $i$ as $\mathbf{x}_i \in \mathbf{X}_i$ from their strategy space $\mathbf{X}_i = \left\{\mathbf{x}_i\mid\mathbf{x}_i \in \prod_{i=1}^{n_i}[0,1],\mathbf{1}_{n_i}^T\mathbf{x}_i = 1 \right\}$. 
Then $x_i(a_{j_i}) \in [0,1]$ is a component of this $\mathbf{x}_i$ and denotes the probability of player $i$ choosing action $a_{j_i} \in \mathbf{A}_i$. Following standard convention, let $\mathbf{x}_{-i}$ denote the mixed strategy of all players \textit{except} player $i$, i.e. $\mathbf{x}_{-i} = [\mathbf{x}_1, \dots , \mathbf{x}_{i-1}, \mathbf{x}_{i+1}, \dots\mathbf{x}_m]$. Then, with slight abuse of and overloading notation, $u_i: \prod_{i=1}^m\mathbf{X}_i \rightarrow \mathbb{R}$ is the expected utility of player $i$ of playing mixed-strategy $\mathbf{x}_i$ given all other players' mixed strategies $\mathbf{x}_{-i}$:
\begin{equation}
    \label{eq::expUtil}
    u_i(\mathbf{x}_i, \mathbf{x}_{-i}) = \sum_{j_1=1}^{n_1}\dots \sum_{j_m=1}^{n_m}u_i(a_{j_1},\dots ,a_{j_m})\cdot x_{1}(a_{j_1})\cdot\dots \cdot x_{m}(a_{j_m}) = \sum_{\mathbf{a}\in \mathbf{A}}u_i(\mathbf{a})\prod_{s=1}^m x_s(a_{j_s})
\end{equation}
\paragraph{Random utilities in games} We consider utilities that are affected by \emph{randomness exogenous to the game}.
Let $\xi$ denote a random vector supported on $\Xi \subseteq \mathbb{R}^k$ for some positive integer $k$. Let $\mathcal{M}(\Xi)$ be the set of all probability measures on $\Xi$. All players are aware of $\Xi$. We abuse notation slightly to use the same letter $\xi$ for the random variable and its realization. Then, player $i$'s utility given pure strategy profile $\mathbf{a}$ and a realization $\xi$ is  $u_i(\mathbf{a} \mid \xi)$, which is a function from $\prod_{i=1}^m \mathbf{A}_i \times \Xi $ to $\mathbb{R}$. Consequently, given the utility function and a realization $\xi$, the player $i$'s expected utility is defined in a similar way to Equation \eqref{eq::expUtil}:
$
    \label{eq::expUtilXi}
    u_i(\mathbf{x}_i, \mathbf{x}_{-i}\mid \xi) 
    \!=\! \sum_{\mathbf{a}\in \mathbf{A}} u_i(\mathbf{a\mid\xi})\prod_{s=1}^m x_s(a_{j_s}).
$
Given probability measure $\mathbb{Q} \in \mathcal{M}(\Xi)$, we can write, \begin{equation} \label{eq:puresplit}
\mathbb{E}^{\mathbb{Q}}\left[u_i(\mathbf{x}_i, \mathbf{x}_{-i}\mid \xi)\right] = \sum_{\mathbf{a}\in \mathbf{A}}\mathbb{E}^\mathbb{Q}[u_i(\mathbf{a}\mid\xi)]\prod_{s=1}^m x_s(a_{j_s}) 
\end{equation}
Here, we introduce the shorthand $\mu^{\mathbb{Q}}_i(\mathbf{x}_i, \mathbf{x}_{-i}) = \mathbb{E}^{\mathbb{Q}}[u_i(\mathbf{x}_i, \mathbf{x}_{-i}\mid \xi)] $ and 
$\mu^{\mathbb{Q}}_i(\mathbf{a}) = \mathbb{E}^{\mathbb{Q}}[u_i(\mathbf{a}\mid \xi)] $.

\paragraph{Distributionally robust game and equilibrium} A \emph{distributionally robust game} (DRG) is one where there is uncertainty about the distribution of $\xi$. The uncertainty is given by an ambiguity set $U\subseteq\mathcal{M}(\Xi)$ of probability measures over $\Xi$ and is known by all players. In a \emph{distributionally robust equilibrium} (DRE), players maximize expected utility under the worst $\mathbb{Q} \in U$. Hence, $(\mathbf{x}^*_i, \mathbf{x}^*_{-i})$ is a DRE if:
\begin{equation}
\label{def::drEquilibrium}  
\mathbf{x}^*_i \in \argmax_{\mathbf{x}_i \in \mathbf{X}_i}\inf_{\mathbb{Q} \in U} \mu^{\mathbb{Q}}_i(\mathbf{x}_i, \mathbf{x}^*_{-i}), \; \forall i \in \{1,\dots ,m\}
\end{equation}
We also define $\rho_i(\mathbf{x}_i, \mathbf{x}_{-i}) \coloneq \inf_{\mathbb{Q} \in U} \mu^{\mathbb{Q}}_i(\mathbf{x}_i, \mathbf{x}_{-i}) = \inf_{\mathbb{Q} \in U} \sum_{\mathbf{a}\in \mathbf{A}}\mu^{\mathbb{Q}}_i(\mathbf{a}) \prod_{s=1}^m x_s(a_{j_s})  $. 

\paragraph{Coherent utility measure game (CUMG) and distributional robustness}
The utility of mixed strategy $(\mathbf{x}_i, \mathbf{x}_{-i})$ for player $i$ using a coherent utility measure $\rho$ is given as
$
    \rho_i(\mathbf{x}_i, \mathbf{x}_{-i}) = \rho(u_i(\mathbf{x}_i, \mathbf{x}_{-i}\mid \xi)).
$
Note that we overload the notation $\rho_i$ above, which is justified because the game with the above $\rho_i$ utility is 
a DRG due to Theorem~\ref{thm:rs06}. We call these games \emph{coherent utility measure games} (CUMGs); every CUMG is a DRG. 
As stated earlier, a key structural property of popular CUMGs (MSD, CVaR) is that the optimization over the ambiguity set $U$ has a closed form solution. 
We further note (and later prove in \textcolor{black}{Proposition~\ref{prop:DRGcont}}) that any DRG, including CUMG, is a continuous game and not a matrix game. This means that the strategy space of any DRG is the continuous space of mixed strategies $\prod_{i=1}^m \mathbf{X}_i$ and the utility $\rho_i$ generally \emph{cannot} be written as a convex combination of the utility of pure strategies, as done for empirical games in Equation~\ref{eq:puresplit}.

\paragraph{Nature of ambiguity set} For data-driven setting, often the ambiguity set is given as $U = \{\mathbb{Q} \mid d(\mathbb{Q},\mathbb{P}) \leq \eta \}$ where the nominal $\mathbb{P}$ is chosen as the \emph{empirical distribution} given by the data samples and $d$ is a distance between probability distributions; common choices for $d$ include Wasserstein distance, f-divergences such as KL and $\chi^2$ divergence. In particular, for CUMGs, one has a closed form solution for the inner optimization over the ambiguity set, yielding a closed form formula for $\rho_i$ as shown in Example~\ref{example:coherent}. We assume $K$ \emph{unique} data samples with the corresponding samples of utilities for player $i$ for all strategy profiles $\mathbf{a}$ given as $u_i(\mathbf{a} \mid \xi_1), \ldots, u_i(\mathbf{a} \mid \xi_K)$. Note that $K$ may be less than the actual number of data samples due to identical sample realizations, which is possible when $\mathbb{P}$ is defined over a discrete space. Then, $\mu^{\mathbb{P}}_i(\mathbf{a}) = \mathbb{E}^{\mathbb{P}}[u_i(\mathbf{a}\mid \xi)] = \sum_{k=1}^K \mathbb P(\xi_k) u_i(\mathbf{a} \mid \xi_k)$.

\section{Nature of DRG and Existence of Equilibrium}
\label{subsec:existence}
In this section, we present structural results relating to Distributionally Robust Equilibria in a data-driven setting with a focus on CUMG.
We start by showing that distributionally robust games are not matrix games. 
\begin{proposition}[DRG are continuous games]
    \label{prop:DRGcont}
    In general, DRG games are continuous games with pure strategy space as $\prod_{i=1}^m\mathbf{X}_i$ and they cannot be viewed as a matrix game in which the pure strategy space is $\mathbf{A}$ and payoffs are component-wise infima $\inf_{\mathbb{Q} \in U} \mathbb{E}^\mathbb{Q}[u_i(\mathbf{a} \mid \xi)]$ over the ambiguity set $U$. Mathematically,
    $
        \rho_i(\mathbf{x}_i, \mathbf{x}_{-i}) = \inf_{\mathbb{Q}\in U} \mu^{\mathbb{Q}}_i(\mathbf{x}_i, \mathbf{x}_{-i}) \geq \sum_{\mathbf{a}\in\mathbf{A}} \inf_{\mathbb{Q} \in U} \mu^{\mathbb{Q}}_i(\mathbf{a}) \prod_{s=1}^m x_s(a_{j_s})
    $. The inequality may be strict; thus, in general, the robust utility of a mixed-strategy profile cannot be obtained as the convex combination of the robust utilities of its constituent pure-strategy profiles.
\end{proposition}

\begin{example}\label{ex:coordinationGame}
We use the following two-sample DRG as an illustrative and running example to show the continuous nature of DRG (and CUMG) games.
\[
\begin{array}{c@{\qquad}c@{\qquad}c}
\begin{array}{c}
\begin{array}{c|cc}
 & L & R\\ \hline
U & (1,1) & (0,0)\\
D & (0,0) & (1,1)
\end{array}
\end{array}
&
\begin{array}{c}
\begin{array}{c|cc}
 & L & R\\ \hline
U & (0,0) & (1,1)\\
D & (1,1) & (0,0)
\end{array}
\end{array}
&
\parbox{0.50\linewidth}{\raggedright
A DRG with $K=2$ payoff samples. The empirical probability is $\hat{p}=\tfrac12$. Here $
\Xi=\{\xi_{1},\xi_{2}\}
$,
where \(\xi_{1}\) corresponds to left bimatrix and \(\xi_{2}\) to right bimatrix.
}
\end{array}
\]
We consider two possible DRG for illustration: (1) an ambiguity set centered around $\hat{p}$ given by $[0.3,0.7]$ (this can be induced by Wasserstein distance) and (2) a MSD game. 
For any mixed-strategy profile \((\mathbf{x}_i,\mathbf{x}_{-i})\) and any \(\mathbb{Q}\in\mathcal{M}(\Xi)\),
$
\mu_i^{\mathbb Q}(\mathbf{x}_i,\mathbf{x}_{-i})
=
\sum_{\mathbf a\in\mathbf A}
\bigl(
\mathbb Q(\xi_{1})\,u_i(\mathbf a\mid\xi_{1}) +
\mathbb Q(\xi_{2})\,u_i(\mathbf a\mid\xi_{2})
\bigr)
\prod_{s=1}^m x_s(a_{j_s}).
$
For the ambiguity set
$
U=\bigl\{\mathbb Q:\mathbb Q(\xi_{1})=p,\;p\in[0.3,0.7]\bigr\},
$
the distributionally robust utility is
$
\rho_i(\mathbf{x}_i,\mathbf{x}_{-i})
=
\inf_{p\in[0.3,0.7]}
\sum_{\mathbf a\in\mathbf A}
\bigl(
p\,u_i(\mathbf a\mid\xi_{1})
+
(1-p)\,u_i(\mathbf a\mid\xi_{2})
\bigr)
\prod_{s=1}^m x_s(a_{j_s}).
$
Using the utilities given above, we get
$
\rho_i(\mathbf{x}_i,\mathbf{x}_{-i})
=
\inf_{p\in[0.3,0.7]}
\Big[
p\bigl(x_1(U)x_2(L)+x_1(D)x_2(R)\bigr)
+
(1-p)\bigl(x_1(U)x_2(R)+x_1(D)x_2(L)\bigr)
\Big].
$
Further, using the fact that mixed strategy probabilities sum to one, we get
\begin{equation}
\rho_i(\mathbf{x}_i,\mathbf{x}_{-i})
=
\inf_{p\in[0.3,0.7]}
\Big[
 \ \frac12 +
\frac{2p-1}{2}\,
\bigl(1-2x_1(U)\bigr)\bigl(1-2x_2(L)\bigr)
\Big]. \label{eq:examplepayoff}
\end{equation}
The payoff $\rho_i$ is $0.5$ whenever either $x_1(U) = 0.5$ or $x_2(L) = 0.5$, for all other cases the payoff is less than $0.5$.
A few immediate conclusions are: (1) $x_1(U) = 0.5$ is a DRE and so is $x_2(L) = 0.5$ and (2) the payoff of a mixed strategy is not a convex combination of payoff of the pure strategies in support, as $\rho_i(x_1, L) = 0.5$ for $x_1(U) = 0.5$ but $\rho_i(U, L) = 0.3$ and $\rho_i(D, L) = 0.3$, which also implies that actions in the support of the DRE do not yield the same payoff as the DRE itself (unlike in a standard game).

For the MSD game with \(\gamma_s \in[0,1]\), one can derive
$
\rho_i(\mathbf{x}_i,\mathbf{x}_{-i})
=
\frac12
-
\frac{\gamma_s}{4}\,
\left|
\bigl(1-2x_1(U)\bigr)\bigl(1-2x_2(L)\bigr)
\right|.
$
We have the same conclusions about the equilibrium and continuous nature of this game as for the DRG. 
However, as we will show later in Lemma~\ref{lem:generalDRGComplementarity}, pure actions in the support of the mixed strategy of a player have equal risk-adjusted payoffs in the DRE equilibrium for data-driven CUMGs.
\end{example}
\paragraph{Equilibrium} Next, we state a general equilibrium existence result from ~\citet{qu2017distributionally}. This also establishes the existence of equilibrium for any CUMG. Recall the notation $\mu^{\mathbb{Q}}_i(\mathbf{a}) = \mathbb{E}^{\mathbb{Q}}[u_i(\mathbf{a}\mid \xi)]$. We still provide a part of the proof (in appendix) since our proof is much more succinct than~\citet{qu2017distributionally}.

\begin{theorem}[\citet{qu2017distributionally}]
    \label{thm::dreSMExist}
    Any finite m-player distributionally robust game with an underlying non-cooperative, simultaneous move, one-shot matrix game and an ambiguity set $U$ has a distributionally robust  equilibrium in mixed strategies, if for all $\mathbb{Q} \in U$ the first moment of utilities for any pure strategy $\mathbf{a}$ is bounded, i.e., $|\mu^{\mathbb{Q}}_i(\mathbf{a})| \leq M < \infty $ for all  $\mathbf{a} \in  \mathbf{A}$ and $\mathbb{Q} \in U$ for some $M$.
\end{theorem}

\begin{remark}
    Below we verify the condition in Theorem~\ref{thm::dreSMExist} for several cases. Recall that the empirical distribution $\mathbb{P}$ in our data-driven setting has a finite support. Also, finite-valued samples imply $\mu^\mathbb{P}_i(\mathbf{a})$ is bounded. 
    \begin{enumerate}[itemsep=0.0mm,parsep=0.2mm]
        \item (CUMG) As stated after Theorem~\ref{thm:rs06}, coherent utility measures have a distributionally robust form with the uncertainty set $U$ containing distributions that are absolutely continuous w.r.t. $\mathbb{P}$, meaning support of $\mathbb{Q}$ is a subset of support of $\mathbb{P}$. Thus, $\mathbb{Q}$ is a discrete probability distribution over a finite support, hence the first moment $\mu^\mathbb{Q}_i(\mathbf{a})$ of any such distribution $\mathbb{Q}$ is bounded.
        \item (f-divergence) Any distribution $\mathbb{Q}$ within the f-divergence ball ($d(\mathbb{Q}, \mathbb{P}) \leq \xi$) is absolutely continuous w.r.t. $\mathbb{P}$, and then the result follow from same reasoning as in CUMGs.
        \item (Wasserstein) It is well known that $|\mathbb{E}^\mathbb{Q}(g(\xi)) - \mathbb{E}^\mathbb{P}(g(\xi))| \leq L_g W_1(\mathbb{Q}, \mathbb{P})$ for any random variable $\xi$ and $L_g$-Lipschitz function $g$, which follows from Kantorovich duality for Wasserstein-1 distance, and also $W_1(\mathbb{Q}, \mathbb{P}) \leq W_p(\mathbb{Q}, \mathbb{P})$ for any  $\mathbb{Q}, \mathbb{P}, p> 1$. Taken together, this implies that if $u_i(\mathbf{a}\mid \xi)$ is Lipschitz in $\xi$ and $\mu^{\mathbb{P}}_i(\mathbf{a})$ is bounded then $\mu^{\mathbb{Q}}_i(\mathbf{a})$ is also bounded for any distribution $\mathbb{Q}$ in the Wasserstein ball of finite radius centered at $\mathbb{P}$. 
    \end{enumerate}
\end{remark}

Next, we present a characterization of CUMG that will be used for proofs in the sequel.
Identify any distribution over the sample support $\Xi=\{\xi_1,\ldots,\xi_K\}$ with a vector in the simplex $\Delta_K$. For $\mathbf q_i\in \Delta_K$, define
$
P_i(\mathbf{x}_i,\mathbf{x}_{-i},\mathbf q_i)
:=
\sum_{\mathbf a\in\mathbf A}\sum_{k=1}^K
q_{i,k}u_i(\mathbf a\mid \xi_k)
\prod_{h=1}^m x_h(a_h).
$
Thus, $P_i(\mathbf{x}_i,\mathbf{x}_{-i},\mathbf q_i)$ is player $i$'s expected payoff when the distribution over data samples is $q_i$. The robust payoff is
$
\rho_i(\mathbf{x}_i,\mathbf{x}_{-i}) :=
\min_{\mathbf q_i\in U} P_i(\mathbf{x}_i,\mathbf{x}_{-i},q_i),
$
where $U\subseteq \Delta_K$ is the ambiguity set, as defined in Theorem~\ref{thm:rs06}.  We present the following extended game formulation that provides insights about CUMGs. 

\paragraph{Extended game}
Given a CUMG, define the extended $2m$-player game as follows. For each original player $i$, introduce an adversarial player $j_i$. Original player $i$ chooses $\mathbf{x}_i\in \mathbf X_i$ and receives payoff
$
P_i(\mathbf{x}_i,\mathbf{x}_{-i},\mathbf q_{j_i}).
$
Adversarial player $j_i$ chooses $\mathbf q_{j_i}\in U$ and receives payoff
$
-P_i(\mathbf{x}_i,\mathbf{x}_{-i},\mathbf q_{j_i}).
$
Thus, the adversarial player $j_i$ minimizes player $i$'s payoff over the set $U$. This extended game can be viewed as a graphical game in its payoff-dependence structure: adversary $j_i$ affects only player $i$'s payoff, and adversaries do not directly interact with one another. However, it is not a standard finite graphical game, since the strategy spaces are continuous: original players choose mixed strategies and each adversary chooses from the set $U$, which need not be the full probability simplex $\Delta_K$. The following result relates the original and extended game.

\begin{lemma}[Extended-game characterization of DRE]
\label{lem:lifted_game_dre}
Consider a data-driven CUMG in which, for every player $i$, the ambiguity set $U\subseteq \Delta_K$ is nonempty, compact, and convex. Then $\mathbf{x}^*$ is a DRE of the original CUMG if and only if there exists $\mathbf q^*=(\mathbf q_1^*,\ldots,\mathbf q_m^*)$, with $\mathbf q_i^*\in U$ for every $i$, such that $(\mathbf{x}^*,\mathbf q^*)$ is a Nash equilibrium of the extended $2m$-player game.
\end{lemma}

\section{Computational Structure and Algorithms}
\label{sec:algorithms}
We study equilibrium computation from three complementary perspectives. First, we characterize the complexity of DRE computation, establishing PPAD results for approximate DREs and FIXP results for exact DREs, including the MSD, MD, and CVaR games considered here. Second, independently of these complexity results, we derive direct complementarity formulations that characterize and compute exact DREs in the original strategy and risk variables. Third, we study scalable approximation: a small-support theorem gives sparse approximate equilibria and motivates support search, while a stochastic first-order method targets settings with large numbers of data samples.

\subsection{Complexity of DRE Computation}
\label{subsec:maincomplexity}

We first characterize the computational complexity of finding DREs. Throughout
this subsection, we assume $0\leq u_i(\mathbf a\mid\xi)\leq 1$ for all players
$i$, pure strategy profiles $\mathbf a$, and payoff samples $\xi$. All payoff entries, nominal probabilities $\mathbb P(\xi_k)$, and risk parameters are rational numbers represented in binary. The listed samples are support points, so $\mathbb P(\xi_k)>0$ for every $k$. In the approximate-computation result, $\epsilon>0$ is also a rational approximation parameter supplied as part of the input in binary. Also, throughout this subsection, we consider the nontrivial risk-averse case with
positive robustness parameters, i.e.,
$\gamma_s,\gamma_d,\gamma_c>0$; the case $\gamma=0$ reduces to the
standard expected-payoff game.
We begin with exact equilibrium computation. 
To establish the corresponding result for
CUMGs, we use the FIXP framework for concave games of
\citet{filosratsikas2023fixp}.  We briefly summarize the computational notion
needed for stating our result. A correspondence is a set-valued mapping;
in particular, $\partial_{\mathbf{x}_i}\rho_i(\mathbf{x}_i,\mathbf{x}_{-i})$
maps a strategy profile to the set of valid supergradients of player $i$'s
robust payoff. We use the pseudogate framework of
\citet{filosratsikas2023fixp} to represent the supergradient needed in the
exact equilibrium computation.  A pseudogate represents such a set-valued mapping by an
algebraic circuit with additional internal variables. The circuit outputs a
candidate vector $\mathbf g_i \in \partial_{\mathbf{x}_i}\rho_i(\mathbf{x}_i,\mathbf{x}_{-i})$ together with updated values of these internal
variables, with the guarantee that whenever the internal variables are
unchanged by the circuit, the output satisfies
$\mathbf g_i\in
\partial_{\mathbf{x}_i}\rho_i(\mathbf{x}_i,\mathbf{x}_{-i})$.
Thus, the pseudogate need not explicitly represent the entire
superdifferential; it only needs to produce one valid supergradient at a
fixed point. A polynomial-size pseudogate has circuit size and a number of
internal variables polynomial in the encoding size of the game.

\begin{theorem}[Complexity of exact DRE]
\label{thm:exactcomplexity}
Under the assumptions in this subsection, the following hold.
\begin{enumerate}
    \item Consider any encoded class of finite data-driven CUMGs for which,
    for every player $i$, there is a polynomial-size pseudogate for
    $\partial_{\mathbf{x}_i}\rho_i(\mathbf{x}_i,\mathbf{x}_{-i})$.
    Computing an exact DRE for this class lies in FIXP. If the class contains
    deterministic single data sample CUMGs corresponding to arbitrary finite
    normal-form games with three or more players, then computing an exact
    DRE is FIXP-complete.

    \item The polynomial-size pseudogate condition in part (1) holds for the MSD,
    MD, and CVaR games considered in this paper (with $\gamma > 0$). Thus, exact DRE computation
for each lies in FIXP and, for three or more players, is FIXP-complete.
\end{enumerate}
\end{theorem}

The general membership result follows by viewing a CUMG as a concave game and
applying the exact-equilibrium result of
\citet{filosratsikas2023fixp}. For the three specific games, we show that the required superdifferential pseudogates can be
constructed from their finite-sample dual formulations. In particular, the
dual form from Theorem~\ref{thm:rs06} are polynomial-size linear programs whose optimal solutions
determine supergradients of the corresponding robust payoff functions. The
full proof is given in Appendix~\ref{subsec:complexity}. For two-player standard finite games, exact Nash equilibrium is PPAD-complete and admits rational solutions; in contrast, two-player CUMGs can have games with only irrational DREs (Appendix~\ref{subsec:irrational}), while our exact results place MSD, MD, and CVaR in FIXP.

We next consider approximate equilibrium computation. For $\epsilon>0$, an
$\epsilon$-approximate DRE is a profile $\mathbf{x}^*$ satisfying
$
\rho_i(\mathbf{x}_i^*,\mathbf{x}_{-i}^*)
\geq
\max_{\mathbf{x}_i\in\mathbf X_i}
\rho_i(\mathbf{x}_i,\mathbf{x}_{-i}^*)-\epsilon,
\;
\forall i\in\{1,\ldots,m\}.
$
Allowing approximation places the equilibrium-computation problem in the
PPAD framework. Our result again first gives a general condition for CUMGs
and then verifies it for the three coherent utility measures considered in
this paper.

\begin{theorem}[Complexity of approximate DRE]
\label{thm:complexity}
Under the assumptions in this subsection, the following hold.
\begin{enumerate}
\item Consider any encoded class of finite data-driven CUMGs for which, whenever $K\geq2$, the projection $\Pi(U)$ obtained by dropping the last probability coordinate has a polynomial-size strong separation-oracle circuit and contains a Euclidean ball of known radius $r_U$ of polynomial encoding length. Computing an $\epsilon$-DRE for this class lies in PPAD. If the class contains the deterministic single-sample CUMGs corresponding to arbitrary finite normal-form games, then the problem is PPAD-complete.
\item The conditions in part (1) hold for the MSD, MD, and CVaR games considered in this paper with positive robustness parameters. Hence computing an $\epsilon$-DRE for each of these three classes is PPAD-complete.
\end{enumerate}
\end{theorem}

For the general membership result, we reduce DRE computation to approximate
equilibrium computation in concave games~\citep{papadimitriou2023computational}.
For MSD, MD, and CVaR, we then establish the conditions in part (1) of the theorem above. The complete
proof is given in Appendix~\ref{subsec:complexity}.

\subsection{Direct Complementarity Formulations for Exact DRE} \label{sec:complementarity}
In this sub-section, for the sake of easy notation, we assume that $0 \leq u_i(\mathbf{a\mid\xi}) \leq 1$ for all $\mathbf{a} \in \mathbf{A}$.
We investigate the technique of complementarity programs. Before diving into the details, we highlight a few high level aspects. As noted in Proposition~\ref{prop:DRGcont}, DRG (including CUMG) are continuous games, which are typically not solved as a complementarity program due to the infinite pure strategy space. However, CUMGs (and DRGs) can be viewed as a ``lifted game'' where the pure strategy space of the CUMG is the mixed strategy space $\prod_{i=1}^m\mathbf{X}_i$ of an underlying finite game. 
This view enables the complementary program approach. We introduce the $\perp$ notation for a \emph{complementarity constraint}, which enforces that, for two nonnegative scalar quantities $u$ and $v$, at least one of them must be zero. Formally, the meaning of $\perp$ notation is as follows:
$
0 \leq u \;\perp\; v \geq 0 \iff u \geq 0,\; v \geq 0,\; u\,v = 0.
$
First, for a fixed opponents' profile $\mathbf{x}_{-i}$, player $i$'s best-response problem in the data-driven CUMG is
$$
\max_{\mathbf{x}_i\in\mathbf{X}_i}\rho_i(\mathbf{x}_i,\mathbf{x}_{-i}), \text{ where } \rho_i(\mathbf{x}_i,\mathbf{x}_{-i})=\rho(u_i(\mathbf{x}_i,\mathbf{x}_{-i}\mid \xi)) \text{ and } \mathbf{X}_i=\{\mathbf{x}_i\in\mathbb{R}^{n_i}:\mathbf{1}^\top\mathbf{x}_i=1,\ \mathbf{x}_i\geq 0\}.
$$
As $u_i(\mathbf{x}_i,\mathbf{x}_{-i}\mid \xi)$ is affine in $\mathbf{x}_i$ and $\rho$ is concave, $\rho_i(\cdot,\mathbf{x}_{-i})$ is concave, the following lemma gives the KKT (complementarity) characterization of DRE.

\begin{lemma}[General complementarity forms for data-driven DRGs]
\label{lem:generalDRGComplementarity}
As defined earlier, consider a data-driven DRG induced by $\rho$ with finite samples $\{\xi_1,\dots,\xi_K\}$ and empirical nominal distribution $\mathbb{P}$. Then a profile $\mathbf{x}^*$ is a DRE iff the following condition holds
\paragraph{Primal form} There exist $\alpha_i\in\mathbb{R}$ and a supergradient $\mathbf{g}_i\in\partial_{\mathbf{x}_i}\rho_i(\mathbf{x}_i^*,\mathbf{x}_{-i}^*)$ such that
\begin{align}
    &\mathbf{1}^\top\mathbf{x}_i^*=1,\qquad \mathbf{x}_i^*\geq 0, \qquad 0\leq \alpha_i-g_{i,l}\perp x_i^*(a_l)\geq 0,
    \; \forall l\in\{1,\dots,n_i\}. \label{eq:generalPrimalMCP}
\end{align}
Here $\alpha_i$ is the Lagrange multiplier for the equality constraint $\mathbf{1}^\top\mathbf{x}_i=1$ in player $i$'s best-response problem at equilibrium. Moreover, $\alpha_i=\max_{1\leq l\leq n_i}g_{i,l}$; hence every action used with positive probability in the equilibrium has maximal risk-adjusted marginal payoff.

\paragraph{Dual form} For the data-driven CUMG that we consider, $\rho$ admits the dual representation in Theorem~\ref{thm:rs06} with risk envelope $U$ as defined earlier. Any $\mathbb{Q}\ll\mathbb{P}$ is identified with a vector $\mathbf{q}=(q_1,\dots,q_K)$, where $q_k=\mathbb{Q}(\xi_k)$. For player $i$, define worst-case dual distributions for player $i$ at $\mathbf{x}^*$:
$
U_i^*(\mathbf{x}^*)=
\argmin_{\mathbf{q}_i\in U}
\sum_{k=1}^K q_{i,k}u_i(\mathbf{x}_i^*,\mathbf{x}_{-i}^*\mid \xi_k).
$ 
For $\mathbf{q}_i\in U_i^*(\mathbf{x}^*)$, define risk-adjusted action values $\mathbf{v}_i(\mathbf{q}_i,\mathbf{x}_{-i}^*)\in\mathbb{R}^{n_i}$ componentwise for each action $a_l$ by
$
v_{i,l}(\mathbf{q}_i,\mathbf{x}_{-i}^*)=
\sum_{k=1}^K q_{i,k}u_i(a_l,\mathbf{x}_{-i}^*\mid \xi_k),
\quad l\in\{1,\dots,n_i\}.
$ 
Then there exist $\alpha_i\in\mathbb{R}$ and
risk-adjusted action values $\mathbf{v}_i\in\operatorname{conv}\{\mathbf{v}_i(\mathbf{q}_i,\mathbf{x}_{-i}^*):\mathbf{q}_i\in U_i^*(\mathbf{x}^*)\}$ such that
\begin{align}
    &\mathbf{1}^\top\mathbf{x}_i^*=1,\qquad \mathbf{x}_i^*\geq 0, \qquad 0\leq \alpha_i-v_{i,l}\perp x_i^*(a_l)\geq 0,
    \; \forall l\in\{1,\dots,n_i\}. \label{eq:generalDualMCP}
\end{align}
The multiplier $\alpha_i$ in this dual form is the same simplex KKT multiplier as in the primal form; the difference is that the supergradient vector $\mathbf{g}_i$ is represented through the vector $\mathbf{v}_i$.
In particular, if $U_i^*(\mathbf{x}^*)=\{\mathbf{q}_i^*\}$ is a singleton, then risk-adjusted payoff $v_{i,l}=\sum_{k=1}^K q_{i,k}^*u_i(a_l,\mathbf{x}_{-i}^*\mid \xi_k)$ is same for every $l$ in support of $\mathbf x^*_i$.
\end{lemma}



Next, we derive complementarity programs for CVaR game from first principles for the sake of clarity; this complementarity program turns out to be multi-linear.
We also obtain similar programs for the simpler MSD and MD games (Appendix~\ref{subsec:MSDcomplementarity}).
In Appendix~\ref{sec:dualSpecificCUMG}, we show that these programs can also be derived as special cases of the dual formulation above. We also show in Appendix~\ref{sec:dualSpecificCUMG} that, for certain CUMGs
such as order-$p$ mean-semideviation, the direct complementarity formulation
obtained from the risk-measure representation is not multilinear.

\paragraph{MLCP for CVaR Game} 
We consider the problem of a single player $i$, given $\mathbf x_{-i}$
\begin{align}\label{opt::cVaRMax}
    \max_{\mathbf{x}_i \in \mathbf{X}_i} & \quad  (1-\gamma_c)  \mu^{\mathbb{P}}_i(\mathbf{x}_i, \mathbf{x}_{-i})  + \gamma_c \max_{z_i \in\mathbb{R}}\Big[ z_i  + \frac{1}{\alpha} \sum_{k=1}^K \mathbb{P}(\xi_k) \min\left( 0, u_i(\mathbf{x}_i, \mathbf{x}_{-i}\mid \xi_k) - z_i \right) \Big] 
\end{align}

Let the auxiliary variable $\nu_{i,k}$ denote $ \min(0, u_i(\mathbf{x}_i, \mathbf{x}_{-i}\mid \xi_k ) - z_i )$. Then, the summation in the objective is $(1/\alpha)\sum_{k=1}^K \mathbb{P}(\xi_k)\nu_{i,k}$. Next, we write the terms in the optimization explicitly in the optimization variables, that is, $\mu^{\mathbb{P}}_i(\mathbf{x}_i, \mathbf{x}_{-i}) = \sum_{l=1}^{n_i} x_i(a_l) \mu^{\mathbb{P}}_i(a_l, \mathbf{x}_{-i})$ and $u_i(\mathbf{x}_i, \mathbf{x}_{-i}\mid \xi) = \sum_{l=1}^{n_i} x_i(a_l) u_i(a_l, \mathbf{x}_{-i}\mid \xi)$.

\begin{subequations} \label{eq:CVARLP}

\begin{align}
\max_{\mathbf{x}_i,z_i,\boldsymbol{\nu}_i} 
& \quad (1-\gamma_c) \sum_{l=1}^{n_i} x_i(a_l)
\mu^{\mathbb{P}}_i(a_l,\mathbf{x}_{-i})
+ \gamma_c \Big[
z_i + \frac{1}{\alpha}\sum_{k=1}^{K}\mathbb{P}(\xi_k)\nu_{i,k}
\Big] \qquad \qquad \qquad \qquad \qquad 
\label{cVarOpt:obj}
\end{align}
\begin{align}
\text{s.t.}\quad
& \nu_{i,k} \leq
\sum_{l=1}^{n_i} x_i(a_l)
u_i(a_l,\mathbf{x}_{-i}\mid \xi_k) - z_i,
\; \forall k \in \{1,\dots,K\}
\label{cVarOpt:cons1}
\\
& \nu_{i,k} \leq 0,
\; \forall k \in \{1,\dots,K\}
\label{cVarOpt:cons2} \\
& \vphantom{\sum_{l=1}^{n_i} x_i(a_l)
u_i(a_l,\mathbf{x}_{-i}\mid \xi_k)}
\mathbf{1}^{\top}\mathbf{x}_i = 1
\label{cVarOpt:cons3}
\\
& x_i(a_l) \geq 0,
\; \forall l \in \{1,\dots,n_i\}
\label{cVarOpt:cons4}
\end{align}
\end{subequations}
The optimality conditions for the LP $\eqref{eq:CVARLP}$ yield dual multipliers $\lambda_{i,k}, \mu_{i,k}, \alpha_i$ and $\beta_{i,l}$ for constraints $\eqref{cVarOpt:cons1}, \eqref{cVarOpt:cons2}, \eqref{cVarOpt:cons3}$, and $\eqref{cVarOpt:cons4}$, respectively.
We introduce the following shorthand notation for easier presentation:
$    v_{i,l}(\mathbf{x}_{-i})  = (1 - \gamma_c)\mu^{\mathbb{P}}_i(a_l, \mathbf{x}_{-i}) + \sum_{k=1}^K\lambda_{i,k} u_i(a_l, \mathbf{x}_{-i} \mid \xi_k)$.
Then, we obtain the following:
\begin{lemma} \label{lem:CVaRMLCP_main}
Consider a new variable $t_i$, the following definition of $\tau ,\Theta$ and substitutions for $\mathbf{x}_i, \mathbf{\nu}_i, z_i,\lambda_{i,k}$
\begin{align*}
    & \tau \coloneq \Big( \prod_{j = 1}^m \alpha_j \Big)^{\frac{1}{m-1}} \;;\;
    \Theta \coloneq \Big( \prod_{j = 1}^m t_j \Big)^{\frac{1}{m-1}}\;; \; 
     \tilde{\mathbf{x}}_i \coloneq \frac{\alpha_i t_i}{\tau\Theta}\mathbf{x_i} , \; \forall i;  \;\;
    \tilde{\mathbf{\nu}}_i \coloneq \frac{\mathbf{\nu}_i}{\tau\Theta} , \; \forall i  ;  \;\;
    \tilde{z}_i \coloneq \frac{z_i}{\tau\Theta}\;; \;\; \tilde{\lambda}_{i,k} \coloneq t_i \lambda_{i,k} \; \forall i,k.   
\end{align*}
Consider the MLCP formed by stacking the constraints below for all players
\begin{subequations} 
\begin{align}
& 0 \leq ((\gamma_c/\alpha)\mathbb{P}({\xi_k}) t_i - \tilde{\lambda}_{i,k}) \perp (-\tilde{\nu}_{i,k}) \geq 0 \quad \forall k \in \{1, \dots, K \} \\
& 0 \leq \tilde{\lambda}_{i,k} \perp \big(\sum_{l=1}^{n_i}\tilde{x}_i(a_l)u_i(a_l, \tilde{\mathbf{x}}_{-i} \mid \xi_k) - \tilde{z}_i - \tilde{\nu}_{i,k} \big) \geq 0 \quad \forall k \in \{1, \dots, K \} \label{mcp:actionValue}\\
& 0 \leq \big(1 - t_i (1 - \gamma_c)\mu^{\mathbb{P}}_i(a_l, \tilde{\mathbf{x}}_{-i}) - \sum_{k=1}^K\tilde{\lambda}_{i,k} u_i(a_l, \tilde{\mathbf{x}}_{-i} \mid \xi_k) \big) \perp \tilde{x}_i(a_l) \geq 0, \quad  \forall l \in \{1, \dots, n_i\}  \\
& 0 \leq t_i
    \;\perp\;
    (\gamma_c t_i - \sum_{k=1}^K \tilde{\lambda}_{i,k})
    \geq 0.
\end{align}
\end{subequations}
Any non-zero $\tilde{\mathbf{x}}_i$ ($\forall i$) MLCP solutions correspond to DREs of the CVaR game, 
with the mixed strategies and game values retrieved by
\[
    \mathbf{x}_i = \frac{\tilde{\mathbf{x}}_i}{\sum_{l=1}^{n_i} \tilde{x}_i(a_l)} \qquad
    \Theta = \Big( \prod_{j = 1}^m t_j \Big)^{\frac{1}{m-1}} \qquad
    \tau\Theta = \frac{1}{\prod_{i=1}^m\sum_{l=1}^{n_i} \tilde{x}_i(a_l)} \qquad
    \alpha_i = \frac{\tau\Theta}{t_i}\sum_{l=1}^{n_i} \tilde{x}_i(a_l)
\]
\end{lemma}
Even for two players, the MLCP remains bilinear and is therefore not a linear complementarity program of the Lemke-Howson type; the appendix establishes the same for MSD and MD games. Overall, the MLCPs can solve games with moderate action spaces and $K$. The next two subsections examine scalability.

\subsection{Small-Support Approximate DRE}
\label{subsec:scalable-action-cumg}

Our first result here establishes small supports for players' mixed strategies and constructs sparse empirical distributions over the $K$ samples that are used only to control deviations; these empirical distributions need not belong to the risk envelope $U$. A mixed strategy
$\bar{\mathbf{x}}_i$ is called $\kappa$-uniform if it is the uniform
distribution over a multiset of $\kappa$ pure actions. Equivalently,
$\bar{x}_i(a_l)$ is an integer multiple of $1/\kappa$ for every
$a_l\in\mathbf A_i$. Similarly, $\bar{\mathbf q}_i\in\Delta_K$ is
$\tau$-uniform if it is the uniform distribution over a multiset of
$\tau$ sample indices from $\{1,\ldots,K\}$.
Next, we have a corollary of a prior result~\cite{babichenko2014simple} that we utilize in our proof for small support.

\begin{proposition}
[Empirical-product concentration]
\label{lem:productEmpiricalConcentration}
Let $\mathbf A=\prod_{i=1}^m \mathbf A_i$, and let
$F:\mathbf A\rightarrow[0,1]$. With slight abuse of notation, let $F$
also denote its multilinear extension to mixed strategies:
$
F(\mathbf x)
=
\sum_{\mathbf a\in\mathbf A}
F(\mathbf a)\prod_{i=1}^m x_i(a_i).
$
For each player $i$, independently draw $\kappa$ pure-action samples
from $\mathbf x_i$ and let $\bar{\mathbf x}_i$ be their empirical
distribution. Then, for every $t>0$,
\[
\Pr\left(
\left|F(\bar{\mathbf x})-F(\mathbf x)\right|\geq t
\right)
\leq
\frac{4}{t}
\exp\left(-\frac{\kappa t^2}{2}\right).
\]
\end{proposition}

\begin{proof}
This follows directly from Lemma~3.4 of
\citet{babichenko2014simple}. To see the correspondence, introduce an
auxiliary player with a single fixed action and define its payoff under
the actions $\mathbf a=(a_1,\ldots,a_m)$ of the other players to be
$F(\mathbf a)$. Its expected payoff under the product distribution
$\mathbf x$ is therefore $F(\mathbf x)$, and its expected payoff under
the product of the empirical distributions is $F(\bar{\mathbf x})$.
Since $F(\mathbf a)\in[0,1]$, Lemma~3.4 of
\citet{babichenko2014simple} applies and gives the stated
bound.
\end{proof}

Next, the main small-support theorem below is proved utilizing the extended game characterization (Lemma~\ref{lem:lifted_game_dre}) with help from the above empirical-product concentration. 

\begin{theorem}[Small-support DRE with sparse adversarial witnesses]
\label{thm:small_support_lifted_drg}
Consider a data-driven CUMG in which, for every player $i$, $U\subseteq\Delta_K$ is nonempty, compact, and convex, and suppose that
$
0\le u_i(\mathbf a\mid \xi_k)\le 1
\;
\text{for every } i,\mathbf a,k.
$
Let $\mathbf{x}^*$ be a DRE. Fix $0<\epsilon\leq1$. If $\kappa,\tau\in\mathbb N$ satisfy
\[
\min\left\{\frac{\kappa}{4},\tau\right\}
>
\frac{18}{\epsilon^2}
\log\left(
\frac{24}{\epsilon}
\left(mK+2\sum_{i=1}^m n_i\right)
\right),
\]
then there exists a profile $\bar{\mathbf{x}}$ such that every $\bar{\mathbf{x}}_i$ is $\kappa$-uniform and $\bar{\mathbf{x}}$ is an $\epsilon$-DRE of the original CUMG. Moreover, the proof constructs, for each player $i$, a $\tau$-uniform vector $\bar{\mathbf q}_i\in\Delta_K$ that sparsely approximates an extended-game worst-case distribution $\mathbf q_i^*$; this $\bar{\mathbf q}_i$ need not belong to $U$.
\end{theorem}
\begin{proof}
By Lemma~\ref{lem:lifted_game_dre}, since $\mathbf{x}^*$ is a DRE, there exists $\mathbf q^*=(\mathbf q_1^*,\ldots,\mathbf q_m^*)$ such that $(\mathbf{x}^*,\mathbf q^*)$ is a Nash equilibrium of the extended game. Hence, for every player $i$, for all $\mathbf{x}_i\in\mathbf X_i$ and all $\mathbf q_i\in U$,
\[
P_i(\mathbf{x}_i,\mathbf{x}_{-i}^*,\mathbf q_i^*)
\le
P_i(\mathbf{x}_i^*,\mathbf{x}_{-i}^*,\mathbf q_i^*)
\le
P_i(\mathbf{x}_i^*,\mathbf{x}_{-i}^*,\mathbf q_i).
\]

For each player $i$, draw $\kappa$ independent pure-action samples from $\mathbf{x}_i^*$ and let $\bar{\mathbf{x}}_i$ be the empirical distribution of these samples. Independently, for each player $i$, draw $\tau$ independent samples from $q_i^*$ and let $\bar q_i$ be the empirical distribution of these sample indices. Then each $\bar{\mathbf{x}}_i$ is $\kappa$-uniform and each $\bar q_i$ is $\tau$-uniform.
For a fixed player $i$, define the three quantities:
\[
D_i :=
\max_{\mathbf y_i\in\mathbf X_i}
\left| P_i(\mathbf y_i,\bar{\mathbf{x}}_{-i},\bar{\mathbf q}_i) -
P_i(\mathbf y_i,\mathbf{x}_{-i}^*,\mathbf q_i^*) \right|,\;
D'_i :=
\max_{\mathbf q_i\in U}
\left| P_i(\mathbf{x}_i^*,\mathbf{x}_{-i}^*,\mathbf q_i) -
P_i(\bar{\mathbf{x}}_i,\bar{\mathbf{x}}_{-i},\mathbf q_i) \right|,
\]
\[
E_i :=
\max_{\mathbf y_i\in\mathbf X_i}
\left| P_i(\mathbf y_i,\bar{\mathbf{x}}_{-i},\bar{\mathbf q}_i) -
P_i(\mathbf y_i,\bar{\mathbf{x}}_{-i},\mathbf q_i^*) \right|.
\]

We first bound $D_i$. Since the expression inside $D_i$ is affine in $\mathbf y_i$, the maximum over $\mathbf X_i$ is attained at a pure action of player $i$. Fix such an action $a_l\in\mathbf A_i$. 
For this fixed action $a_l$, by the triangle inequality,
\begin{align*}
&
\left|
P_i(a_l,\bar{\mathbf{x}}_{-i},\bar{\mathbf q}_i)
-
P_i(a_l,\mathbf{x}_{-i}^*,\mathbf q_i^*)
\right|
\leq
\left|
P_i(a_l,\bar{\mathbf{x}}_{-i},\bar{\mathbf q}_i)
-
P_i(a_l,\mathbf{x}_{-i}^*,\bar{\mathbf q}_i)
\right|
+
\left|
P_i(a_l,\mathbf{x}_{-i}^*,\bar{\mathbf q}_i)
-
P_i(a_l,\mathbf{x}_{-i}^*,\mathbf q_i^*)
\right|.
\end{align*}

We first bound the action-sampling term. Conditional on
$\bar{\mathbf q}_i$, define
$
F(\mathbf a_{-i}) = \sum_{k=1}^K
\bar q_{i,k}
u_i(a_l,\mathbf a_{-i}\mid\xi_k).
$
Since $0\leq u_i(\mathbf a\mid\xi_k)\leq1$ and
$\bar{\mathbf q}_i\in\Delta_K$, we have
$F(\mathbf a_{-i})\in[0,1]$. Moreover,
\[
F(\bar{\mathbf x}_{-i})
=
P_i(a_l,\bar{\mathbf x}_{-i},\bar{\mathbf q}_i),
\qquad
F(\mathbf x_{-i}^*)
=
P_i(a_l,\mathbf x_{-i}^*,\bar{\mathbf q}_i).
\]
Thus, by Proposition~\ref{lem:productEmpiricalConcentration},
\[
\mathbb P\left(
\left|
P_i(a_l,\bar{\mathbf{x}}_{-i},\bar{\mathbf q}_i)
-
P_i(a_l,\mathbf{x}_{-i}^*,\bar{\mathbf q}_i)
\right|>\frac{t}{2}
\,\middle|\,
\bar{\mathbf q}_i
\right)
\le
\frac{8}{t}
\exp\left(-\frac{\kappa t^2}{8}\right).
\]
Since this bound does not depend on the realization of
$\bar{\mathbf q}_i$, the same bound holds unconditionally.

For the data-sample term, conditional on $\mathbf x_{-i}^*$,
$P_i(a_l,\mathbf{x}_{-i}^*,\bar{\mathbf q}_i)$ is the empirical
average of $\tau$ independent $[0,1]$-valued random variables with
mean
$P_i(a_l,\mathbf{x}_{-i}^*,\mathbf q_i^*)$.
Hence, Hoeffding's inequality gives
\[
\mathbb P\left(
\left|
P_i(a_l,\mathbf{x}_{-i}^*,\bar{\mathbf q}_i)
-
P_i(a_l,\mathbf{x}_{-i}^*,\mathbf q_i^*)
\right|>\frac{t}{2}
\right)
\le
2\exp\left(-\frac{\tau t^2}{2}\right).
\]
Therefore, by the triangle inequality and a union bound,
$
\mathbb P\left(
\left|
P_i(a_l,\bar{\mathbf{x}}_{-i},\bar{\mathbf q}_i)
-
P_i(a_l,\mathbf{x}_{-i}^*,\mathbf q_i^*)
\right|>t
\right)
\le
\frac{8}{t}
\exp\left(-\frac{\kappa t^2}{8}\right)
+
2\exp\left(-\frac{\tau t^2}{2}\right).
$
Finally, taking a union bound over the $n_i$ pure actions of player $i$,
\[
\mathbb P(D_i>t)
\le
n_i\left[
\frac{8}{t}
\exp\left(-\frac{\kappa t^2}{8}\right)
+
2\exp\left(-\frac{\tau t^2}{2}\right)
\right].
\]

Next, we bound $D'_i$. Since $U\subseteq\Delta_K$,
\[
D'_i
\le
\max_{k\in\{1,\ldots,K\}}
\left|
u_i(\mathbf{x}^*\mid \xi_k)
-
u_i(\bar{\mathbf{x}}\mid \xi_k)
\right|.
\]
For a fixed $k$, define
$
F(\mathbf a)=u_i(\mathbf a\mid\xi_k).
$
By assumption, $F(\mathbf a)\in[0,1]$ for every
$\mathbf a\in\mathbf A$. Its multilinear extension satisfies
$
F(\mathbf x^*)=u_i(\mathbf x^*\mid\xi_k),
\;
F(\bar{\mathbf x})=u_i(\bar{\mathbf x}\mid\xi_k).
$
Therefore, by Proposition~\ref{lem:productEmpiricalConcentration},
$
\mathbb P\left(
\left|
u_i(\mathbf{x}^*\mid\xi_k)
-
u_i(\bar{\mathbf{x}}\mid\xi_k)
\right|>t
\right)
\le
\frac{4}{t}
\exp\left(-\frac{\kappa t^2}{2}\right).
$
Taking a union bound over $k\in\{1,\ldots,K\}$ gives
\[
\mathbb P(D'_i>t)
\le
\frac{4K}{t}
\exp\left(-\frac{\kappa t^2}{2}\right).
\]

Finally, we bound $E_i$. Conditional on $\bar{\mathbf{x}}$, the random variable $P_i(a_l,\bar{\mathbf{x}}_{-i},\bar q_i)$ has expectation $P_i(a_l,\bar{\mathbf{x}}_{-i},q_i^*)$. Changing one of the $\tau$ sample-index draws used to form $\bar q_i$, say from $k$ to $k'$, changes $\bar q_i$ by $(e_{k'}-e_k)/\tau$ and hence changes $P_i(a_l,\bar{\mathbf{x}}_{-i},\bar q_i)$ by at most $1/\tau$, since the payoff averaged over the fixed profile $\bar{\mathbf{x}}_{-i}$ lies in $[0,1]$. Since the maximum over $\mathbf y_i\in\mathbf X_i$ is again attained at a pure action, for all $\mathbf x_{-i}\in\prod_{h\ne i}\mathbf X_h$, for every fixed realization $\mathbf x$ of $\bar{\mathbf x}$, McDiarmid's inequality and a union bound over the $n_i$ pure actions give
$
\mathbb P(E_i >t \mid \bar{\mathbf x}= \mathbf x)\le 2n_i\exp(-2\tau t^2).
$
Averaging this bound over the random draw of $\bar{\mathbf x}$ gives
$$
\mathbb P(E_i > t)\le 2n_i\exp(-2\tau t^2).
$$
Set $t=\epsilon/3$ and let $\Lambda_{\kappa,\tau}:=\min\{\kappa/4,\tau\}$. Then
\begin{align*}
&\mathbb P(D_i>\epsilon/3)\le n_i\left[\frac{24}{\epsilon}\exp\left(-\frac{\kappa\epsilon^2}{72}\right)+2\exp\left(-\frac{\tau\epsilon^2}{18}\right)\right],\quad
\mathbb P(D'_i>\epsilon/3)\le \frac{12K}{\epsilon}\exp\left(-\frac{\kappa\epsilon^2}{18}\right),\\
&\mathbb P(E_i>\epsilon/3)\le 2n_i\exp\left(-\frac{2\tau\epsilon^2}{9}\right).
\end{align*}
Since $\Lambda_{\kappa,\tau}\leq \kappa/4$ and $\Lambda_{\kappa,\tau}\leq\tau$, we have $\exp(-\kappa\epsilon^2/72)\leq\exp(-\epsilon^2\Lambda_{\kappa,\tau}/18)$ and $\exp(-\tau\epsilon^2/18)\leq\exp(-\epsilon^2\Lambda_{\kappa,\tau}/18)$. Also, $\exp(-\kappa\epsilon^2/18)\leq\exp(-\epsilon^2\Lambda_{\kappa,\tau}/18)$ and $\exp(-2\tau\epsilon^2/9)\leq\exp(-\epsilon^2\Lambda_{\kappa,\tau}/18)$. Therefore, taking a union bound over all players gives
\[
\mathbb P\left(\exists i:D_i>\epsilon/3\text{ or }D'_i>\epsilon/3\text{ or }E_i>\epsilon/3\right)
\le
\frac{24}{\epsilon}\left(mK+2\sum_{i=1}^m n_i\right)
\exp\left(-\frac{\epsilon^2\Lambda_{\kappa,\tau}}{18}\right).
\]
For above we use the fact that $0<\epsilon\leq1$, $\frac{24}{\epsilon}\sum_i n_i+\frac{12mK}{\epsilon}+4\sum_i n_i\leq\frac{24}{\epsilon}\left(mK+2\sum_i n_i\right)$.
Then, if
\[
\Lambda_{\kappa,\tau}
>
\frac{18}{\epsilon^2}
\log\left(\frac{24}{\epsilon}\left(mK+2\sum_{i=1}^m n_i\right)\right),
\]
the probability above is strictly smaller than one.
Therefore there exists a realization of the samples such that, for every player $i$,
\[
D_i\le \epsilon/3,\qquad D'_i\le \epsilon/3,\qquad E_i\le \epsilon/3.
\]
Fix such a realization.
We now show that the corresponding $\bar{\mathbf{x}}$ is an $\epsilon$-DRE. First, for each player $i$,
\[
\max_{\mathbf y_i\in\mathbf X_i}
P_i(\mathbf y_i,\bar{\mathbf{x}}_{-i},\bar{\mathbf q}_i)
\le
\max_{\mathbf y_i\in\mathbf X_i}
P_i(\mathbf y_i,\mathbf{x}_{-i}^*,\mathbf q_i^*) + D_i = P_i(\mathbf{x}_i^*,\mathbf{x}_{-i}^*,\mathbf q_i^*) + D_i.
\]
where the equality is by the extended-game equilibrium condition.
Also, for every $\mathbf q_i\in U$,
$
P_i(\mathbf{x}_i^*,\mathbf{x}_{-i}^*,\mathbf q_i)
\ge
P_i(\mathbf{x}_i^*,\mathbf{x}_{-i}^*,\mathbf q_i^*),
$
and by the definition of $D'_i$,
$
P_i(\bar{\mathbf{x}}_i,\bar{\mathbf{x}}_{-i},\mathbf q_i)
\ge
P_i(\mathbf{x}_i^*,\mathbf{x}_{-i}^*,\mathbf q_i)-D'_i.
$
Therefore,
\[
\min_{\mathbf q_i\in U}
P_i(\bar{\mathbf{x}}_i,\bar{\mathbf{x}}_{-i},\mathbf q_i)
\ge
P_i(\mathbf{x}_i^*,\mathbf{x}_{-i}^*,\mathbf q_i^*)-D'_i.
\]
Combining the last two displayed equations gives
\begin{equation} \label{eq:dd'}
\max_{\mathbf y_i\in\mathbf X_i}
P_i(\mathbf y_i,\bar{\mathbf{x}}_{-i},\bar{\mathbf q}_i)
-
\min_{\mathbf q_i\in U}
P_i(\bar{\mathbf{x}}_i,\bar{\mathbf{x}}_{-i},\mathbf q_i)
\le
D_i+D'_i.
\end{equation}
Now fix any $\mathbf y_i\in\mathbf X_i$. We have
$
\rho_i(\mathbf y_i,\bar{\mathbf{x}}_{-i}) = \min_{\mathbf q_i\in U}
P_i(\mathbf y_i,\bar{\mathbf{x}}_{-i},\mathbf q_i) \le P_i(\mathbf y_i,\bar{\mathbf{x}}_{-i},\mathbf q_i^*).
$
By the definition of $E_i$,
$
P_i(\mathbf y_i,\bar{\mathbf{x}}_{-i},\mathbf \mathbf \mathbf{q}_i^*)
\le
P_i(\mathbf y_i,\bar{\mathbf{x}}_{-i},\bar{\mathbf q}_i)+E_i.
$
Thus,
$
\rho_i(\mathbf y_i,\bar{\mathbf{x}}_{-i})
\le
\max_{\mathbf z_i\in\mathbf X_i}
P_i(\mathbf z_i,\bar{\mathbf{x}}_{-i},\bar{\mathbf q}_i)
+
E_i.
$
Using this and the bound~\eqref{eq:dd'} we get,
$
\rho_i(\mathbf y_i,\bar{\mathbf{x}}_{-i})
\le
\min_{\mathbf q_i\in U}
P_i(\bar{\mathbf{x}}_i,\bar{\mathbf{x}}_{-i},\mathbf q_i)
+
D_i+D'_i+E_i.
$
Since $
\min_{\mathbf q_i\in U}
P_i(\bar{\mathbf{x}}_i,\bar{\mathbf{x}}_{-i},\mathbf q_i)
=
\rho_i(\bar{\mathbf{x}}_i,\bar{\mathbf{x}}_{-i}),
$
we obtain
\[
\rho_i(\mathbf y_i,\bar{\mathbf{x}}_{-i})
\le
\rho_i(\bar{\mathbf{x}}_i,\bar{\mathbf{x}}_{-i})
+
D_i+D'_i+E_i
\le
\rho_i(\bar{\mathbf{x}}_i,\bar{\mathbf{x}}_{-i})
+
\epsilon.
\]
This holds for every $\mathbf y_i\in\mathbf X_i$ and every player $i$. Therefore $\bar{\mathbf{x}}$ is an $\epsilon$-DRE of the original CUMG.
\end{proof}

Theorem~\ref{thm:small_support_lifted_drg} builds on the classical small-support argument of~\citet{lipton2003playing} and the sharper concentration argument of~\citet{babichenko2014simple}. As in these works, we sample from an equilibrium and use concentration to show that sparse empirical strategies approximately preserve equilibrium incentives. The concentration result of~\citet{babichenko2014simple} avoids a dependence on the number of players in the action-sampling error. The key difference is that
in a standard finite game the payoff is multilinear and a best response can
be taken to be pure, so it suffices to preserve the payoffs of pure
deviations. In a CUMG, however,
$\rho_i(\mathbf x)=\min_{\mathbf q_i\in U}P_i(\mathbf x,\mathbf q_i)$
is generally nonlinear in the player's mixed strategy, and sampling
$\mathbf x^*$ alone does not provide the required control of full-game
regret. Our proof therefore first lifts a DRE to an equilibrium
$(\mathbf x^*,\mathbf q^*)$ of the extended game and samples both
$\mathbf x^*$ and the worst-case distributions $\mathbf q_i^*$.
The resulting $\tau$-uniform vectors $\bar{\mathbf q}_i$ are sparse
witnesses used to control deviations and need not themselves belong to $U$.

As in the small-support arguments of~\citet{lipton2003playing,babichenko2014simple}, Theorem~\ref{thm:small_support_lifted_drg} also yields a 
quasi-polynomial-time approximation scheme (QPTAS) when a best response and the
robust payoff at a given profile can be computed in polynomial time. This
is the case for the MSD, MD, and CVaR games considered here.

\begin{corollary}[QPTAS for MSD, MD, and CVaR games]
\label{cor:qptas}
For explicitly represented finite-sample MSD, MD, and CVaR games, an
$\epsilon$-DRE can be computed in quasi-polynomial time for every fixed
$\epsilon>0$.
\end{corollary}

\begin{proof} For $\epsilon\geq1$, any strategy profile is an $\epsilon$-DRE since payoffs lie in $[0,1]$, so assume $0<\epsilon<1$.
Set $\tau=\kappa$ in Theorem~\ref{thm:small_support_lifted_drg}. Then $\min\{\kappa/4,\tau\}=\kappa/4$, so the condition of the theorem is satisfied for $\kappa=O((\log(mK+\sum_i n_i)+\log(1/\epsilon))/\epsilon^2)$. Let $L$ denote the binary encoding size of the explicitly represented finite-sample game. Since the input explicitly contains the payoff tables for all $K$ samples, $\log(mK+\sum_i n_i)=O(\log L)$. Hence $\kappa=O((\log L+\log(1/\epsilon))/\epsilon^2)$.
There are at most $\prod_i n_i^\kappa=|\mathbf A|^\kappa$ profiles in which every player's strategy is $\kappa$-uniform. Since $|\mathbf A|\leq L$, all such profiles can therefore be enumerated in $L^{O((\log L+\log(1/\epsilon))/\epsilon^2)}$ time. For each candidate profile $\bar{\mathbf{x}}$, we compute $\rho_i(\bar{\mathbf{x}})$ and a best response $\max_{\mathbf y_i\in\mathbf X_i}\rho_i(\mathbf y_i,\bar{\mathbf{x}}_{-i})$ for every player $i$, and accept $\bar{\mathbf{x}}$ if every player's gain from deviation is at most $\epsilon$. These computations are polynomial-time for MSD, MD, and CVaR, as they reduce to linear programs. By Theorem~\ref{thm:small_support_lifted_drg}, there exists a $\kappa$-uniform profile satisfying this test. Thus, for fixed $\epsilon$, the enumeration runs in $L^{O(\log L)}$ time and returns an $\epsilon$-DRE.
\end{proof}

Despite this worst-case guarantee, the resulting QPTAS can be impractical
for data-driven games. In particular, evaluating each candidate requires robust-payoff and best-response computations whose size grows with the
number of payoff samples $K$, in addition to the cost of enumerating small-support profiles. This motivates two heuristic algorithms presented
next, which is designed to reduce the computational burden associated with
large $K$.
A first natural variant is to enumerate jointly the $\kappa$-uniform strategy
profiles $\bar{\mathbf{x}}$ and $\tau$-uniform vectors $\bar{\mathbf{q}}$ appearing in
Theorem~\ref{thm:small_support_lifted_drg}. For the explicitly represented
MSD, MD, and CVaR games considered here, this joint enumeration is still
quasi-polynomial for suitable $\kappa$ and $\tau$, and yields a QPTAS when
each candidate $\bar{\mathbf{x}}$ is verified by its full-game regret
$\eta_K(\bar{\mathbf{x}})=\max_i\{\max_{\mathbf y_i\in\mathbf X_i}
\rho_i(\mathbf y_i,\bar{\mathbf{x}}_{-i})-\rho_i(\bar{\mathbf{x}})\}\leq\epsilon$. However, the full game regret verification still uses all the $K$ data samples in evaluation $\rho_i$, and hence this does not provide any improvement in time over the QPTAS approach in Corollary~\ref{cor:qptas}.
Instead, for a candidate pair $(\bar{\mathbf{x}},\bar{\mathbf{q}})$ one can first cheaply test
$P_i(a_l,\bar{\mathbf{x}}_{-i},\bar{\mathbf{q}}_i)\leq
\rho_i(\bar{\mathbf{x}})+\epsilon_{\rm scr}$ for all $i,l$. 
Note that 
$P_i(\mathbf y_i,\bar{\mathbf{x}}_{-i},\bar{\mathbf{q}}_i)$ is linear in the deviating
strategy $\mathbf y_i$, so checking all pure actions yields $\max_{\mathbf y_i\in\mathbf X_i}
P_i(\mathbf y_i,\bar{\mathbf{x}}_{-i},\bar{\mathbf q}_i)
\le \rho_i(\bar{\mathbf{x}}) + 
\epsilon_{\rm scr}.$
This sparse-$\mathbf q$ test is not itself an approximate $\epsilon$-equilibrium guarantee,
because $\bar{\mathbf{q}}_i$ need not belong to $U$ and therefore does not represent
the true best response. The sparse-$\mathbf q$ test is used only for screening, and the
full-game regret check over all $K$ samples is performed only for candidates
that pass the screen. This is a valid screening step because, under the condition of Theorem~\ref{thm:small_support_lifted_drg}, the proof (specifically, Equation~\ref{eq:dd'}) guarantees that at least one enumerated candidate pair $(\bar{\mathbf{x}},\bar{\mathbf{q}})$ has small $\epsilon_{\rm scr}$ at most $2\epsilon/3$, and this same $\bar{\mathbf{x}}$ is an $\epsilon$-DRE. We call this \emph{screened-QPTAS heuristic}.

Yet another heuristic is to iterate over small supports and not just use uniform strategies, and optimize strategies over the small supports.
For a candidate action support profile $S=(S_1,\ldots,S_m)$ and data sample support profile $T=(T_1,\ldots,T_m)$, the screening searches over sparse $x_i\in\Delta(\mathbf A_i), \operatorname{supp}(\mathbf x_i)\subseteq S_i$ and $q_i\in\Delta_K,\operatorname{supp}(\mathbf q_i)\subseteq T_i$. Define the \emph{restricted-profile gap} by
\begin{align*}
\widehat\eta_{S,T}:=\min_{\mathbf x,\mathbf q,\eta};\eta\quad\text{s.t.}\quad &\mathbf x_i\in\Delta(\mathbf A_i), \operatorname{supp}(\mathbf x_i)\subseteq S_i, \mathbf q_i\in\Delta_K,\operatorname{supp}(\mathbf q_i)\subseteq T_i, \forall i \\
& P_i(a_l,\mathbf x_{-i},\mathbf q_i)\leq \rho_i(\mathbf x)+\eta,\ \forall i, \; \forall a_l \in \mathbf A_i.
\end{align*}
The restricted-profile gap asks whether the action supports $S$ and data-sample supports $T$ contain some sparse profile $\mathbf x$ and sparse $\mathbf q$ for which all pure deviations look nearly unprofitable relative to the exact robust value $\rho_i(\mathbf x)$. This is a valid screening step because, under the condition of Theorem~\ref{thm:small_support_lifted_drg}, the proof guarantees that at least one enumerated support pair $(S^\circ, T^\circ)$ has small restricted-profile gap of at most $2\epsilon/3$, and this same support $S^\circ$ also supports a profile $\mathbf x$ that is an $\epsilon$-DRE. 
Passing this screening is not an equilibrium guarantee, since the sparse $\mathbf q_i$ need not belong to $U$ and best responses under $\rho_i$ may be mixed (as demonstrated in Example~\ref{ex:coordinationGame}). Thus, screening only reduces the number of candidates: every successfully screened support is checked using the full-game regret; first, the full-game regret of a profile $\mathbf x$ is
$
\eta_K(\mathbf x):= \max_i \max_{\mathbf y_i\in\mathbf X_i}\rho_i(\mathbf y_i,\mathbf x_{-i})-\rho_i(\mathbf x).
$ Then,
$\eta_S^\star:=\min\{\eta_K(\mathbf x):x_i\in\Delta(\mathbf A_i), \operatorname{supp}(\mathbf x_i)\subseteq S_i,\ \forall i\}$
gives the minimum full-game regret attainable over support profile $S$. Unlike the screened-QPTAS heuristic, which evaluates full-game regret at a candidate $\kappa$-uniform profile $\bar{\mathbf x}$, here the probabilities within $S$ are optimized continuously. We call this the \emph{support-iteration heuristic} and present this in Algortihm~\ref{alg:sparse_witness_support_search}. The following is a correctness result for support-iteration heuristic (correctness of screened-QPTAS heuristic is a straightforward consequence of this result).

\begin{algorithm}[t]
\caption{Support-iteration heuristic: small-support search with data-sample support screening}
\label{alg:sparse_witness_support_search}
\DontPrintSemicolon
\KwIn{Tolerance $\epsilon$, screening tolerance $\epsilon_{\mathrm{scr}} = 2\epsilon/3$, support sizes $(\kappa,\tau)$, and rules for generating $S_i\subseteq\mathbf A_i$ and $T_i\subseteq \{1, \ldots, K\}$ with $|S_i|\leq\kappa$ and $|T_i|\leq\tau$.}
\For{each candidate support pair $(S,T)$}{
Solve the restricted-profile gap problem defining $\widehat\eta_{S,T}$.\;
\If{$\widehat\eta_{S,T}\leq\epsilon_{\mathrm{scr}}$}{
Solve full-game regret $\eta_S^\star:=\min{\eta_K(\mathbf x):x_i\in\Delta(\mathbf A_i), \operatorname{supp}(\mathbf x_i)\subseteq S_i,\ \forall i}$, to get $\mathbf x_S$.\;
\lIf{$\eta_S^\star\leq\epsilon$}{
\Return $\mathbf x_S$.
}
}
}
\end{algorithm}

\begin{lemma}[Correctness and completeness]
\label{lem:correctness}
Assume that Algorithm~\ref{alg:sparse_witness_support_search} enumerates all support profiles with $|S_i|\leq\kappa$ and $|T_i|\leq\tau$, and that all optimization problems are solved globally. If $(\kappa,\tau)$ satisfy the small-support bound of Theorem~\ref{thm:small_support_lifted_drg} and $\epsilon_{\mathrm{scr}}\geq 2\epsilon/3$, then the algorithm returns an $\epsilon$-DRE.
\end{lemma}
\begin{proof} 
First, we use the stronger conclusion established in the proof of Theorem~\ref{thm:small_support_lifted_drg}. On the good sampling event in that proof, there exist a $\kappa$-supported profile $\bar{\mathbf x}$ and $\tau$-supported sparse witnesses $\bar{\mathbf q}$ with supports $S^\circ_i=\mathrm{supp}(\bar{\mathbf x}_i)$ and $T^\circ_i=\mathrm{supp}(\bar{\mathbf q}_i)$ such that, for every player $i$, the proof errors satisfy $D_i\leq\epsilon/3$, $D_i'\leq\epsilon/3$, and $E_i\leq\epsilon/3$. The first two errors control the restricted-profile gap: for every $i$, $\max_{a_l\in\mathbf A_i}P_i(a_l,\bar{\mathbf x}_{-i},\bar{\mathbf q}_i)-\rho_i(\bar{\mathbf x})\leq D_i+D_i'\leq 2\epsilon/3$, and hence $\widehat\eta(\bar{\mathbf x},\bar{\mathbf q})\leq 2\epsilon/3$, where $\widehat\eta(\mathbf x,\mathbf q):=\max_i{\max_{a_l\in\mathbf A_i}P_i(a_l,\mathbf x_{-i},\mathbf q_i)-\rho_i(\mathbf x)}$. 
The role of $E_i$ is to compare deviations evaluated against the sparse witness $\bar{\mathbf q}_i$ with deviations evaluated against the true adversarial distribution $\mathbf q_i^*$ from the proof of Theorem~\ref{thm:small_support_lifted_drg}. For any mixed strategy deviation $\mathbf y_i\in\mathbf X_i$, 
$$\rho_i(\mathbf y_i,\bar{\mathbf x}_{-i})=\min_{\mathbf q_i\in U}P_i(\mathbf y_i,\bar{\mathbf x}_{-i},\mathbf q_i)\leq P_i(\mathbf y_i,\bar{\mathbf x}_{-i},\mathbf q_i^*)\leq P_i(\mathbf y_i,\bar{\mathbf x}_{-i},\bar{\mathbf q}_i)+E_i.$$ 
Since $P_i(\mathbf y_i,\bar{\mathbf x}_{-i},\bar{\mathbf q}_i)$ is linear in $\mathbf y_i$, it is at most $\max_{a_l\in\mathbf A_i}P_i(a_l,\bar{\mathbf x}_{-i},\bar{\mathbf q}_i)$. Combining this with the restricted-profile gap bound gives 
$$\rho_i(\mathbf y_i,\bar{\mathbf x}_{-i})\leq \rho_i(\bar{\mathbf x})+2\epsilon/3+\epsilon/3=\rho_i(\bar{\mathbf x})+\epsilon.$$
Thus, $\eta_K(\bar{\mathbf x})\leq\epsilon$. Thus, the theorem proof supplies not only an $\epsilon$-DRE candidate, but also a support pair $(S^\circ,T^\circ)$ that passes the restricted-profile screen.

Since the enumeration is exhaustive, the algorithm considers $(S^\circ,T^\circ)$. Since $\bar{\mathbf x}_i\in\Delta(S^\circ_i)$ and $\bar{\mathbf q}_i\in\Delta(T^\circ_i)$, the pair $(\bar{\mathbf x},\bar{\mathbf q})$ is feasible for the screening problem, and therefore $\widehat\eta_{S^\circ,T^\circ}\leq \widehat\eta(\bar{\mathbf x},\bar{\mathbf q})\leq 2\epsilon/3\leq\epsilon_{\mathrm{scr}}$. Thus, this support pair passes the screening. The same $\bar{\mathbf x}$ is feasible for the final $S^\circ$-supported full-game regret minimization, so $\eta_{S^\circ}^\star\leq\eta_K(\bar{\mathbf x})\leq\epsilon$. Hence, when this support pair is reached, the final check succeeds; if the algorithm returns earlier, soundness already guarantees that the returned profile is an $\epsilon$-DRE.
\end{proof}
As with the complementarity formulations, the fixed-support restricted-profile gap problems for MSD and CVaR games are nonconvex multilinear programs; see Appendix~\ref{subsec:restrictedgap}. Nonlinearity enters through multilinear payoff terms $u_i(\mathbf x\mid \xi_k)$ and $u_i(a_l,\mathbf x_{-i}\mid \xi_k)$, as well as the sparse deviation payoff $P_i$ coupling $\mathbf q_i$ and $\mathbf x_{-i}$.

\subsection{First-order Methods for Approximate DRE for Large K}
\label{subsec:scalable-samples-cumg}
First-order methods often provide a more scalable practical route for equilibrium computation, if applicable, compared to approaches such as support enumeration.
We use a stochastic first-order error minimization
approach, closest in spirit
to~\citet{gemp2024approximating} for standard normal form games, but the bottlenecks are different in our setting. In a standard normal-form game, fixing the opponents' strategies makes a
player's payoff linear in its own mixed strategy, so the gradient with
respect to that strategy is simply the vector of pure-action payoffs.
\citet{gemp2024approximating} add entropy regularization to ensure that the
equilibrium of the regularized game is interior. At such an equilibrium,
the first-order optimality condition requires the regularized payoff
gradient to be constant across actions. This motivates a loss given by the
squared norm of the gradient after removing its common additive component:
the loss is zero exactly when these first-order equilibrium conditions hold.
A key computational advantage is that the required payoff-gradient
quantities, and hence the gradient of this loss, admit unbiased Monte Carlo
estimates from sampled action profiles, avoiding enumeration of the full
payoff tensor. The resulting optimization problem is nonconvex, so stochastic
gradient methods are guaranteed to approach stationary points of the loss,
rather than necessarily a zero-loss that yields the equilibrium.


We adapt this idea to CUMGs, where the risk-adjusted payoff is generally
nonlinear and nonsmooth in a player's own mixed strategy, so the corresponding
first-order error requires a different construction. More importantly, the
action-profile sampling used in standard games does not carry over directly:
even after smoothing the nonsmooth risk terms, it generally produces biased
gradient estimates. We first establish this limitation. We then develop
unbiased stochastic gradients by subsampling the payoff data; this addresses
scalability in the number of data samples $K$, but not the action-space dependence
handled by the standard action-sampling approach.

\paragraph{Action sampling results in a biased estimator for gradients} For CUMGs, we start with a negative result showing that the action sampling argument of~\citet{gemp2024approximating}
fails. The gradient of risk-adjusted payoff is not simply $(\rho_i(a_l,\mathbf x_{-i}))_l$, and kinks create
non-differentiability. To simplify the presentation and not deal with Clarke's generalized Jacobians, we first smooth the kinks, for which we will later introduce approximation bounds. Towards this, we define an augmented representation. For several coherent utility measures, the utility admits an augmented
representation. 
We use an auxiliary variable
$\boldsymbol{\theta}_i\in\Theta_i$, for $\Theta_i$ convex and possibly unbounded, and suppose
$
\rho_i(\mathbf{x})
=
\sup_{\boldsymbol{\theta}_i\in\Theta_i}
\rho_i^{\mathrm{aux}}(\mathbf{x},\boldsymbol{\theta}_i),
\quad
\rho_i^{\mathrm{aux}}(\mathbf{x},\boldsymbol{\theta}_i) = b_i(\mathbf{x},\boldsymbol{\theta}_i) +
\sum_{k=1}^K \mathbb P(\xi_k)
\psi_{i,k}(\mathbf{x},\boldsymbol{\theta}_i)
$; the regret certificate below uses a compact convex set $\bar\Theta_i\subseteq\Theta_i$ on which the same supremum is attained.
Here $b_i$ is a sample-summary term whose evaluation does not scale with $K$ (after a possible one time pre-computation), while the sample average contains the $K$ sample-indexed non-differentiable terms $\psi_{i,k}$. For MSD, $\boldsymbol{\theta}_i$ is empty; for CVaR, $\boldsymbol{\theta}_i$ contains the VaR variable $z_i$. 
To handle nonsmooth risk measures such as MSD and CVaR, let
$\psi_{i,k}^{\tau}$ be a smooth approximation of $\psi_{i,k}$ dependent on a parameter $\tau > 0$ and define
$
\rho_{i}^{\mathrm{aux},\tau}(\mathbf{x},\boldsymbol{\theta}_i) = b_i(\mathbf{x},\boldsymbol{\theta}_i) + \sum_{k=1}^K \mathbb P(\xi_k) \psi_{i,k}^{\tau}(\mathbf{x},\boldsymbol{\theta}_i).
$

For the smoothing approximation assume that 
$
\sup_{\mathbf{x},\boldsymbol{\theta}_i}
\left|
\rho_i^{\mathrm{aux}}(\mathbf{x},\boldsymbol{\theta}_i)
-
\rho_i^{\mathrm{aux},\tau}(\mathbf{x},\boldsymbol{\theta}_i)
\right|
\leq \delta_i(\tau),
$ for some $\delta_i(\tau)$ with $\lim_{\tau\to 0}\delta_i(\tau)=0$.
The nonsmooth sample terms in MSD and CVaR are compositions with the positive-part kink: $\varphi(a)=\max(0,a)$. Let $\varphi_\tau$ denote a differentiable smoothing of $\varphi$. Since $\varphi_\tau$ is nonlinear, smoothing complicates gradient estimation. For MSD, the smoothed gradient contains the weight $\varphi_\tau'(\mu_i^{\mathbb P}(\mathbf x)-u_i(\mathbf x\mid \xi_k))$, which depends on the full mixed-profile payoff $u_i(\mathbf x\mid \xi_k)$. If this payoff is replaced by an action-sampling estimate $\widehat u_i(\mathbf x\mid \xi_k)$, the nonlinear weight is generally biased:
$
\mathbb E\!\left[
\varphi_\tau'\!\left(\mu_i^{\mathbb P}(\mathbf x)-\widehat u_i(\mathbf x\mid \xi_k)\right)
\right]
\neq
\varphi_\tau'\!\left(\mu_i^{\mathbb P}(\mathbf x)-u_i(\mathbf x\mid \xi_k)\right).
$
Thus, action sampling does not yield a clean unbiased estimator of the MSD payoff-gradient vector. An analogous issue arises for CVaR.
Thus, we \emph{cannot use this idea to scale in the number of actions}.

\paragraph{Stochastic gradient with subsamples of data samples} Nonetheless, we adapt the loss-based idea to address the large-$K$ problem. 
For $\kappa>0$, define the entropy-regularized augmented utility
\begin{align*}
\rho_{i}^{\kappa, \tau}(\mathbf{x},\boldsymbol{\theta}_i) =
\rho_{i}^{\mathrm{aux},\tau}(\mathbf{x},\boldsymbol{\theta}_i) +
\kappa H_i(\mathbf{x}_i),
\qquad
H_i(\mathbf{x}_i) =
-\sum_{l=1}^{n_i}x_i(a_l)\log x_i(a_l),\qquad 
\text{ and gradients}\\
\mathbf{g}_{i}^{\tau} =
\nabla_{\mathbf{x}_i}
\rho_{i}^{\mathrm{aux},\tau}(\mathbf{x},\boldsymbol{\theta}_i),
\qquad
\mathbf{h}_{i}^{\tau} = \nabla_{\boldsymbol{\theta}_i}
\rho_{i}^{\mathrm{aux},\tau}(\mathbf{x},\boldsymbol{\theta}_i), \qquad \mathbf q_i^\tau =
\mathbf g_i^\tau - \kappa(\log \mathbf x_i+\mathbf 1), \qquad
\end{align*}
where $\mathbf h_i^\tau$ is omitted when $\Theta_i$ is absent.
At an exact entropy-regularized best response $\mathbf{x}_i^\star$, the strategy-part first-order optimality condition is
$\langle \mathbf q_i^\star,\mathbf y_i-\mathbf x_i^\star\rangle\leq 0$ for every $\mathbf y_i\in\mathbf X_i$. 
In a standard normal-form game, the entropy-regularized best response has an explicit logit form because the gradient of payoff w.r.t. $\mathbf x_i$ is fixed once $\mathbf{x}_{-i}$ is fixed. In a CUMG, the analogous condition is generally implicit: although the KKT condition can be written as $x_i(a_l)\propto \exp(g_{i,l}^\tau (\mathbf{x},\boldsymbol{\theta}_i)/\kappa)$, the gradient vector $\mathbf g_i^\tau$ may itself depend on $\mathbf x_i$. Thus, the entropy term forces the equilibrium to lie in the relative interior of the probability simplex, but does not provide a closed-form logit response.
Proposition~\ref{prop:entropy-interiority} in the appendix shows this formally. 

Hence, since $\mathbf{x}_i^\star$ has full support, applying the first-order condition to both $\mathbf{x}_i^\star+\epsilon\mathbf d$ and $\mathbf{x}_i^\star-\epsilon\mathbf d$ for every direction $\mathbf d$ with $\mathbf 1^\top\mathbf d=0$ gives $\langle\mathbf q_i^\star,\mathbf d\rangle=0$. Hence $\mathbf q_i^\star$ is orthogonal to the tangent space of the simplex, so $\mathbf q_i^\star=\alpha_i\mathbf 1$ for some scalar $\alpha_i$. For a candidate $\mathbf x_i$, first-order comparisons only depend on $\mathbf q_i^\tau$ modulo additive constants: for any feasible $\mathbf x_i,\mathbf y_i\in\mathbf X_i$ and any scalar $c$, $\langle \mathbf q_i^\tau+c\mathbf 1,\mathbf y_i-\mathbf x_i\rangle=\langle \mathbf q_i^\tau,\mathbf y_i-\mathbf x_i\rangle$, since $\mathbf 1^\top(\mathbf y_i-\mathbf x_i)=0$. So, we use the centered representative $\mathbf r_i^\tau=\mathbf q_i^\tau-\langle\mathbf q_i^\tau,\mathbf x_i\rangle\mathbf 1$. This centering preserves all feasible first-order comparisons and vanishes whenever the interior KKT condition $\mathbf q_i^\tau=\alpha_i\mathbf 1$ holds.

Let $\mathbf x_i=\operatorname{softmax}(\mathbf w_i)$ and $\widetilde{\mathbf z}=(\mathbf w,\boldsymbol{\theta})$ be the corresponding parametrization of $\mathbf z=(\mathbf x,\boldsymbol{\theta})$. The full smoothed first-order error vector is
$
\mathcal R^{\kappa,\tau}(\mathbf z) =
(\mathbf r_1^\tau,\ldots,\mathbf r_m^\tau,
 \mathbf h_1^\tau,\ldots,\mathbf h_m^\tau).
$
We minimize the composed squared norm
$
M^{\kappa,\tau}(\widetilde{\mathbf z}) = \frac12 \left\| \mathcal R^{\kappa,\tau}(\mathbf z)
\right\|_2^2 .
$
A small value of $M^{\kappa,\tau}(\widetilde{\mathbf z})$ ensures that the first-order errors are small.
Let
$
J^{\kappa,\tau}(\widetilde{\mathbf z}) =
D_{\widetilde{\mathbf z}}\mathcal R^{\kappa,\tau}(\mathbf z)
$
denote the Jacobian, with differentiation taken through $\operatorname{softmax}$.
Then
$
\nabla M^{\kappa,\tau}(\widetilde{\mathbf z}) = J^{\kappa,\tau}(\widetilde{\mathbf z})^\top
\mathcal R^{\kappa,\tau}(\mathbf z)
$.
Note that for computation, the Jacobian is not formed explicitly; only the vector-Jacobian product is
needed.
For a mini-batch $B$ obtained by independently sampling indices
$k\in[K]$ according to $\mathbb P(\xi_k)$, define
\[
\textstyle \widehat{\mathbf g}_{i,B}^{\tau} =
\nabla_{\mathbf{x}_i}b_i(\mathbf{x},\boldsymbol{\theta}_i) + \frac1{|B|}
\sum_{k\in B}
\nabla_{\mathbf{x}_i}
\psi_{i,k}^{\tau}(\mathbf{x},\boldsymbol{\theta}_i),
\qquad 
\widehat{\mathbf h}_{i,B}^{\tau} =
\nabla_{\boldsymbol{\theta}_i}b_i(\mathbf{x},\boldsymbol{\theta}_i) +
\frac1{|B|}
\sum_{k\in B}
\nabla_{\boldsymbol{\theta}_i}
\psi_{i,k}^{\tau}(\mathbf{x},\boldsymbol{\theta}_i).
\]
From these define
$
\widehat{\mathbf q}_{i,B}^{\tau} = \widehat{\mathbf g}_{i,B}^{\tau} -
\kappa(\log\mathbf x_i+\mathbf 1),
\;
\widehat{\mathbf r}_{i,B}^{\tau} =
\widehat{\mathbf q}_{i,B}^{\tau} -
\langle
\widehat{\mathbf q}_{i,B}^{\tau},\mathbf x_i
\rangle\mathbf 1,
$
and stack them to obtain
$\widehat{\mathcal R}^{\kappa,\tau}_{B}(\mathbf z)$.
Indeed, since the indices in $B$ are sampled i.i.d. according to
$\mathbb P(\xi_k)$, the mini-batch estimators
$\widehat{\mathbf g}_{i,B}^{\tau}$ and
$\widehat{\mathbf h}_{i,B}^{\tau}$ are unbiased estimators of the
corresponding $\mathbb P$-weighted full-sample gradients. Thus, conditional
on $\widetilde{\mathbf z}$,
$
\mathbb E_B[\widehat{\mathbf g}_{i,B}^{\tau}(\mathbf z)
\mid \widetilde{\mathbf z}]
=
\mathbf g_i^\tau(\mathbf z),
\;
\mathbb E_B[\widehat{\mathbf h}_{i,B}^{\tau}(\mathbf z)
\mid \widetilde{\mathbf z}]
=
\mathbf h_i^\tau(\mathbf z).
$
The centering operation is linear in
$\widehat{\mathbf q}_{i,B}^{\tau}$ for fixed $\mathbf x_i$, and therefore
$
\mathbb E_B[\widehat{\mathbf r}_{i,B}^{\tau}(\mathbf z)
\mid \widetilde{\mathbf z}]
=
\mathbf r_i^\tau(\mathbf z).
$
Stacking over players gives
$
\mathbb E_B[\widehat{\mathcal R}^{\kappa,\tau}_{B}(\mathbf z)
\mid \widetilde{\mathbf z}]
=
\mathcal R^{\kappa,\tau}(\mathbf z).
$
Moreover, because $\psi_{i,k}^{\tau}$ is smooth and the expectation is a
finite $\mathbb P$-weighted sum, differentiation commutes with expectation.
Hence,
$
\mathbb E_B[
D_{\widetilde{\mathbf z}}
\widehat{\mathcal R}^{\kappa,\tau}_{B}(\mathbf z)
\mid \widetilde{\mathbf z}]
=
D_{\widetilde{\mathbf z}}
\mathcal R^{\kappa,\tau}(\mathbf z)
=
J^{\kappa,\tau}(\widetilde{\mathbf z}).
$
Thus, both the residual vector and its Jacobian are unbiased mini-batch
estimators of their full-sample counterparts.

Let $B^{(1)}$ and $B^{(2)}$ be two independent mini-batches, conditional on
$\widetilde{\mathbf z}$. We use the stochastic gradient estimator
$
\widehat G^{\kappa,\tau}(\widetilde{\mathbf z})
=
D_{\widetilde{\mathbf z}}
\widehat{\mathcal R}^{\kappa,\tau}_{B^{(1)}}(\mathbf z)^\top
\widehat{\mathcal R}^{\kappa,\tau}_{B^{(2)}}(\mathbf z).
$
By conditional independence of the two mini-batches,
$
\mathbb E[
\widehat G^{\kappa,\tau}(\widetilde{\mathbf z})
\mid \widetilde{\mathbf z}]
=
\mathbb E[
D_{\widetilde{\mathbf z}}
\widehat{\mathcal R}^{\kappa,\tau}_{B^{(1)}}(\mathbf z)
\mid \widetilde{\mathbf z}]^\top
\mathbb E[
\widehat{\mathcal R}^{\kappa,\tau}_{B^{(2)}}(\mathbf z)
\mid \widetilde{\mathbf z}]
=
J^{\kappa,\tau}(\widetilde{\mathbf z})^\top
\mathcal R^{\kappa,\tau}(\mathbf z)
=
\nabla M^{\kappa,\tau}(\widetilde{\mathbf z}).
$
Therefore,
$\widehat G^{\kappa,\tau}$
is an unbiased stochastic gradient estimator of the smoothed squared error
$M^{\kappa,\tau}$. This yields the stochastic gradient method in
Algorithm~\ref{alg:gempStyleCUMG}, together with the regret bound below.

\begin{algorithm}[t]
\caption{Stochastic first-order error minimization for smoothed augmented CUMGs}
\label{alg:gempStyleCUMG}
\DontPrintSemicolon
Choose $\kappa>0$, smoothing level $\tau>0$, feasible set $\widetilde{\mathcal Z}$ for $(\mathbf w,\boldsymbol{\theta})$, and stepsizes $\{\eta_t\}_{t\ge0}$\;
Initialize $\widetilde{\mathbf z}_0=(\mathbf w_0,\boldsymbol{\theta}_0)\in\widetilde{\mathcal Z}$ and set $\mathbf x_{i,0}=\operatorname{softmax}(\mathbf w_{i,0})$ for every $i$ to get $\mathbf z_0 = (\mathbf x_0, \boldsymbol{\theta}_0)$\;
\For{$t=0,1,2,\ldots$ till convergence or iteration upper bound}{
    Draw independent mini-batches $B_t^{(1)},B_t^{(2)}$,
set $\mathbf x_{i,t}=\operatorname{softmax}(\mathbf w_{i,t})$
 $\forall i$, and set
$\mathbf z_t=(\mathbf x_t,\boldsymbol{\theta}_t)$\;
    Compute $\widehat{\mathcal R}^{\kappa,\tau}_{B_t^{(2)}}(\mathbf z_t)$ and 
    $
    \widehat G^{\kappa,\tau}_t \!=\!
    D_{\widetilde{\mathbf z}}\widehat{\mathcal R}^{\kappa,\tau}_{B_t^{(1)}}(\mathbf z_t)^\top
    \widehat{\mathcal R}^{\kappa,\tau}_{B_t^{(2)}}(\mathbf z_t)
    $
    by auto differentiation\;
    Update
    $
    \widetilde{\mathbf z}_{t+1}
    =
    \Pi_{\widetilde{\mathcal Z}}
    \left[
    \widetilde{\mathbf z}_t-\eta_t\widehat G^{\kappa,\tau}_t
    \right].
    $
}
\end{algorithm}

\begin{lemma}[Regret bound]
\label{lem:smoothedResidualRegret}
Fix $(\mathbf x,\boldsymbol\theta)$ and
let $\bar\Theta_i\subseteq\Theta_i$ be a nonempty convex compact set such that
$\rho_i(\mathbf x)=\sup_{\theta_i\in\bar\Theta_i}\rho_i^{\mathrm{aux}}(\mathbf x,\theta_i)$  $\forall \mathbf x\in\mathbf X$.
Let $\Delta_i^{\mathrm{aux},\tau}(\mathbf x,\boldsymbol{\theta}_i)=\sup_{\mathbf y_i\in\mathbf X_i,\boldsymbol{\vartheta}_i\in \bar\Theta_i}\left[\rho_i^{\mathrm{aux},\tau}(\mathbf y_i,\mathbf x_{-i},\boldsymbol{\vartheta}_i)-\rho_i^{\mathrm{aux},\tau}(\mathbf x,\boldsymbol{\theta}_i)\right]$. Then, regret of player $i$ in the original CUMG satisfies
$$
\sup_{\mathbf y_i\in\mathbf X_i}\left[\rho_i(\mathbf y_i,\mathbf x_{-i})-\rho_i(\mathbf x)\right]\le \Delta_i^{\mathrm{aux},\tau}(\mathbf x,\boldsymbol{\theta}_i)+2\delta_i(\tau).
$$
Moreover, if, for fixed $\mathbf x_{-i}$, $\rho_i^{\mathrm{aux},\tau}(\mathbf x_i,\mathbf x_{-i},\boldsymbol{\theta}_i)$ is jointly concave in $(\mathbf x_i,\boldsymbol{\theta}_i)$, then
\[
\Delta_i^{\mathrm{aux},\tau}(\mathbf x,\boldsymbol{\theta}_i)\le \sqrt{2}\|\mathbf r_i^\tau\|_2+\sup_{\boldsymbol{\vartheta}_i\in \bar\Theta_i}\langle \mathbf h_i^\tau,\boldsymbol{\vartheta}_i-\boldsymbol{\theta}_i\rangle+\kappa\log n_i.
\]
Consequently, if the RHS above plus $2\delta_i(\tau)$ is $\leq \epsilon$ for all players, then $\mathbf x$ is an $\epsilon$-DRE of the original CUMG.
\end{lemma}
\begin{proof}

Let
$\rho_i^\tau(\mathbf x)=\sup_{\boldsymbol{\vartheta}_i\in\bar\Theta_i}\rho_i^{\mathrm{aux},\tau}(\mathbf x,\boldsymbol{\vartheta}_i)$.
The uniform smoothing bound and the fact that the same supremum for the nonsmoothed utility is attained over $\bar\Theta_i$ imply
$|\rho_i^\tau(\mathbf x)-\rho_i(\mathbf x)|\leq\delta_i(\tau)$.
Moreover, since
$\rho_i(\mathbf x)=\sup_{\boldsymbol{\vartheta}_i\in\Theta_i}\rho_i^{\mathrm{aux}}(\mathbf x,\boldsymbol{\vartheta}_i)$,
we have
$\rho_i(\mathbf x)\geq\rho_i^{\mathrm{aux}}(\mathbf x,\boldsymbol{\theta}_i)$
for every $\boldsymbol{\theta}_i\in\Theta_i$. Therefore, for any $\mathbf y_i\in\mathbf X_i$,
\[
\begin{aligned}
\rho_i(\mathbf y_i,\mathbf x_{-i})-\rho_i(\mathbf x)
&\leq
\rho_i^\tau(\mathbf y_i,\mathbf x_{-i})
+\delta_i(\tau)-\rho_i(\mathbf x)\\
&\leq
\rho_i^\tau(\mathbf y_i,\mathbf x_{-i})
+\delta_i(\tau)
-\rho_i^{\mathrm{aux}}(\mathbf x,\boldsymbol{\theta}_i)\\
&\leq
\rho_i^\tau(\mathbf y_i,\mathbf x_{-i})
-\rho_i^{\mathrm{aux},\tau}(\mathbf x,\boldsymbol{\theta}_i)
+2\delta_i(\tau) =
\sup_{\boldsymbol{\vartheta}_i\in\bar\Theta_i}
\rho_i^{\mathrm{aux},\tau}
(\mathbf y_i,\mathbf x_{-i},\boldsymbol{\vartheta}_i)
-\rho_i^{\mathrm{aux},\tau}
(\mathbf x,\boldsymbol{\theta}_i)
+2\delta_i(\tau).
\end{aligned}
\]
The first inequality uses the smoothing bound for the deviating profile, the second uses that $\rho_i(\mathbf x)$ is the supremum of the nonsmoothed auxiliary objective over $\Theta_i$, and the third uses the smoothing bound at $(\mathbf x,\boldsymbol{\theta}_i)$. Taking the supremum over $\mathbf y_i\in\mathbf X_i$ gives the first claim.
For the second claim, define $\rho_i^{\kappa,\tau}(\mathbf x,\boldsymbol{\theta}_i)=\rho_i^{\mathrm{aux},\tau}(\mathbf x,\boldsymbol{\theta}_i)+\kappa H_i(\mathbf x_i)$. By joint concavity, for any $\mathbf y_i$ and $\boldsymbol{\vartheta}_i$,
\[
\rho_i^{\mathrm{aux},\tau}(\mathbf y_i,\mathbf x_{-i},\boldsymbol{\vartheta}_i)-\rho_i^{\mathrm{aux},\tau}(\mathbf x,\boldsymbol{\theta}_i)\le \langle \mathbf q_i^\tau,\mathbf y_i-\mathbf x_i\rangle+\langle \mathbf h_i^\tau,\boldsymbol{\vartheta}_i-\boldsymbol{\theta}_i\rangle+\kappa\left(H_i(\mathbf x_i)-H_i(\mathbf y_i)\right).
\]
Since $\mathbf y_i-\mathbf x_i$ has zero sum, $\langle \mathbf q_i^\tau,\mathbf y_i-\mathbf x_i\rangle=\langle \mathbf r_i^\tau,\mathbf y_i-\mathbf x_i\rangle\le\sqrt{2}\|\mathbf r_i^\tau\|_2$, using the Euclidean diameter of the simplex. Also, $H_i(\mathbf x_i)-H_i(\mathbf y_i)\le\log n_i$. Taking the supremum over $\mathbf y_i$ and $\boldsymbol{\vartheta}_i \in \bar\Theta_i$ gives the second claim. Combining the two bounds proves the final statement.
\end{proof}

\paragraph{Instantiations for MSD and CVaR}
For both MSD and CVaR the non-smoothness is due to the $(a)_+ = \max(0,a) = - \min(0,-a)$ function. Let
$\varphi_\tau(a)=\tau\log(1+\exp(a/\tau))$. Then
$\varphi_\tau'(a)=1/(1+\exp(-a/\tau))$,
$0\le\varphi_\tau'(a)\le1$, and
$0\le\varphi_\tau(a)-(a)_+\le c_\varphi\tau$
with $c_\varphi=\log2$. With this, and $\Theta_i$ empty for MSD and $\Theta_i = \mathbb{R}$ with $\bar \Theta_i = [0,1]$ for CVaR (under $u_i(\mathbf x\mid\xi_k)\in[0,1]$), we get
\begin{align*}
\sup_{\mathbf y_i\in\mathbf X_i}\left[\rho_{i,\mathrm{MSD}}(\mathbf y_i,\mathbf x_{-i})-\rho_{i,\mathrm{MSD}}(\mathbf x)\right] & \leq \sqrt{2}\|\mathbf r_i^\tau\|_2+\kappa\log n_i+2\gamma_s c_\varphi\tau.\\
\sup_{\mathbf y_i\in\mathbf X_i}\left[\rho_{i,\mathrm{CVaR}}(\mathbf y_i,\mathbf x_{-i})-\rho_{i,\mathrm{CVaR}}(\mathbf x)\right] & \leq \sqrt{2}\|\mathbf r_i^\tau\|_2
+
\left(
1+\tau
\left|
\log\left(\frac{\alpha}{1-\alpha}\right)
\right|
\right)
|h_i^\tau| +
\kappa\log n_i +
\frac{2\gamma_c c_\varphi\tau}{\alpha}.
\end{align*}
Detailed step-by-step derivation is in Appendix~\ref{subsec:instantiation}. Assuming Algorithm~\ref{alg:gempStyleCUMG} finds $\mathbf r_i^\tau$ and $h_i^\tau$ close to zero, from Lemma~\ref{lem:smoothedResidualRegret} we can get the desired $\epsilon$ approximate DRE by choosing appropriate small $\kappa$ and $\tau$.

\section{Experiments}
\label{sec:experiments}
Our experiments have two goals. First, we use three illustrative
games to examine how risk aversion affects equilibrium robustness,
out-of-sample performance, and payoff risk under MSD and CVaR. Second, we
study the computational scalability of the various approaches in this work as the number of actions and payoff samples increases.
The setup specific to each experiment is described in the corresponding
subsection.
We use the PATH solver~\citep{Ferris_Munson_1999} for MLCP formulations. The code for all the experiments and figures is available on the authors' Github.\footnote{https://github.com/bharatgw/cumg}

\subsection{Experiments Exploring Structural Properties of DRE} \label{subsec:addtionalexperiments}

Intuitively, one expects players in risk-averse games to obtain benefits similar to those found in single-agent variance-penalized optimization, such as out-of-sample performance improvements, variance reduction, and probabilistic guarantees. Our initial numerical experiments present some evidence for such advantages of CUMG games in data-driven settings. We implement and solve for the equilibrium in three games.
First, with the underlying distribution estimated closely (large-$K$), the mean-semideviation (MSD) equilibrium set in a coordination game setup eliminates pure action equilibria under distributional uncertainty over the player's payoffs. Second, in small-$K$ setting in a variant of prisoner's dilemma, the players choose conservatively and, hence, achieve out-of-sample performance improvement of their actual reward. Finally, the CVaR game demonstrates a monotonic decrease in the player's payoff variance in robustness parameter $\gamma_c$ at equilibrium and yields probabilistic guarantees on the payoff in a general sum game.

\paragraph{Large-K Coordination}
We first return to the setting of Example \ref{ex:coordinationGame} with general empirical probability $\hat{p}$. Recall that $\hat{p}$ is the probability of drawing the State A in the coordination game. In this setting, our mean-semideviation CUMG with parameter $\gamma_s \in [0,1]$ yields the payoff function
\begin{figure}[t]
    \centering
    \begin{subfigure}[b]{0.43\textwidth}
        \centering
        \includegraphics[width=0.95\textwidth]{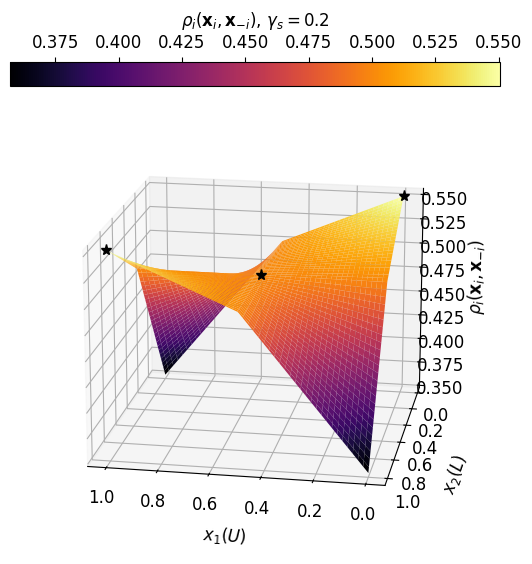}
        \caption{MSD payoffs for the Coordination game with $\hat p = 0.6$ and $\gamma_s = 0.2$.}
        \label{fig:msdCoordG0.2}
    \end{subfigure}
    \hfill
    \begin{subfigure}[b]{0.45\textwidth}
        \centering
        \includegraphics[width=0.95\textwidth]{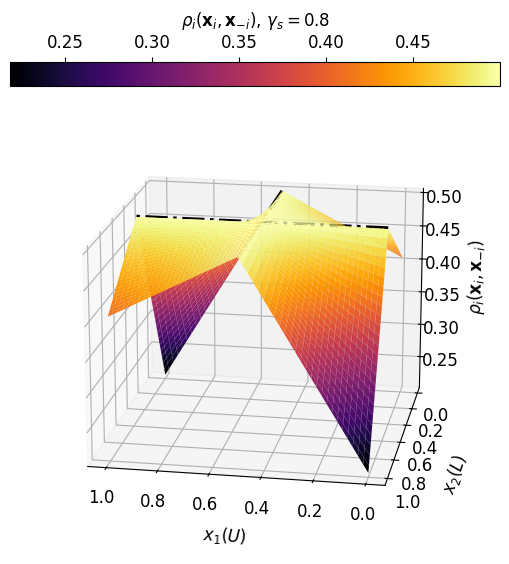}
        \caption{MSD payoffs for the Coordination game with $\hat p = 0.6$ and $\gamma_s = 0.8$.}
        \label{fig:msdCorrdG0.8}
    \end{subfigure}
    
    \caption{MSD-Coordination game payoffs with $\hat p = 0.6$ and different risk aversion levels. $\cdot-\cdot-$ and $\star$ denote the Nash equilibria under the different values of $\gamma_s$. The game does not admit pure action equilibria at higher $\gamma_s$.}
    \label{fig:msdCoord}
\end{figure}
\[
    \rho_i(\mathbf{x}_1, \mathbf{x}_2) =
    \frac{1}{2} + 
    \frac{2\hat{p}-1}{2} 
    \bigl(1-2x_1(U)\bigr)
    \bigl(1-2x_2(L)\bigr) - 
    \gamma_s(1-\hat{p})\hat{p}
    \Bigl|
    \bigl(1-2x_1(U)\bigr)
    \bigl(1-2x_2(L)\bigr)
    \Bigr|
\]

If $x_1(U) = \frac12$ and $x_2(L) = \frac12$, both players receive $\rho_1(\mathbf{x}_1, \mathbf{x}_2) = \rho_2(\mathbf{x}_1, \mathbf{x}_2) = \frac{1}{2}$ independent of $\hat{p}$, thereby attaining a mixed MSD equilibrium. Increasing $\gamma_s$ eliminates pure action equilibria in the game despite greater distributional skew in $\hat{p}$. When the empirical distribution is sufficiently skewed, specifically, when $\hat p$ lies outside an explicit interval derived below, additional pure-strategy MSD equilibria emerge.

The threshold value of $\hat{p}$ which allows other MSD equilibria is determined by comparing the marginal gain from action alignment under the empirical mean payoff with the marginal penalty induced by the downside semideviation. Specifically, $(U, L)$ and $(D,R)$ are pure strategy equilibria whenever
$
    \frac{2\hat{p}-1}{2} > \gamma_s(1-\hat{p})\hat{p}
    \implies \hat{p} > \frac{\gamma_s-1+\sqrt{1+\gamma_s^2}}{2\gamma_s} = \bar{p}
$
and $(U, R)$ and $(D, L)$ are pure strategy equilibria if
$
    \frac{1-2\hat{p}}{2} > \gamma_s(1-\hat{p})\hat{p}
    \implies \hat{p} < \frac{\gamma_s+1-\sqrt{1+\gamma_s^2}}{2\gamma_s} = 1-\bar{p}
$.
Consequently, pure-strategy MSD equilibria exist if and only if $\hat p \notin (1-\bar p,\bar p)$. For instance, at $\hat p=0.6$, $(U,L)$ and $(D,R)$ are pure-strategy equilibria when $\gamma_s=0.2$, but cease to be equilibria when $\gamma_s$ increases to $0.8$. Figure \ref{fig:msdCoord} illustrates the corresponding MSD payoffs. Since $\bar p(\gamma_s)$ is increasing in $\gamma_s$, the interval $(1-\bar p,\bar p)$ in which no pure-strategy equilibria exist expands as $\gamma_s$ increases. Within this interval, the equilibrium set consists of profiles where $x_1(U)=0.5$ or $x_2(L)=0.5$, and both players receive a payoff of $0.5$ independent of $\hat p$. Thus, stronger downside risk aversion eliminates pure-strategy equilibria under a wider range of empirical distributions.

\paragraph{Small-K Prisoner's Dilemma}
An advantage of penalized optimization in single-agent settings is the variance-reduction achieved at the optima which leads to performance improvements in the test distribution compared to empirical optimization. In this example, we show that risk-averse mean-semideviation equilibria demonstrate similar properties. In this example, our true payoff matrices are
\[
\begin{array}{c@{\quad}c @{\quad}c}
    \begin{array}{c}
        \textbf{Typical Dilemma: $\mathbb{P}(\xi_1)=\frac{1}{2}$}\\[4pt]
        \begin{array}{c|cc}
            & \text{Coord.} & \text{Defect}\\ \hline
            \text{Coord.} & (3,3) & (0,5)\\
            \text{Defect} & (5,0) & (1,1)
        \end{array}
    \end{array}
    &
    \begin{array}{c}
        \textbf{Deadlock : $\mathbb{P}(\xi_2) = \frac{1}{4}$}\\[4pt]
        \begin{array}{c|cc}
            & \text{Coord.} & \text{Defect}\\ \hline
            \text{Coord.} & (3,3) & (0,0)\\
            \text{Defect} & (0,0) & (0,0)
        \end{array}
    \end{array}
    & 
    \begin{array}{c}
        \textbf{Ext. Evidence: $\mathbb{P}(\xi_3)=\frac{1}{4}$}\\[4pt]
        \begin{array}{c|cc}
            & \text{Coord.} & \text{Defect}\\ \hline
            \text{Coord.} & (0,0) & (0,0)\\
            \text{Defect} & (0,0) & (1,1)
        \end{array}
\end{array}
\end{array}
\]
\begin{figure}[t]
    \centering
\includegraphics[width=0.9\textwidth]{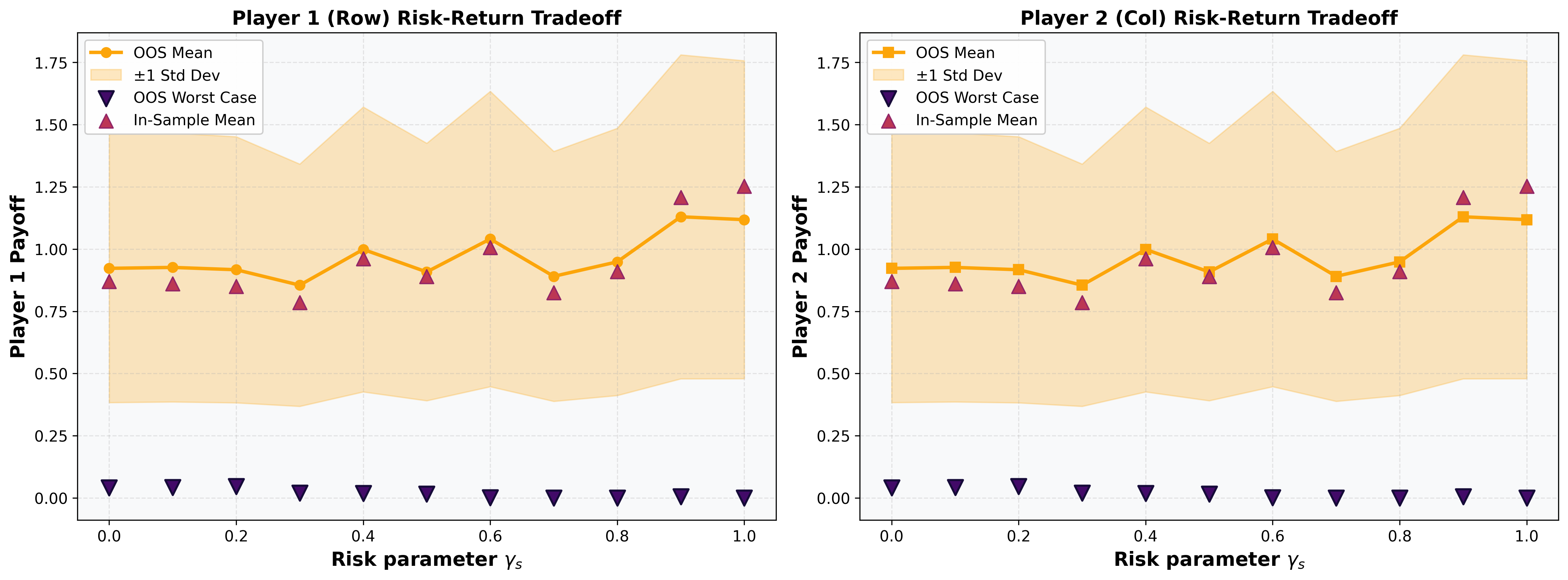}
    \caption{OOS-mean payoff $\pm 1$ s.d. of the two players as a function of $\gamma_s$ in the Small-$K$ Prisoner's Dilemma.
    }
\label{fig:prisonerDilemma}
\end{figure}
We take \emph{five} samples from the true distribution and compute the resulting sample MSD equilibrium using \texttt{PATH} solver~\cite{Ferris_Munson_1999}. This is repeated 100 times to compute 100 such sample equilibria. Out-of-sample mean payoff and standard deviation of these sample equilibrium strategies is computed using a larger dataset of 10,000 samples. The results for out-of-sample performance are visualized in Figure \ref{fig:prisonerDilemma}. Similar to single-agent settings, we observe that a higher downside risk penalty $\gamma_s$ leads to an increase in out-of-sample performance. However, the increase is not monotonic and is accompanied by a falling worst-case payoff. This is expected since non-unique equilibria and changes in the game equilibrium induced by $\gamma_s$ preclude monotonic statements without additional structure. Settings in which mean-payoff and standard deviation are monotonic in $\gamma_s$ are left open for further research.

\paragraph{Large-K CVaR}
Consider the following non-zero sum game from \cite{10.1007/978-3-030-64793-3_22} where only (row) player 1 faces uncertainty in their payoff \(\; \begin{array}{c|cc}
      & y_1 & y_2\\ \hline
x_1   & (1+\xi,3) & (0,2)\\
x_2   & (2,0)     & (-1,1)
\end{array}\; \).
Let $\xi$ take values in $\{-1, 1\}$ with probability of each state being $1/2$. For ease of notation, we use the letters $x_i$ and $y_i$ to denote both the pure actions and their probabilities within mixed strategies. Using action-value complementarity conditions we can solve this game in a closed form by choosing mixed strategies that equalize the action values for the other player. Specifically, player 2 chooses $\mathbf{y}$ to set $\rho_1(x_1, \mathbf{y}) = \rho_1(x_2, \mathbf{y})$ and player 1 chooses $\mathbf{x}$ to set $\rho_2(\mathbf{x}, y_1) = \rho_2(\mathbf{x}, y_2)$. Expanding these, we find equilibrium strategies $\mathbf{x}^*, \mathbf{y}^*$ in a closed-form which are plotted in Figure \ref{fig:cVarGameExample} and obtained as
\begin{align*}
& \alpha\! <\! 0.5\!:\; y^*_1 = \frac{1}{2+\gamma_c}, \;
z^*_1 = \frac{3}{2}y^*_1 - \frac{1}{2}, \;
\rho_1(\mathbf{x}^*, \mathbf{y}^*) = (1-\gamma_c)(y^*_1), \; 
\mathbb{V}[u_1(\mathbf{x}^*, \mathbf{y}^* \mid \xi)] = \frac{(y^*_1)^2}{4}\\
& \alpha\! \geq\! 0.5\!: \; y^*_1\! =\! \frac{\alpha}{2\alpha \!+\! \gamma_c(1-\alpha)}, 
z^*_1\! =\! \frac{5}{2}y^*_1 - \frac{1}{2}, 
\rho_1(\mathbf{x}^*, \mathbf{y}^*) \!=\!
\Big(2 + \frac{\alpha-1}{2\alpha}\gamma_c\Big)y_1^* - \frac{1}{2}, 
\mathbb{V}[u_1(\mathbf{x}^*, \mathbf{y}^* \mid \xi)]\! =\! \frac{(y^*_1)^2}{4}
\end{align*}
\begin{figure}[t]
    \centering
    \includegraphics[width=0.9\textwidth]{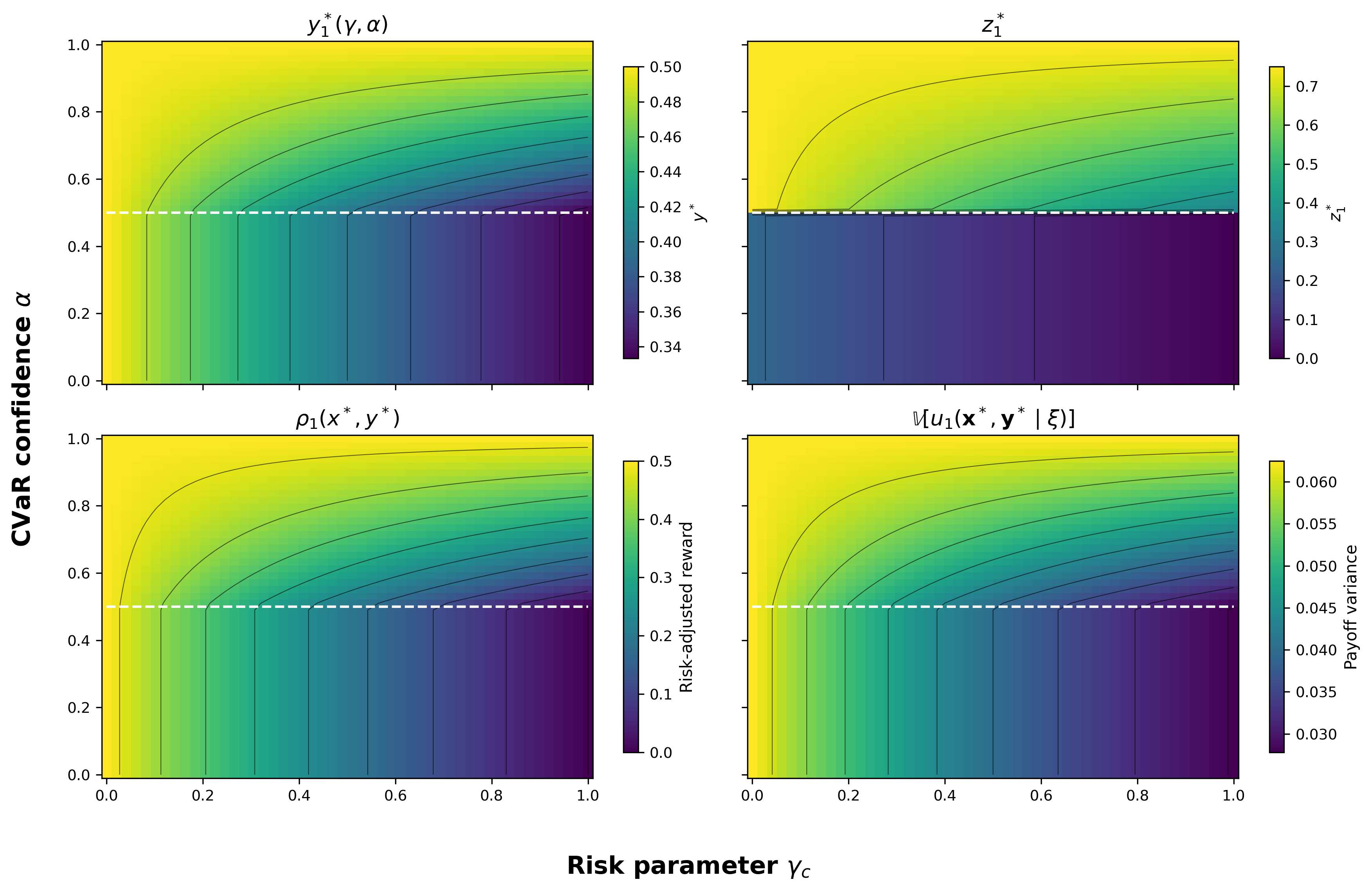}
    \caption{Player 2's strategy $y^*_1$ and player 1's risk-adjusted payoff $\rho_1(\mathbf{x}^*, \mathbf{y}^*)$, VaR, $z^*_1$, and payoff variance $\mathbb{V}[u_1(\mathbf{x}^*, \mathbf{y}^* \mid \xi)]$ as a function of $\alpha$ and $\gamma_c$ in equilibrium. Player 1 experiences a regime change at $\alpha=0.5$. $\mathbb{V}[u_1(\mathbf{x}^*, \mathbf{y}^* \mid \xi)]$ monotonically decreases in $\gamma_c$.}
    \label{fig:cVarGameExample}
\end{figure}
In the above, $z^*_1$ denotes the payoff-VaR obtained by player 1 at the equilibrium. By the definition of VaR, player 1 is guaranteed the payoff $z_1^*$ with probability at least $1-\alpha$. The value of the game for player 2 is $1.5$, induced by $x^*_1=x^*_2=0.5$. The game has a kink at $\alpha = 0.5$, with a change in regime for player 2's strategy $\mathbf{y}^*$ which subsequently induces regime shifts for player 1's risk-adjusted payoff $\rho_1(\mathbf{x}^*, \mathbf{y}^*)$, payoff VaR $z^*_1$, and payoff variance $\mathbb{V}[u_1(\mathbf{x}^*, \mathbf{y}^* \mid \xi)]$. 
For $\alpha < 0.5$, $\rho_1(\mathbf{x}^*, \mathbf{y}^*)$, $z^*_1$, and $y^*_1$ are independent of $\alpha$. Conversely, $\alpha \geq 0.5$ encodes a higher risk-tolerance and therefore raises the expected payoff of the player. Increasing $\alpha$ raises player 1's payoff threshold $z^*_1$ while reducing the probability of attaining it. Player 1's expected payoff and payoff variance $\mathbb{V}[u_1(\mathbf{x}^*, \mathbf{y}^* \mid \xi)]$ decreases monotonically in $\gamma_c$.

\subsection{Scalability Experiments}

As a scalability and feasibility check for the proposed computational formulations, we compare the performance of several algorithms we propose in this work for MSD and CVaR games using two different underlying distributions of game payoffs. These include the support-iteration heuristic with and without screening, the QPTAS and screened QPTAS, and the gradient approach using full batch (all data points) and or randomly sampled mini-batches  to compute gradients. 
We used a 2.0GHz 8-core CPU x86(AMD) server running Ubuntu 24.04 with 32GB of RAM for all scalability experiments. In the first distribution used, all payoffs were sampled from a uniform distribution on $[0,1]$ and $\epsilon$ was set to $0.01$. 
For both the MSD and the CVaR game, $\gamma$ was set to 0.5. For the CVaR game, $\alpha$ was set to 0.05. The MLCP was run with all actions and data samples, and a restricted MLCP was also run with $\kappa$ actions sampled and $\tau$ data sampled at random. A final regret check verified if the strategy output by the restricted MLCP was an approximate equilibrium. For the support-iteration heuristic, the screening step and the final regret computation were performed using the SLSQP solver within the scipy library in Python. The screening step for the support iteration algorithm was initialized from two random points in the simplex and one warm start from the solution of an MLCP program for the subsampled game, over which the minimum was taken. The number of optimizer iterations for the screening step was capped at 1,000. The final regret computation was initialized from 20 random points in the simplex, over which the minimum was taken; failure to certify $\eta_K(x)\le \epsilon$ was counted as failure.  QPTAS was similarly restricted to sampling 1000 candidate $\kappa$-uniform strategy profiles, with subsequent regret checks for each player to certify the profile, and we also experimented with a screened variant (with 1,000 samples) as discussed after the QPTAS results. Gradient method hyperparameters were calibrated using small-scale experiments and there were five runs for each game instance, with the corresponding five initial points chosen randomly in the probability simplex. Each run was limited to 2000 gradient-descent iterations. Further setup details are in Appendix~\ref{sec:experimentdetails}. 

\begin{figure}[t]
\centering\includegraphics[width=0.9\linewidth]{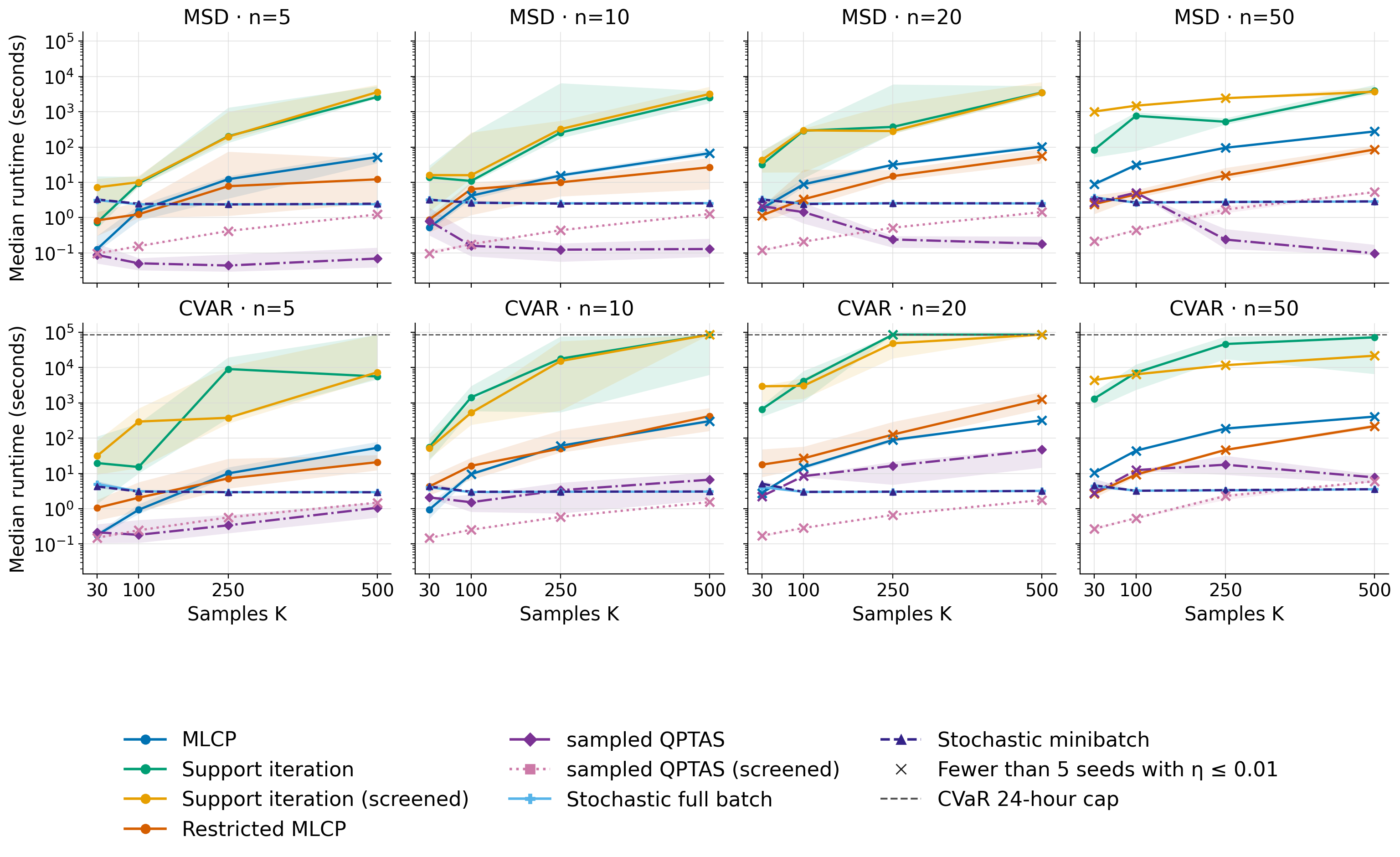}
    \caption{Median and Inter-Quartile Range of the runtime for various algorithms in equilibrium computation for MSD and CVaR games with data sampled from a uniform distribution. $<5$ successes are indicated as failures.} 
    \label{fig:smallSupportMCPPerformance}
\end{figure}

\begin{figure}[t]
   \centering
   \includegraphics[width=0.9\linewidth]{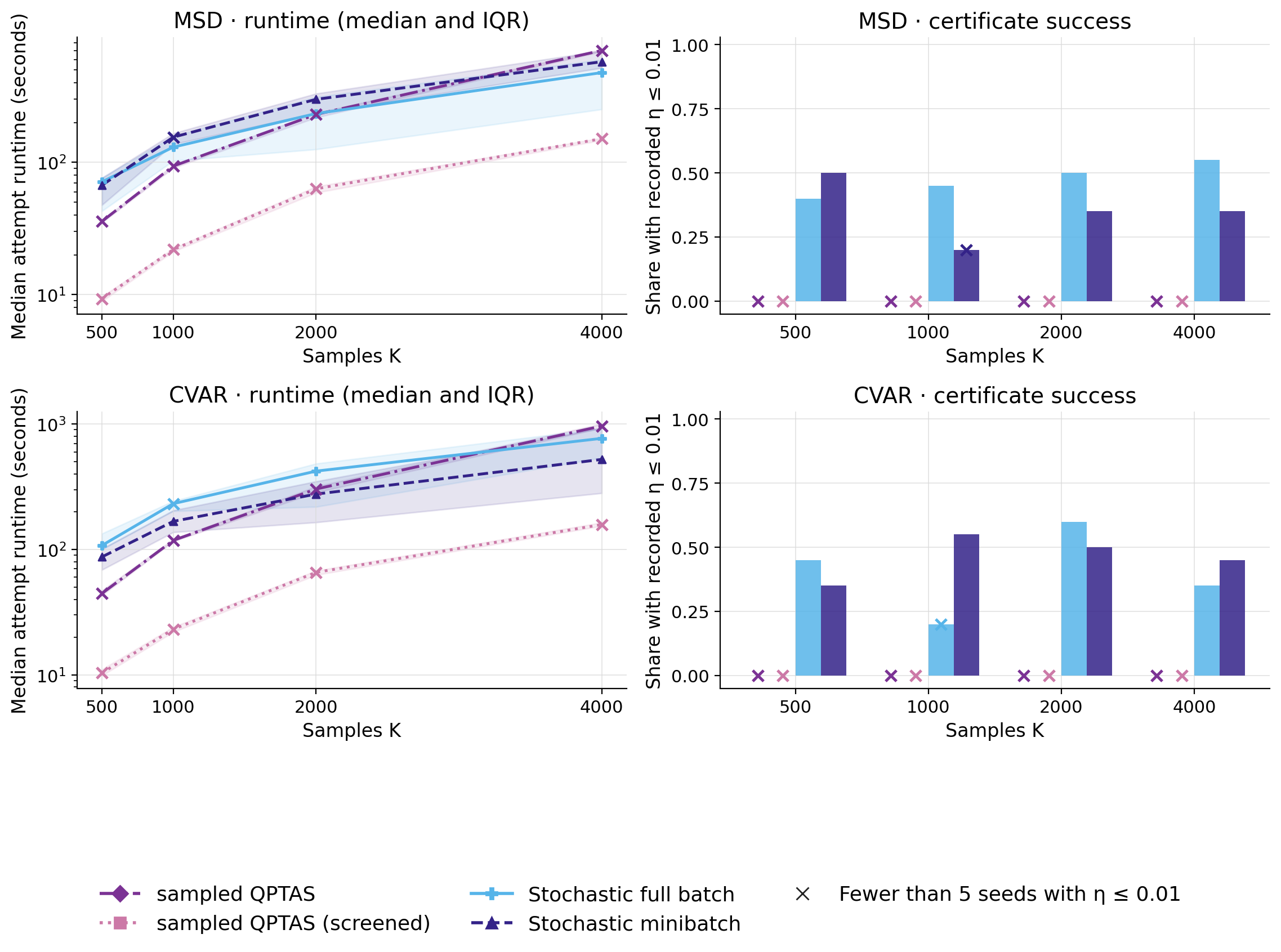}
   \caption{Runtime performance and equilibrium regret certificate success rate for MSD and CVaR games with large $K$ sampled from heterogeneous payoff populations}
   \label{fig:betaPopulationGame}
\end{figure}

The algorithms were compared on a grid comprising $K \in \{30, 100, 250, 500 \}$, and $n \in \{ 5, 10, 20, 50 \}$ for two players. For each combination, the algorithms were tested on 20 randomly generated games with $\kappa$ and $\tau$ chosen as follows:
$\kappa = \min\left(n, \max\left(2, \left\lceil \sqrt{n} \right\rceil\right)\right), 
    \tau = \min(K, \max(5, \lceil \sqrt{K} \rceil)).
$
The median runtime over all runs, along with the interquartile range, for all game configurations is presented in Figure~\ref{fig:smallSupportMCPPerformance}. We observe that as the game size increases, the MLCP programs simply fail to produce an approximate equilibrium. A similar observation holds for support iteration with screening, whereas without screening, the algorithm does output an approximate equilibrium, yet both of them are the most time-consuming among all algorithms. The gradient full-batch and mini-batch take the same time (lines in the graph overlap) and are close to the times observed for QPTAS. The screened QPTAS fails to find an approximate equilibrium. The failure of screening in finding an approximate equilibrium for both QPTAS and support iteration may be because the screening could be eliminating the exploration of any of the limited 1,000 candidates as a candidate approximate equilibrium, where some of these might be an approximate equilibrium. To compare the best-performing algorithms with larger data and prevent adapting to the regularities of one fixed distribution for payoffs, we perform an additional experiment.  


Figure~\ref{fig:betaPopulationGame} presents the additional scalability experiment with heterogeneous payoff distributions. For each player and each joint pure-action profile, we independently assign one of five distributions with equal probability:
$\operatorname{Uniform}[0,1]$, $\operatorname{Beta}(1,4)$,
$\operatorname{Beta}(2,3)$, $\operatorname{Beta}(3,2)$, and
$\operatorname{Beta}(4,1)$.
Holding these assignments fixed, we draw $K$ independent payoff samples for each entry from its assigned distribution, yielding $Kn^2$ payoff observations per player. The experiment uses $n=50$ and $K\in\{500,1000,2000,4000\}$. Neither sampled QPTAS variant achieves a certificate within its 1,000-candidate budget. Reduced concentration in population mean payoffs across action profiles in this heterogeneous regime may contribute to the poorer performance of the sampled QPTAS searches. In contrast, the gradient methods are able to find an approximate equilibrium with certificates $\eta\leq 0.01$ in under half of all the games.

\section{Related Work}
Ambiguity sets for DRG are often specified in two ways in the literature: by optimal transport/statistical-distance balls centered on a nominal distribution or restrictions on moments of the distributions within the ambiguity set. In most such settings, the uncertainty over the payoffs is exogenous to the game. \citet{10.1007/978-3-030-64793-3_22} build and solve a Wasserstein distance-based DRG in a data-driven environment; ~\citet{qu2017distributionally,ahipacsaouglu2015beyond,sunxu} use moment constrained ambiguity sets in a DRG set-up. In the limit, these models converge to a robust game~\citep{Aghassi2006-lx} as the probability mass concentrates on the worst outcome in the support. \citet{loizou2015distributionallyrobustgametheory} combine moment-constrained ambiguity sets and CVaR objective within their DRG model. Despite using CVaR, their approach precludes a tractable closed form solution for the inner infimum, due to the introduction of additional robustness (over CVaR). Other simultaneous, one-shot distributionally robust game settings in the literature include ambiguity sets based on f-divergence~\citep{10.1145/3150928.3150950} and chance constraints~\citep{Singh2017-en}, and maximization over parameterized Bayesian posteriors~\citep{liu2025bayesiandistributionallyrobustnash}. Recent work on DRG by \citet{lanzetti2025strategicallyrobustgametheory} makes the uncertainty endogenous by introducing distributional robustness over other players' mixed strategy profile. Distributional robustness has also been explored in the Stackelberg game~\citep{Liu2018-ng, ananthanarayanan2022computingoptimaldistributionallyrobuststrategy}. In particular, \citet{Liu2018-ng} study general continuous-action Nash and Stackelberg games with random costs convex in each player’s own decision and evaluated under distributional ambiguity; coherent risk measures arise as special cases through their dual risk-envelope representations. However, their framework does not assume finite pure-action sets, payoff matrices, or the multilinear mixed-strategy structure of normal-form games. Consequently, they do not address the finite-game questions studied here, including PPAD/FIXP complexity, complementarity formulations for data-driven MSD/MD/CVaR games, sparse approximate equilibria and QPTASs, or stochastic data-minibatching methods. Finally, a number of works also consider transition uncertainty in a repeated single player setting~\citep{iyengar2005robust,li2025rectangularity}, which has recently been extended to Markov games~\citep{blanchet2023doublepessimismprovablyefficient}. Recent work also reflects growing interest in ambiguity and distributional uncertainty in algorithmic-economic models, including ambiguous contracts~\citep{duetting2023ambiguous}, ambiguous persuasion~\citep{cheng2024persuasion}, and ambiguity-aware bandits~\citep{li2024optimism}. Empirical game-theoretic analysis (EGTA) constructs empirical game models from payoff observations obtained through simulation or data and analyzes equilibria in the resulting estimated game \cite{wellman2006methods,jordan2007empirical,wellman2025egta}; unlike our setting, EGTA does not model distributional uncertainty in payoffs through players' risk preferences, but instead uses sampled outcomes primarily to estimate expected game payoffs. In contrast to all these works, this paper studies one-shot, simultaneous-move games in which payoff uncertainty is modeled directly through coherent utility measures.  For finite data-driven CUMGs, we are not aware of prior complementarity, small-support, or stochastic first-order equilibrium methods of the kind developed here.

\section{Conclusion}
This paper introduces Coherent Utility Measure Games which model players' distributional uncertainty over their payoffs using interpretable coherent utility (risk) measures. Our results show that CUMGs are neither classical matrix games nor arbitrary continuous games, but occupy an intermediate regime. 
Their continuous nature introduces new conceptual and computational challenges, including the inapplicability of familiar techniques such as Lemke–Howson for LCP and, in many cases, the inapplicability of classical linear programming formulation for best response. 
These challenges point to several promising directions for future research, such as designing more scalable algorithms for computing the equilibrium of CUMG. Extending the game framework to dynamic, repeated, or information-asymmetric settings would also allow further enrichment of the modeling of robust strategic behavior in data-driven environments.

\newpage
 
\bibliographystyle{plainnat}
\bibliography{references}

\appendix

\startcontents[appendices]

\section*{Appendix Contents}
\addcontentsline{toc}{section}{Appendix Contents}
\printcontents[appendices]{}{1}{\setcounter{tocdepth}{1}}

\section{Missing Details and Proofs}


\subsection{Proofs in Section~\ref*{subsec:existence}~\nameref*{subsec:existence}}
\paragraph{Proof of Theorem~\ref{thm::dreSMExist}}

\begin{proof}
This proof reduces the number of steps compared to the proof of \citet{qu2017distributionally}, particularly in the proof of continuity below.
Following \citet{qu2017distributionally}, there are two critical results needed to invoke Kakutani's Fixed Point Theorem for the best response correspondence and prove the existence of an equilibrium. We show these two critical result. 
First, we show \emph{concavity} of $\rho$ in $\mathbf{x}_i$,
\begin{align*}
\rho_{i}(\lambda \mathbf{x}'_i + (1 - \lambda) \mathbf{x}''_i, \mathbf{x}_{-i}) &= \inf_{\mathbb{Q}\in U}\mathbb{E}^{\mathbb{Q}}\left[u_i(\lambda \mathbf{x}'_i + (1 - \lambda) \mathbf{x}''_i, \mathbf{x}_{-i}\mid \xi)\right] \\
& = \inf_{\mathbb{Q}\in U} \left[ \lambda \mathbb{E}^{\mathbb{Q}}\left[u_i(\mathbf{x}'_i , \mathbf{x}_{-i}\mid \xi)\right] + (1 - \lambda) \mathbb{E}^{\mathbb{Q}}\left[u_i( \mathbf{x}''_i, \mathbf{x}_{-i}\mid \xi)\right] \right]\\
& \geq \inf_{\mathbb{Q}\in U}\lambda \mathbb{E}^{\mathbb{Q}}\left[u_i(\mathbf{x}'_i , \mathbf{x}_{-i}\mid \xi)\right] + \inf_{\mathbb{Q}\in U} (1 - \lambda) \mathbb{E}^{\mathbb{Q}}\left[u_i( \mathbf{x}''_i, \mathbf{x}_{-i}\mid \xi)\right] \\
& = \lambda \rho_{i}(\mathbf{x}'_i , \mathbf{x}_{-i}) + (1 - \lambda) \rho_{i}(\mathbf{x}''_i , \mathbf{x}_{-i})
\end{align*}
The inequality above follows from the standard result \textcolor{black}{$\inf \sum f_i(x) \geq \sum \inf f_i(x)$}.
Second, we show the \emph{continuity} of $\rho_i$ in its arguments using a much more succinct proof than shown in \citet{qu2017distributionally}. Recall that $\mu^{\mathbb{Q}}_i(\mathbf{a}) = \mathbb{E}^{\mathbb{Q}}[u_i(\mathbf{a}\mid \xi)]$, and consider the set of real-valued vectors $E = \left\{\left(\mu^{\mathbb{Q}}_i(\mathbf{a})\right)_{\mathbf{a} \in \mathbf{A}} \mid \mathbb{Q} \in U \right\}$. It can be easily seen that 
$$\rho_{i}(\mathbf{x}_i , \mathbf{x}_{-i}) = \inf_{\mathbb{Q} \in U} \sum_{\mathbf{a}\in \mathbf{A}} \mu^{\mathbb{Q}}_i(\mathbf{a}) \prod_{s=1}^m x_s(a_{j_s}) = \inf_{(\mu^{\mathbb{Q}}_i(\mathbf{a}))_{\mathbf{a} \in \mathbf{A}} \in E} \sum_{\mathbf{a}\in \mathbf{A}} \mu^{\mathbb{Q}}_i(\mathbf{a}) \prod_{s=1}^m x_s(a_{j_s}).$$

Because of the bounded first moment condition in the theorem, $E$ is a bounded set in $\mathbb{R}^{|\mathbf{A}|}$. Consider the closure $\Bar{E}$ of set $E$, clearly $\Bar{E}$ is compact; also inf over $E$ is the same as inf over $\Bar{E}$, which implies that
$$
\rho_{i}(\mathbf{x}_i , \mathbf{x}_{-i}) = \inf_{(v_\mathbf{a})_{\mathbf{a} \in \mathbf{A}} \in \Bar{E}} \sum_{\mathbf{a}\in \mathbf{A}}v_{\mathbf{a}}\prod_{s=1}^m x_s(a_{j_s}).
$$
Next, using the fact that $g(t) = \inf_{v \in V} f(t,v)$ is continuous if $f$ is continuous and $V$ is compact, we can equate $t$ to $(\mathbf{x}_i , \mathbf{x}_{-i})$ and $V$ to $\Bar{E}$ to claim continuity of $\rho_i$.
Now, we have all the ingredients to invoke Kakutani's Fixed Point Theorem, which proves the existence of a fixed point of the best response correspondence and hence the existence of equilibrium. 
\end{proof}

\paragraph{Proof of Lemma~\ref{lem:lifted_game_dre}}

\begin{proof}
First suppose $(\mathbf{x}^*,\mathbf q^*)$ is a Nash equilibrium of the extended game. Then, for every player $i$,
$
P_i(\mathbf{x}_i^*,\mathbf{x}_{-i}^*,\mathbf q_i^*)
\ge
P_i(\mathbf{x}_i,\mathbf{x}_{-i}^*,\mathbf q_i^*)
\;
\text{for all } \mathbf{x}_i\in\mathbf X_i,
$
and, because adversary $j_i$ maximizes $-P_i$, we have
$
P_i(\mathbf{x}_i^*,\mathbf{x}_{-i}^*,\mathbf q_i^*)
\le
P_i(\mathbf{x}_i^*,\mathbf{x}_{-i}^*,\mathbf q_i)
$
for all $\mathbf q_i\in U.
$
Thus, $\mathbf q_i^*$ is a worst-case distribution for $\mathbf{x}^*$, and
$
P_i(\mathbf{x}_i^*,\mathbf{x}_{-i}^*,q_i^*) = \min_{q_i\in U} P_i(\mathbf{x}_i^*,\mathbf{x}_{-i}^*,\mathbf q_i)$ $=
\rho_i(\mathbf{x}_i^*,\mathbf{x}_{-i}^*).
$
For any $\mathbf{x}_i\in\mathbf X_i$,
\[
\rho_i(\mathbf{x}_i,\mathbf{x}_{-i}^*) =
\min_{q_i\in U} P_i(\mathbf{x}_i,\mathbf{x}_{-i}^*,\mathbf q_i)
\le P_i(\mathbf{x}_i,\mathbf{x}_{-i}^*,\mathbf q_i^*)
\le P_i(\mathbf{x}_i^*,\mathbf{x}_{-i}^*,\mathbf q_i^*) =
\rho_i(\mathbf{x}_i^*,\mathbf{x}_{-i}^*).
\]
Hence $\mathbf{x}^*$ is a DRE.
\emph{Conversely, suppose $\mathbf{x}^*$ is a DRE}. Fix a player $i$. Since $P_i(\mathbf{x}_i,\mathbf{x}_{-i}^*,\mathbf q_i)$ is linear in $\mathbf{x}_i$ and linear in $\mathbf q_i$, and since $\mathbf X_i$ and $U$ are nonempty, compact, and convex, Sion's minimax theorem gives
\[
\max_{\mathbf{x}_i\in\mathbf X_i}
\min_{q_i\in U}
P_i(\mathbf{x}_i,\mathbf{x}_{-i}^*,\mathbf q_i) =
\min_{q_i\in U}
\max_{\mathbf{x}_i\in\mathbf X_i}
P_i(\mathbf{x}_i,\mathbf{x}_{-i}^*,\mathbf q_i).
\]
Because $\mathbf{x}^*$ is a DRE, $\mathbf{x}_i^*$ attains the left-hand-side maximum. Let
\[
v_i :=
\rho_i(\mathbf{x}_i^*,\mathbf{x}_{-i}^*) =
\max_{\mathbf{x}_i\in\mathbf X_i}
\min_{q_i\in U}
P_i(\mathbf{x}_i,\mathbf{x}_{-i}^*,\mathbf q_i).
\]
By compactness of $U$ and continuity of $P_i$, there exists
$
\mathbf q_i^*\in
\argmin_{\mathbf q_i\in U}
\max_{\mathbf{x}_i\in\mathbf X_i}
P_i(\mathbf{x}_i,\mathbf{x}_{-i}^*,\mathbf q_i).
$
Then,
$
\max_{\mathbf{x}_i\in\mathbf X_i}
P_i(\mathbf{x}_i,\mathbf{x}_{-i}^*,\mathbf q_i^*) = v_i.
$
Also,
\[
P_i(\mathbf{x}_i^*,\mathbf{x}_{-i}^*,\mathbf q_i^*)
\ge
\min_{q_i\in U}
P_i(\mathbf{x}_i^*,\mathbf{x}_{-i}^*,\mathbf q_i)
=
v_i,
\;\text{ while }\;
P_i(\mathbf{x}_i^*,\mathbf{x}_{-i}^*,\mathbf q_i^*)
\le
\max_{\mathbf{x}_i\in\mathbf X_i}
P_i(\mathbf{x}_i,\mathbf{x}_{-i}^*,\mathbf q_i^*)
=
v_i.
\]
Therefore,
$
P_i(\mathbf{x}_i^*,\mathbf{x}_{-i}^*,\mathbf q_i^*)=v_i.
$
It follows that
$
P_i(\mathbf{x}_i,\mathbf{x}_{-i}^*,\mathbf q_i^*)
\le
P_i(\mathbf{x}_i^*,\mathbf{x}_{-i}^*,\mathbf q_i^*)
\le
P_i(\mathbf{x}_i^*,\mathbf{x}_{-i}^*,\mathbf q_i)
$
for all $\mathbf{x}_i\in\mathbf X_i$ and all $\mathbf q_i\in U$. Thus, original player $i$ and adversary $j_i$ are mutual best responses. Repeating this argument for every player $i$ gives $\mathbf q^*=(\mathbf q_1^*,\ldots,\mathbf q_m^*)$ such that $(\mathbf{x}^*,\mathbf q^*)$ is a Nash equilibrium of the extended game.
\end{proof}

\subsection{Details and Proofs in Section~\ref*{subsec:maincomplexity}~\nameref*{subsec:maincomplexity}}
\label{subsec:complexity}

For exact equilibrium computation, we use the FIXP framework of
\citet{etessami2010complexity} and the pseudogate construction of
\citet{filosratsikas2023fixp}. We briefly recall the relevant notion.
An algebraic circuit is a finite circuit using the operations
$+,-,\times,\div,\max,\min$ and rational constants. A \emph{pseudogate}
for a set valued correspondence $F$ is an algebraic circuit with additional
internal variables and outputs such that, at any fixed point of the
auxiliary variables, its primary output belongs to $F$. Thus, a
pseudogate for
$\partial_{\mathbf{x}_i}\rho_i(\mathbf{x}_i,\mathbf{x}_{-i})$
does not need to represent the entire superdifferential; at a fixed point,
it only needs to output some
$\mathbf{g}_i\in
\partial_{\mathbf{x}_i}\rho_i(\mathbf{x}_i,\mathbf{x}_{-i})$.
A polynomial-size pseudogate is one whose circuit size and number of
internal variables are polynomial in the encoding size of the game.

\citet{filosratsikas2023fixp} show that exact equilibrium computation
for concave games lies in FIXP under a computational representation in
which the players' strategy sets satisfy their convex-set representation
requirements and polynomial-size pseudogates are available for the
superdifferentials of the players' utility functions.

\paragraph{Proof of Theorem~\ref{thm:exactcomplexity}}
\begin{proof}
\textbf{Part 1}
For every player $i$, the strategy space
$
\mathbf{X}_i=\{\mathbf{x}_i\in[0,1]^{n_i}:
\mathbf{1}_{n_i}^{T}\mathbf{x}_i=1\}
$
is a compact convex set given by polynomially many linear constraints.
To match the computational representation of
\citet{filosratsikas2023fixp}, $\mathbf X_i$ can be represented inside
$[-2,2]^{n_i}$ by the equality
$\mathbf 1_{n_i}^{T}\mathbf x_i=1$ and the linear inequalities
$-x_i(a_l)\leq0$. The equality constraint has full row rank, and the
uniform mixed strategy satisfies all inequalities strictly while lying
in the interior of the bounding box. Thus, the strategy-set constraints
satisfy the required explicit Slater condition and have polynomial-size
algebraic representations.
For fixed $\mathbf{x}_{-i}$,
$u_i(\mathbf{x}_i,\mathbf{x}_{-i}\mid\xi)$ is affine in
$\mathbf{x}_i$. Since $\rho$ is a coherent utility measure and hence
concave,
$
\rho_i(\mathbf{x}_i,\mathbf{x}_{-i})
=
\rho(u_i(\mathbf{x}_i,\mathbf{x}_{-i}\mid\xi))
$
is concave in $\mathbf{x}_i$. Thus, a finite data-driven CUMG is a
concave game.

By assumption in the theorem, a polynomial-size pseudogate is available for
$\partial_{\mathbf{x}_i}\rho_i(\mathbf{x}_i,\mathbf{x}_{-i})$ for
every player $i$. Therefore the exact-equilibrium result for concave
games of~\citet{filosratsikas2023fixp} applies and yields a
polynomial-size FIXP circuit whose fixed points are profiles
$\mathbf{x}^*$ satisfying
\[
\mathbf{x}_i^*
\in
\argmax_{\mathbf{x}_i\in\mathbf{X}_i}
\rho_i(\mathbf{x}_i,\mathbf{x}_{-i}^*),
\qquad
\forall i\in\{1,\ldots,m\}.
\]
By Equation~\eqref{def::drEquilibrium}, such a profile is exactly a DRE.
Hence exact DRE computation lies in FIXP.

For hardness, consider an arbitrary finite normal-form game with rational
payoffs. Construct a CUMG with $K=1$ and set
$
u_i(\mathbf{a}\mid\xi_1)=u_i(\mathbf{a})
$
for every player $i$ and $\mathbf{a}\in\mathbf{A}$.
For any mixed-strategy profile $\mathbf{x}$,
$u_i(\mathbf{x}_i,\mathbf{x}_{-i}\mid\xi_1)$ is then a deterministic
random variable. A coherent utility measure evaluates a deterministic
value at that value: positive homogeneity gives $\rho(0)=0$, and
translation equivariance then gives $\rho(c)=c$ for every
$c\in\mathbb{R}$. Consequently,
$
\rho_i(\mathbf{x}_i,\mathbf{x}_{-i})
=
u_i(\mathbf{x}_i,\mathbf{x}_{-i}).
$
Thus, a DRE of the constructed CUMG is exactly a Nash equilibrium of the
original normal-form game. The construction is polynomial in the input
size. Since exact Nash equilibrium computation with an $\geq 3$ players is FIXP-hard~\citep{etessami2010complexity}, exact DRE
computation is FIXP-hard for any CUMG class containing these
single-sample instances. Together with membership, this establishes
FIXP-completeness.

\textbf{Part 2}:
Following the result of part 1, it remains only to
show that, for MSD, MD, and CVaR, there is a
polynomial-size pseudogate for
$\partial_{\mathbf{x}_i}\rho_i(\mathbf{x}_i,\mathbf{x}_{-i})$ for every
player $i$.

We use the finite-sample dual formulations (Theorem~\ref{thm:rs06}) of the three coherent utility
measures. For MSD and CVaR, these dual formulations are derived as linear programs in
Appendix~\ref{sec:dualSpecificCUMG}. For MD, the dual follows from the
two-inequality linear-program representation of the absolute-value term
described in Appendix~\ref{subsec:MSDcomplementarity}. We use the same
Danskin argument as in the proof of
Lemma~\ref{lem:generalDRGComplementarity} to obtain a supergradient from
an optimal solution of these dual problems.

For MSD, using the notation
$d_{i,k}(\mathbf{x}_i,\mathbf{x}_{-i})
=u_i(\mathbf{x}_i,\mathbf{x}_{-i}\mid\xi_k)
-\mu_i^{\mathbb P}(\mathbf{x}_i,\mathbf{x}_{-i})$, the dual formulation from~\eqref{eq:appMSDdual} in
Appendix~\ref{sec:dualSpecificCUMG} is
\[
\rho_i(\mathbf{x}_i,\mathbf{x}_{-i})
=
\min_{\substack{
0\leq\lambda_{i,k}\leq\gamma_s\mathbb P(\xi_k),\\
k=1,\ldots,K}}
\left\{
\mu_i^{\mathbb P}(\mathbf{x}_i,\mathbf{x}_{-i})
+
\sum_{k=1}^K
\lambda_{i,k}d_{i,k}(\mathbf{x}_i,\mathbf{x}_{-i})
\right\}.
\]
For fixed $(\mathbf{x}_i,\mathbf{x}_{-i})$, this is a linear program in
the variables $\lambda_{i,k}$ with an explicitly bounded feasible region.
Since $\gamma_s>0$ and $\mathbb P(\xi_k)>0$, the choice
$\lambda_{i,k}=\gamma_s\mathbb P(\xi_k)/2$ satisfies all inequalities
strictly. Hence the explicit Slater condition of
\citet{filosratsikas2023fixp} holds, and their Theorem~3.1 gives a
polynomial-size pseudogate whose output, at a fixed point of its internal
variables, is an optimal choice of the variables $\lambda_{i,k}$.

Applying the Danskin argument used in the proof of
Lemma~\ref{lem:generalDRGComplementarity}, such an optimal choice gives a
supergradient $\mathbf g_i$ whose components are
$
g_{i,l}
=
\mu_i^{\mathbb P}(a_l,\mathbf{x}_{-i})
+
\sum_{k=1}^K\lambda_{i,k}
\left(
u_i(a_l,\mathbf{x}_{-i}\mid\xi_k)
-
\mu_i^{\mathbb P}(a_l,\mathbf{x}_{-i})
\right).
$
The terms $u_i(a_l,\mathbf{x}_{-i}\mid\xi_k)$ and
$\mu_i^{\mathbb P}(a_l,\mathbf{x}_{-i})$ are multilinear functions of
$\mathbf{x}_{-i}$ of degree at most $m-1$, and therefore have
polynomial-size algebraic circuits. Composing these circuits with the
pseudogate for the optimal $\lambda_{i,k}$ gives a polynomial-size
pseudogate for
$\partial_{\mathbf{x}_i}\rho_i(\mathbf{x}_i,\mathbf{x}_{-i})$.

For MD, the dual of the absolute-value formulation described in
Appendix~\ref{subsec:MSDcomplementarity} is
\[
\rho_i(\mathbf{x}_i,\mathbf{x}_{-i})
=
\min_{\substack{
-\gamma_d\mathbb P(\xi_k)\leq\lambda_{i,k}
\leq\gamma_d\mathbb P(\xi_k),\\
k=1,\ldots,K}}
\left\{
\mu_i^{\mathbb P}(\mathbf{x}_i,\mathbf{x}_{-i})
+
\sum_{k=1}^K
\lambda_{i,k}d_{i,k}(\mathbf{x}_i,\mathbf{x}_{-i})
\right\}.
\]
And with the same reasoning as for MSD, MD also admits a polynomial-size pseudogate for
$\partial_{\mathbf{x}_i}\rho_i(\mathbf{x}_i,\mathbf{x}_{-i})$.

For CVaR, we use the finite-sample dual formulation
\eqref{eq:appCVaRdual} derived in
Appendix~\ref{sec:dualSpecificCUMG}.
\[
\rho_i(\mathbf{x}_i,\mathbf{x}_{-i})
=
\min_{\substack{
0\leq\lambda_{i,k}\leq
\frac{\gamma_c}{\alpha}\mathbb P(\xi_k),\\
\sum_{k=1}^K\lambda_{i,k}=\gamma_c}}
\left\{
(1-\gamma_c)\mu_i^{\mathbb P}(\mathbf{x}_i,\mathbf{x}_{-i})
+
\sum_{k=1}^K
\lambda_{i,k}u_i(\mathbf{x}_i,\mathbf{x}_{-i}\mid\xi_k)
\right\}.
\]
For $\alpha<1$, this linear program has a feasible region explicitly bounded, and the choice
$\lambda_{i,k}=\gamma_c\mathbb P(\xi_k)$ satisfies the equality constraint
and all inequality constraints strictly, since
$0<\gamma_c\mathbb P(\xi_k)
<(\gamma_c/\alpha)\mathbb P(\xi_k)$. The single equality constraint has
full row rank. Hence the explicit Slater condition holds, and
Theorem~3.1 of~\citet{filosratsikas2023fixp} gives a polynomial-size
pseudogate for an optimal choice of $\lambda_{i,k}$. The same Danskin
argument gives a supergradient with components
$
g_{i,l}
=
(1-\gamma_c)\mu_i^{\mathbb P}(a_l,\mathbf{x}_{-i})
+
\sum_{k=1}^K
\lambda_{i,k}u_i(a_l,\mathbf{x}_{-i}\mid\xi_k).
$
Thus, CVaR also admits a polynomial-size pseudogate for its
superdifferential. 
When $\alpha=1$, $\mathrm{CVaR}_{1}[X]=\mathbb E^{\mathbb P}[X]$, so $\rho_i=\mu_i^{\mathbb P}$ and the required supergradient is obtained directly, without an optimization pseudogate.

Therefore, MSD, MD, and CVaR satisfy the superdifferential-pseudogate
condition of Part 1, and exact DRE
computation for all three classes lies in FIXP.

For hardness, the single-sample reduction of Part 1 applies directly, as when $K=1$, the
MSD and MD deviation terms are zero, while CVaR of a deterministic payoff
equals that payoff. 
Together with
membership in FIXP, this establishes FIXP-completeness.
\end{proof}

\paragraph{Proof of Theorem~\ref{thm:complexity}}

\begin{proof}
\noindent\textbf{Part 1.}
For $K=1$, a coherent utility measure evaluates the deterministic payoff at its value. The CUMG is therefore a finite normal-form game, for which approximate Nash equilibrium computation lies in PPAD. Consider $K\geq 2$.

We apply the result for \textsc{ConcaveGames with SO} of \citet{papadimitriou2023computational} to a normalized version of the extended game. Recall that, in the extended game, original player $i$ receives $P_i(\mathbf{x},\mathbf q_{j_i})$ and adversary $j_i$ receives $-P_i(\mathbf{x},\mathbf q_{j_i})$, where
\[
P_i(\mathbf{x},\mathbf q_{j_i})=\sum_{k=1}^K q_{j_i,k}\sum_{\mathbf a\in\mathbf A}u_i(\mathbf a\mid\xi_k)\prod_{h=1}^m x_h(a_h).
\]
The adversary's strategy set is $U$, and $\rho_i(\mathbf{x})=\min_{\mathbf q_{j_i}\in U}P_i(\mathbf{x},\mathbf q_{j_i})$.

First, the extended game's joint strategy set is $\prod_{i=1}^m\mathbf X_i\times U^m$. Let $\Pi$ drop the last coordinate of each probability vector. Its application to the joint strategy set leaves $D=\sum_{i=1}^m(n_i-1)+m(K-1)$ coordinates. We will transform this set into a compact convex set $S\subseteq[-1,1]^D$ containing $\mathbf 0$, with a polynomial-size strong separation-oracle circuit and positive radii $r,R$ of polynomial encoding length such that $B(\mathbf a_0,r)\subseteq S\subseteq B(\mathbf 0,R)\subseteq[-1,1]^D$ for some $\mathbf a_0$, as required by \textsc{ConcaveGames with SO}.

We apply $\Pi$ separately to each original player's simplex and each adversary's envelope.
Thus $\Pi(\mathbf{x},\mathbf{q})=\bigl((x_i(a_\ell))_{\ell=1}^{n_i-1},(q_{j_i,k})_{k=1}^{K-1}\bigr)_{i=1}^m\in\mathbb{R}^D$, where $D=\sum_{i=1}^m(n_i-1)+m(K-1)$; each omitted coordinate is determined by the requirement that probabilities sum to one. By the hypothesis in part (1), there is a $\mathbf q^\circ\in U$ such that $B(\Pi(\mathbf q^\circ),r_U)\subseteq\Pi(U)$. The hypothesis also provides a polynomial-size strong separation-oracle circuit for $\Pi(U)$.

Applying strong-oracle convex feasibility to $\Pi(U)$ with the bound $r_U$ yields a rational projected point. Restoring its last coordinate gives a rational feasible point $\mathbf q^0\in U$. Give each original player the coordinates $D^{-1}\Pi(\mathbf x_i)$ and each adversary $j_i$ the coordinates $D^{-1}(\Pi(\mathbf q_{j_i})-\Pi(\mathbf q^0))$. Their product strategy set is $S=\prod_{i=1}^m\{D^{-1}\Pi(\mathbf{x}_i):\mathbf{x}_i\in\mathbf{X}_i\}\times\prod_{i=1}^m\{D^{-1}(\Pi(\mathbf{q}_{j_i})-\Pi(\mathbf{q}^0)):\mathbf{q}_{j_i}\in U\}$. The zero vector belongs to $S$: each original player can put all probability on the dropped action, and each adversary can choose $\mathbf q^0$. Before division by $D$, every coordinate has absolute value at most one, so the norm of a joint profile after division is at most $\sqrt D/D\leq1$. Hence $S\subseteq B(\mathbf 0,1)\subseteq[-1,1]^D$.

For each original player $i$, $B(\Pi(\mathbf{1}_{n_i}/n_i),n_i^{-2})\subseteq\Pi(\mathbf X_i)$. Indeed, if the retained coordinates change by $\boldsymbol\delta$ with $\|\boldsymbol\delta\|_2\leq n_i^{-2}$, each retained coordinate is at least $1/n_i-n_i^{-2}>0$, while the omitted coordinate is at least $1/n_i-(n_i-1)n_i^{-2}=n_i^{-2}>0$. Consequently, $S$ contains 
$$B(D^{-1}(\Pi(\mathbf{1}_{n_1}/n_1),\ldots,\Pi(\mathbf{1}_{n_m}/n_m),\Pi(\mathbf q^\circ)-\Pi(\mathbf q^0),\ldots,\Pi(\mathbf q^\circ)-\Pi(\mathbf q^0)),D^{-1}\min\{r_U,\min_i n_i^{-2}\}),$$ 
where the adversary block occurs $m$ times. This positive radius has polynomial encoding length, so the inner-ball inclusion and the outer bound above make $S$ well-bounded. The simplex inequalities and the assumed separation oracle for $\Pi(U)$, composed with the rational affine coordinate changes, give a polynomial-size strong separation-oracle circuit for $S$; the inverse coordinate changes are polynomial-time computable.

We next supply the utility circuits and the required payoff normalization. On the full box $[-1,1]^D$, each probability coordinate recovered by the inverse affine changes has absolute value at most $1+D^2$. Let $B=1+K|\mathbf A|(1+D^2)^{m+1}$. For $l<n_i$, the transformed coordinate is $x_i(a_l)/D$. At a point of $[-1,1]^D$, that coordinate has absolute value at most one, so $|x_i(a_l)/D|\leq1$, which gives $|x_i(a_l)|\leq D$, thus, at any point of $[-1,1]^D$, the inverse coordinate changes give $|x_i(a_l)|\leq D$ for $l<n_i$ and similarly, $|x_i(a_{n_i})|\leq1+D(n_i-1)\leq1+D^2$. Likewise, since $\mathbf q^0\in U\subseteq\Delta_K$, they give $|q_{j_i,k}|\leq1+D$ for $k<K$ and $|q_{j_i,K}|\leq1+D(K-1)\leq1+D^2$. Thus every product $q_{j_i,k}\prod_{h=1}^m x_h(a_h)$ has absolute value at most $(1+D^2)^{m+1}$. As $0\leq u_i(\mathbf a\mid\xi_k)\leq1$ and $P_i$ has $K|\mathbf A|$ such summands, $|P_i|\leq K|\mathbf A|(1+D^2)^{m+1}<B$ throughout the cube. For this PPAD reduction only, give original player $i$ utility $V_i=(P_i+B)/(2B)$ and adversary $j_i$ utility $1-V_i$. Both utilities lie in $[0,1]$ on the full box. These positive affine changes preserve every player's exact best responses: original player $i$ still maximizes $P_i$, and adversary $j_i$ still minimizes it.

The explicit payoff table yields a polynomial-size circuit for $P_i$: it is the sum of the $K|\mathbf A|$ explicitly represented terms
$
u_i(\mathbf a\mid\xi_k)q_{j_i,k}\prod_{h=1}^m x_h(a_h),
$
and every factor $x_h(a_h)$ and $q_{j_i,k}$ is an affine function of the transformed coordinates. The partial derivatives also have polynomial-size circuits. In particular, for an original player's transformed coordinate $x_g(a_\ell)/D$, where $\ell<n_g$, differentiating the displayed definition of $P_i$ gives
\[
\frac{\partial P_i}
{\partial\bigl(x_g(a_\ell)/D\bigr)}
=
\sum_{k=1}^K\sum_{\mathbf a\in\mathbf A}
u_i(\mathbf a\mid\xi_k)q_{j_i,k}
\frac{\partial x_g(a_g)}
{\partial\bigl(x_g(a_\ell)/D\bigr)}
\prod_{h\neq g}x_h(a_h),
\]
where the derivative of $x_g(a_g)$ is $D$ if $a_g=a_\ell$, is $-D$
if $a_g$ is the dropped action $a_{n_g}$, and is zero otherwise. Similarly, for an adversary's transformed coordinate
$(q_{j_i,\ell}-q_\ell^0)/D$, where $\ell<K$,
\[
\frac{\partial P_i}
{\partial\bigl((q_{j_i,\ell}-q_\ell^0)/D\bigr)}
=
\sum_{k=1}^K\sum_{\mathbf a\in\mathbf A}
u_i(\mathbf a\mid\xi_k)
\frac{\partial q_{j_i,k}}
{\partial\bigl((q_{j_i,\ell}-q_\ell^0)/D\bigr)}
\prod_{h=1}^m x_h(a_h),
\]
where the derivative of $q_{j_i,k}$ is $D$ if $k=\ell$, is $-D$ if
$k=K$, and is zero otherwise. Derivatives with respect to the coordinates
of any other adversary are zero. Each displayed derivative is therefore a sum of at most $K|\mathbf A|$ products of at most $m$ affine factors.
Since there are $D$ transformed coordinates, polynomial-size circuits for $P_i$ and all its partial derivatives can be constructed from the
explicit payoff table in polynomial time. Since $V_i=(P_i+B)/(2B)$, the corresponding value-and-gradient circuits for $V_i$ and $1-V_i$ require only affine output operations. For fixed opponents, $P_i$ is affine separately in $\mathbf x_i$ and
$\mathbf q_{j_i}$, and the affine coordinate changes preserve this property. Thus $V_i$ is affine, hence concave, in original player $i$'s transformed coordinates, while $1-V_i$ is affine, hence concave, in adversary $j_i$'s transformed coordinates.

The same derivative formulas above give a polynomially encoded Lipschitz parameter. Every nonzero derivative of a probability coordinate appearing in those formulas has absolute value $D$. Each resulting summand contains $m$ undifferentiated strategy factors, each having absolute value at most $1+D^2$ on the box, while $|u_i(\mathbf a\mid\xi_k)|\leq1$. Hence each summand in a partial derivative has absolute value at most $D(1+D^2)^m$. 
Summing over the $K|\mathbf A|$ summands gives the bound $K|\mathbf A|D(1+D^2)^m$. After division by $2B$, every partial derivative of $V_i$ or $1-V_i$ has absolute value at most $D/(2(1+D^2))$. Therefore $\|\nabla V_i\|_2\leq\|\nabla V_i\|_1 \leq D^2/(2(1+D^2))<1/2\leq D$, and the same holds for $1-V_i$.
Thus $D$ is a valid Euclidean Lipschitz bound. The binary encoding lengths of $B$ and $D$, and the sizes of the circuits, are polynomial in the CUMG input length.

The normalized extended game is therefore an instance of \textsc{ConcaveGames with SO}. Its PPAD membership result gives a feasible $\delta$-approximate Nash equilibrium for $\delta=\epsilon/(4B)$. Decode its coordinates to obtain $(\mathbf{x},\mathbf q)$. The adversary's approximate best-response condition implies $P_i(\mathbf{x},\mathbf q_{j_i})\leq\rho_i(\mathbf{x})+2B\delta$. The original player's approximate best-response condition implies $P_i(\mathbf{x}_i',\mathbf{x}_{-i},\mathbf q_{j_i})\leq P_i(\mathbf{x},\mathbf q_{j_i})+2B\delta$ for every $\mathbf{x}_i'\in\mathbf X_i$. Since $\rho_i(\mathbf{x}_i',\mathbf{x}_{-i})\leq P_i(\mathbf{x}_i',\mathbf{x}_{-i},\mathbf q_{j_i})$, these inequalities give $\rho_i(\mathbf{x}_i',\mathbf{x}_{-i})\leq\rho_i(\mathbf{x})+4B\delta=\rho_i(\mathbf{x})+\epsilon$. This holds for every player $i$ and every unilateral deviation, so $\mathbf{x}$ is an $\epsilon$-DRE. The choice of $\delta$ has polynomial encoding length, establishing PPAD membership.

For hardness, start with a finite normal-form game with rational payoffs in $[0,1]$, and make it a single-sample CUMG by setting $u_i(\mathbf a\mid\xi_1)=u_i(\mathbf a)$ for every $i$ and $\mathbf a$. Every resulting payoff random variable is deterministic. Positive homogeneity gives $\rho(0)=0$, and translation equivariance then gives $\rho(c)=c$ for every constant $c$. Thus the CUMG payoff of every mixed-strategy profile equals its payoff in the original game, and an $\epsilon$-DRE is exactly an $\epsilon$-Nash equilibrium. This is a polynomial-time reduction from approximate equilibrium in finite normal-form games. Any class containing these instances is therefore PPAD-hard, proving the final claim of part (1).

\noindent\textbf{Part 2.}
It remains to verify the risk-envelope conditions in part (1) for MSD, MD, and CVaR. Write $p_k=\mathbb P(\xi_k)>0$ and $p_{\min}=\min_k p_k$. If $K=1$, all three coherent utilities equal the deterministic payoff, as used in part (1). Assume $K\geq2$, and use the same coordinate projection $\Pi$ that drops the last probability coordinate. For any distribution $\mathbf q$ reconstructed from projected coordinates satisfying $\|\Pi(\mathbf q)-\Pi(\mathbf p)\|_2\leq r$, the first $K-1$ coordinates satisfy $|q_k-p_k|\leq r$. Since the last coordinate is one minus their sum, $|q_K-p_K|\leq(K-1)r$. In particular, $|q_k-p_k|\leq Kr$ for every $k$. Also, because every envelope below lies in $\Delta_K$, its projection lies in the Euclidean unit ball.

Fix a player $i$ and suppress $i$ on the dual variables and distributions in the remainder of this proof. For MSD, let $d_{i,k}(\mathbf{x})=u_i(\mathbf{x}\mid\xi_k)-\mu_i^{\mathbb P}(\mathbf{x})$. For MSD, the dual formulation from~\eqref{eq:appMSDdual} in Appendix~\ref{sec:dualSpecificCUMG} is
\[
\rho_i(\mathbf{x})=\min_{0\leq\lambda_k\leq\gamma_s p_k,\ \forall k}\left\{\mu_i^{\mathbb P}(\mathbf{x})+\sum_{k=1}^K\lambda_k d_{i,k}(\mathbf{x})\right\}.
\]
Expanding the objective gives $q_k=p_k+\lambda_k-p_k\sum_r\lambda_r$. These weights sum to one. Setting $t_k=\lambda_k/p_k$ and $\bar t=\sum_rp_rt_r$ gives $q_k/p_k=1+t_k-\bar t$. Since $0\leq t_k\leq\gamma_s\leq1$, every $q_k$ is nonnegative, and $|q_k/p_k-q_\ell/p_\ell|\leq\gamma_s$ for all $k,\ell$.
Conversely, suppose $\mathbf q\in\Delta_K$ satisfies these pairwise inequalities. Set $y_k=q_k/p_k-1$ and $t_k=y_k-\min_r y_r$. The pairwise inequalities give $0\leq t_k\leq\gamma_s$. Since $\sum_kp_ky_k=0$, we have $\sum_kp_kt_k=-\min_r y_r$. Thus $\lambda_k=p_kt_k$ satisfies the dual bounds and reproduces $q_k=p_k+\lambda_k-p_k\sum_r\lambda_r$. Therefore
\[
U_{\mathrm{MSD}}=\left\{\mathbf q\in\Delta_K:\left|q_k/p_k-q_{k'}/p_{k'}\right|\leq\gamma_s\ \text{for all }k,{k'}\right\}.
\]
After multiplying by $p_kp_{k'}>0$, each pairwise condition consists of two rational linear inequalities. After substituting $q_K=1-\sum_{k=1}^{K-1}q_k$, these inequalities and the simplex constraints give a polynomial-size strong separation-oracle circuit for $\Pi(U_{\mathrm{MSD}})$: test the inequalities and return any violated one as a separating hyperplane.

To verify an inner ball, set $r_s=\gamma_s p_{\min}/(4K)$. If $\|\Pi(\mathbf q)-\Pi(\mathbf p)\|_2\leq r_s$, then $|q_k-p_k|\leq Kr_s$ for all $k$. Consequently $q_k\geq p_{\min}(1-\gamma_s/4)>0$ and $|q_k/p_k-q_{k'}/p_{k'}|\leq2Kr_s/p_{\min}=\gamma_s/2<\gamma_s$. The reconstructed coordinates sum to one, so $\mathbf q\in U_{\mathrm{MSD}}$. Hence the ball of radius $r_s$ centered at $\Pi(\mathbf p)$ lies in $\Pi(U_{\mathrm{MSD}})$.

For MD, we have $\rho_i(\mathbf{x}_i,\mathbf{x}_{-i})=\mu_i^{\mathbb P}(\mathbf{x}_i,\mathbf{x}_{-i})-\gamma_d\sum_{k=1}^K p_k|d_{i,k}(\mathbf{x}_i,\mathbf{x}_{-i})|$. Then, for each $k$, $-\gamma_dp_k|d_{i,k}(\mathbf{x}_i,\mathbf{x}_{-i})|=\min_{|\lambda_{i,k}|\leq\gamma_dp_k}\lambda_{i,k}d_{i,k}(\mathbf{x}_i,\mathbf{x}_{-i})$. Thus the dual formulation is
\[
\rho_i(\mathbf{x})=\min_{-\gamma_d p_k\leq\lambda_k\leq\gamma_d p_k,\ \forall k}\left\{\mu_i^{\mathbb P}(\mathbf{x})+\sum_{k=1}^K\lambda_k d_{i,k}(\mathbf{x})\right\}.
\]
The same expansion gives $q_k=p_k+\lambda_k-p_k\sum_r\lambda_r$. Now $-\gamma_d\leq t_k=\lambda_k/p_k\leq\gamma_d$, so $q_k/p_k=1+t_k-\bar t$. Since $\gamma_d\leq1/2$, these weights are nonnegative, and $|q_k/p_k-q_{k'}/p_{k'}|\leq2\gamma_d$.

Conversely, suppose $\mathbf q\in\Delta_K$ satisfies these pairwise inequalities. Set $y_k=q_k/p_k-1$, $y_{\min}=\min_k y_k$, $y_{\max}=\max_k y_k$, and $t_k=y_k-(y_{\max}+y_{\min})/2$. Since $y_{\max}-y_{\min}\leq2\gamma_d$, we have $-\gamma_d\leq t_k\leq\gamma_d$. Also $\sum_kp_kt_k=-(y_{\max}+y_{\min})/2$, because $\sum_kp_ky_k=0$. Taking $\lambda_k=p_kt_k$ therefore recovers the stated expression for $q_k$. Thus
\[
U_{\mathrm{MD}}=\left\{\mathbf q\in\Delta_K:\left|q_k/p_k-q_{k'}/p_{k'}\right|\leq2\gamma_d\ \text{for all }k,{k'}\right\}.
\]
As with MSD, substituting $q_K=1-\sum_{k=1}^{K-1}q_k$ gives polynomially many rational linear inequalities and a polynomial-size strong separation-oracle circuit for $\Pi(U_{\mathrm{MD}})$. Set $r_d=\gamma_d p_{\min}/(4K)$. If $\|\Pi(\mathbf q)-\Pi(\mathbf p)\|_2\leq r_d$, then $|q_k-p_k|\leq Kr_d$, so $q_k\geq p_{\min}(1-\gamma_d/4)>0$ and $|q_k/p_k-q_{k'}/p_{k'}|\leq2Kr_d/p_{\min}=\gamma_d/2<2\gamma_d$. Hence the ball of radius $r_d$ centered at $\Pi(\mathbf p)$ lies in $\Pi(U_{\mathrm{MD}})$.

For CVaR, the finite-sample dual formulation derived in Appendix~\ref{sec:dualSpecificCUMG} is
\[
\rho_i(\mathbf{x})=\min_{0\leq\lambda_k\leq\gamma_c p_k/\alpha\ \forall k,\ \sum_k\lambda_k=\gamma_c}\left\{(1-\gamma_c)\mu_i^{\mathbb P}(\mathbf{x})+\sum_{k=1}^K\lambda_k u_i(\mathbf{x}\mid\xi_k)\right\}.
\]
Setting $q_k=(1-\gamma_c)p_k+\lambda_k$ yields
\[
U_{\mathrm{CVaR}}=\left\{\mathbf q\in\Delta_K:(1-\gamma_c)p_k\leq q_k\leq(1-\gamma_c+\gamma_c/\alpha)p_k\ \text{for all }k\right\}.
\]
Indeed, the dual bounds and $\sum_k\lambda_k=\gamma_c$ imply these constraints and $\sum_kq_k=1$. Conversely, a distribution satisfying the displayed constraints gives $\lambda_k=q_k-(1-\gamma_c)p_k$, which obeys every dual bound and sums to $\gamma_c$. After substituting $q_K=1-\sum_{k=1}^{K-1}q_k$, the displayed inequalities and simplex constraints give a polynomial-size strong separation-oracle circuit for $\Pi(U_{\mathrm{CVaR}})$.

Since $\gamma_c>0$ and $\alpha\in(0,1)$, the distance from $p_k$ to the lower bound is $\gamma_cp_k$, and its distance to the upper bound is $\gamma_c(1-\alpha)p_k/\alpha$. Set $r_c=\gamma_cp_{\min}\min\{1,(1-\alpha)/\alpha\}/(2K)$. If $\|\Pi(\mathbf q)-\Pi(\mathbf p)\|_2\leq r_c$, then every $|q_k-p_k|$ is at most $Kr_c$, which is at most half the smaller of these two coordinate slacks. The reconstructed coordinates sum to one and satisfy all the CVaR bounds. Thus the ball of radius $r_c$ centered at $\Pi(\mathbf p)$ lies in $\Pi(U_{\mathrm{CVaR}})$.

The three radii are positive and have polynomial encoding length in the encoded probabilities and risk parameters. Each projected envelope also lies in the Euclidean unit ball, so each is full-dimensional and well-bounded. Finally, $\mathbf p\in U$ in all three cases: take $\lambda_k=0$ for MSD and MD, and $\lambda_k=\gamma_cp_k$ for CVaR. Thus $\mathbf p$ supplies the feasible point $q^0$ used in part (1). The PPAD membership conclusion follows.

The single-sample hardness reduction in part (1) applies to each class. For $K=1$, the MSD and MD deviation terms vanish, while CVaR of a deterministic payoff equals that payoff. Therefore $\epsilon$-DRE computation for MSD, MD, and CVaR is PPAD-complete.
\end{proof}


\subsubsection{Irrational DRE for 2-player Game}
\label{subsec:irrational}
Consider the 2-player game with $K=2$ with the following payoff matrices
\[
    \xi_1:\quad
    \begin{array}{c|cc}
          & L           & R           \\ \hline
        U & (5/8,\,5/8) & (4/8,\,2/8) \\
        D & (2/8,\,4/8) & (3/8,\,3/8)
    \end{array}
    \qquad\qquad
    \xi_2:\quad
    \begin{array}{c|cc}
          & L           & R           \\ \hline
        U & (3/8,\,3/8) & (4/8,\,6/8) \\
        D & (6/8,\,4/8) & (5/8,\,5/8)
    \end{array}
\]
with $\mathbb{P}(\xi_1) = \mathbb{P}(\xi_2) = \frac{1}{2}$.
Let $x = P_1(U)$, $y = P_2(L)$, and $\gamma \in (0, 1]$. In this setting, the payoffs for player 1 and player 2 are
\begin{align*}
    \rho_1(x, y)  = \frac{1}{2} - \frac{\gamma}{16}\left| 2xy + x - y - 1 \right|, \qquad
    \rho_2(x, y)  = \frac{1}{2} - \frac{\gamma}{16}\left| 2xy - x + y - 1 \right|. 
\end{align*}
and the DRG equilibrium is obtained when
\begin{align*}
    2xy + x - y - 1  = 0, \qquad
    2xy - x + y - 1  = 0 
\end{align*}
Solving these, we obtain the equilibrium in irrational mixed strategies with
$
    (x ,y) = \left(\frac{\sqrt{2}}{2}, \frac{\sqrt{2}}{2} \right).
$

\subsection{Details and Proofs in Section~\ref*{sec:complementarity}~\nameref*{sec:complementarity}}
\allowdisplaybreaks[3]
\paragraph{Proof of Lemma~\ref{lem:generalDRGComplementarity}}
\begin{proof}
Fix a player $i$ and the opponents' strategy profile $\mathbf{x}_{-i}^*$. By definition of DRE, $\mathbf{x}^*$ is a DRE if and only if $\mathbf{x}_i^*$ solves $\max_{\mathbf{x}_i\in\mathbf{X}_i}\rho_i(\mathbf{x}_i,\mathbf{x}_{-i}^*)$ for every player $i$. Since $u_i(\mathbf{x}_i,\mathbf{x}_{-i}^*\mid \xi)$ is affine in $\mathbf{x}_i$ and $\rho$ is concave, $\rho_i(\cdot,\mathbf{x}_{-i}^*)$ is concave. Therefore, the KKT conditions for this concave maximization problem over the simplex are necessary and sufficient.

Let $\mathbf{g}_i\in\partial_{\mathbf{x}_i}\rho_i(\mathbf{x}_i^*,\mathbf{x}_{-i}^*)$. The KKT conditions imply that there exist $\alpha_i\in\mathbb{R}$ and nonnegative multipliers $\beta_{i,l}\geq 0$ for the constraints $x_i(a_l)\geq 0$ such that $g_{i,l}-\alpha_i+\beta_{i,l}=0$ and $\beta_{i,l}x_i^*(a_l)=0$ for all $l$. Eliminating $\beta_{i,l}$ gives $0\leq\alpha_i-g_{i,l}\perp x_i^*(a_l)\geq 0$, together with $\mathbf{1}^\top\mathbf{x}_i^*=1$ and $\mathbf{x}_i^*\geq 0$. This proves the primal/supergradient form. Since $\alpha_i-g_{i,l}\geq 0$ for all $l$, $\alpha_i$ is an upper bound on all $g_{i,l}$. Also, because $\mathbf{x}_i^*$ lies on the simplex, at least one action has positive probability, and complementarity gives equality for every such action. Hence $\alpha_i=\max_{1\leq l\leq n_i}g_{i,l}$.

We now derive the dual form. Since we are in the finite-sample setting, the risk envelope is represented as a closed subset $U\subseteq\Delta_K$, so the minimum defining $U_i^*(\mathbf x^*)$ is attained. By Theorem~\ref{thm:rs06}, $\rho(X)=\inf_{\mathbb{Q}\in U}\mathbb{E}^{\mathbb{Q}}[X]$. In the finite-sample setting, this is equivalently
\[
\rho_i(\mathbf{x}_i,\mathbf{x}_{-i})
=
\inf_{\mathbf{q}_i\in U}
\sum_{k=1}^K q_{i,k}u_i(\mathbf{x}_i,\mathbf{x}_{-i}\mid \xi_k).
\]
For fixed $\mathbf{x}_{-i}^*$, the inner objective is affine in $\mathbf{x}_i$. Therefore, Danskin's theorem gives
\[
\partial_{\mathbf{x}_i}\rho_i(\mathbf{x}_i^*,\mathbf{x}_{-i}^*)
=
\operatorname{conv}\{\mathbf{v}_i(\mathbf{q}_i,\mathbf{x}_{-i}^*):\mathbf{q}_i\in U_i^*(\mathbf{x}^*)\}.
\]
Substituting any vector from this superdifferential into the primal form gives Equation~\ref{eq:generalDualMCP}. If the worst-case dual distribution is unique, the active set $U_i^*(\mathbf{x}^*)$ is a singleton, so the convex hull reduces to the single vector generated by $\mathbf{q}_i^*$. This proves the stated dual form.
\end{proof}

\subsubsection{Details of MSD and MD Complementarity Program}
\label{subsec:MSDcomplementarity}


To formulate the MLCP to solve for the equilibrium of a MSD game, we consider the problem of a single player $i$, given $\mathbf{x}_{-i}$. Recall the notation $\mu^{\mathbb{P}}_i(\mathbf{x}_i, \mathbf{x}_{-i}) = \mathbb{E}^{\mathbb{P}}[u_i(\mathbf{x}_i, \mathbf{x}_{-i}\mid \xi)] $ and 
$\mu^{\mathbb{Q}}_i(\mathbf{a}) = \mathbb{E}^{\mathbb{Q}}[u_i(\mathbf{a}\mid \xi)] $.
\begin{align}\label{opt::ogMax}
\max_{\mathbf{x}_i \in \mathbf{X}_i} & \quad\mu^{\mathbb{P}}_i(\mathbf{x}_i, \mathbf{x}_{-i}) - 
\gamma_s \mathbb{E}^{\mathbb{P}}[\max(0, \mu^{\mathbb{P}}_i(\mathbf{x}_i, \mathbf{x}_{-i}) - u_i(\mathbf{x}_i, \mathbf{x}_{-i}\mid \xi) )]
\end{align}

Let the auxiliary variable $z_{i,k}$ denote $-\max(0, \mu^{\mathbb{P}}_i(\mathbf{x}_i, \mathbf{x}_{-i}) - u_i(\mathbf{x}_i, \mathbf{x}_{-i}\mid \xi_k) ) = \min(0, u_i(\mathbf{x}_i, \mathbf{x}_{-i}\mid \xi_k ) - \mu^{\mathbb{P}}_i(\mathbf{x}_i, \mathbf{x}_{-i}) )$. Then, the second term in the objective is $\sum_{k=1}^K \mathbb{P}(\xi_k)z_{i,k}$.
Next, we write the terms in the optimization explicitly in the optimization variables, that is, $\mu^{\mathbb{P}}_i(\mathbf{x}_i, \mathbf{x}_{-i}) = \sum_{l=1}^{n_i} x_i(a_l) \mu^{\mathbb{P}}_i(a_l, \mathbf{x}_{-i})$ and $u_i(\mathbf{x}_i, \mathbf{x}_{-i}\mid \xi) = \sum_{l=1}^{n_i} x_i(a_l) u_i(a_l, \mathbf{x}_{-i}\mid \xi)$. 
\begin{subequations} \label{eq:MSDLP}
\begin{align}
\max_{\mathbf{x}_i, \mathbf{z}_i} & \quad \sum_{l=1}^{n_i} x_i(a_l) \mu^{\mathbb{P}}_i(a_l, \mathbf{x}_{-i}) + 
\gamma_s \sum_{k=1}^K \mathbb{P}(\xi_k)z_{i,k}\label{newMax::obj} \\
\text{s.t.}  & \quad z_{i,k} \leq \sum_{l=1}^{n_i} x_i(a_l) u_i(a_l, \mathbf{x}_{-i}\mid \xi_k) - \sum_{l=1}^{n_i} x_i(a_l) \mu^{\mathbb{P}}_i(a_l, \mathbf{x}_{-i}), \quad \forall k \in \{1, \dots, K\} \label{MSDLP::cons1} \\
& \quad z_{i,k} \leq 0, \quad \forall k \in \{1, \dots, K\}  \label{MSDLP::cons2}  \\
& \quad  \mathbf{1}^T\mathbf{x}_i=1 \label{MSDLP::cons3}  \\
& \quad 
x_i(a_l) \geq 0, \quad \forall l \in \{1, \ldots, n_i\}  \label{MSDLP::cons4} 
\end{align}
\end{subequations}
The optimality conditions for the linear program $\eqref{eq:MSDLP}$ yield dual multipliers $\lambda_{i,k}$ for constraints in \eqref{MSDLP::cons1}, $\mu_{i,k}$ for \eqref{MSDLP::cons2}, $\alpha_i$ for \eqref{MSDLP::cons3}, and $\beta_{i,l}$ for \eqref{MSDLP::cons4}.
We introduce two shorthand notation for easier presentation:
\begin{align*}
    v_{i,l}(\mathbf{x}_{-i}) & = \mu^{\mathbb{P}}_i(a_l, \mathbf{x}_{-i}) + \sum_{k=1}^K\lambda_{i,k}\left( u_i(a_l, \mathbf{x}_{-i} \mid \xi_k) - \mu^{\mathbb{P}}_i(a_l, \mathbf{x}_{-i}) \right),\;\\
    d_{i,k}(\mathbf{x}_i, \mathbf{x}_{-i}) & = \sum_{l=1}^{n_i}x_i(a_l) u_i(a_l, \mathbf{x}_{-i} \mid \xi_k) - \sum_{l=1}^{n_i}x_i(a_l) \mu^{\mathbb{P}}_i(a_l, \mathbf{x}_{-i})
\end{align*}
Then, we obtain the following:
\begin{lemma} \label{lem:msdMLCP_main}
Consider the following definition of $\tau$ and substitutions for $\mathbf{x}_i$ and $\mathbf{z}_i$,
\[
    \tau \coloneq \Big( \prod_{j = 1}^m \alpha_j \Big)^{\frac{1}{m-1}} \qquad \tilde{\mathbf{x}}_i \coloneq \frac{\alpha_i}{\tau}\mathbf{x}_i, \quad \forall i  \qquad \tilde{z}_{i,k} \coloneq \frac{z_{i,k}}{\tau}, \quad \forall i,k \;, 
\]
and the MLCP formed by stacking the constraints below for all players
\begin{subequations}
\begin{align}
    & 0 \leq \lambda_{i,k} \perp  \left(d_{i,k}(\tilde{\mathbf{x}}_i, \tilde{\mathbf{x}}_{-i}) - \tilde{z}_{i,k} \right) \geq 0, \quad \forall k \in \{1, \dots, K\}   \\
    & 0 \leq (\gamma_s \mathbb{P}(\xi_k) - \lambda_{i,k})  \perp (-\tilde{z}_{i,k}) \geq 0, \quad \forall k \in \{1, \dots, K\}  \\
    & 0 \leq (1 - v_{i,l}(\tilde{\mathbf{x}}_{-i})) \perp \tilde{x}_i(a_l) \geq 0, \quad \forall l \in \{1, \dots, n_i\}
\end{align}
\end{subequations}
Any non-zero $\tilde{\mathbf{x}}_i$ (forall $i$) MLCP solutions correspond to DREs of the mean-semideviation game, 
with the mixed strategies and game values retrieved by
\[
    \mathbf{x}_i = \frac{\tilde{\mathbf{x}}_i}{\sum_{l=1}^{n_i} \tilde{x}_i(a_l)} \qquad \alpha_i = \frac{\sum_{l=1}^{n_i} \tilde{x}_i(a_l)}{\prod_{i=1}^m\sum_{l=1}^{n_i} \tilde{x}_i(a_l)}
\]
\end{lemma}
For the special case of two players, $m=2$, let $(A^\xi,B^\xi)$ denote the payoff matrices contingent on $\xi$, $\big(A = \mathbb{E}^\mathbb{P}\left[A^\xi\right], B = \mathbb{E}^\mathbb{P}\left[B^\xi\right]\big)$. Then, we obtain the following multilinear complementarity program:
\begin{equation}\label{opt::twoPlayerFinal_main}
\begin{gathered}
    0 \leq \lambda_k^1 \perp \tilde{\mathbf{x}}^\top(A^{\xi_k} -  A)\tilde{\mathbf{y}} - \tilde{z}_k^1 \geq 0\\
    0 \leq \gamma_s \mathbb{P}(\xi_k) - \lambda^1_k \perp -\tilde{z}_k^1 \geq 0\\
    0 \leq 1 - \tilde{v}_{l_1}^1(\tilde{\mathbf{y}}) \perp \tilde{x}(a_{l_1}) \geq 0
\end{gathered}
\qquad
\begin{gathered}
    0 \leq \lambda_k^2 \perp \tilde{\mathbf{x}}^\top(B^{\xi_k} -  B)\tilde{\mathbf{y}} - \tilde{z}_k^2 \geq 0,\\
    0 \leq \gamma_s \mathbb{P}(\xi_k) - \lambda^2_k \perp -\tilde{z}_k^2 \geq 0,\\
    0 \leq 1 - \tilde{v}_{l_2}^2(\tilde{\mathbf{x}}) \perp \tilde{y}(a_{l_2}) \geq 0
\end{gathered}
\qquad
\begin{gathered}
    \forall k \in \{1, \dots, K\}\\
    \forall k \in \{1, \dots, K\}\\
    \forall l_i \in \{1, \dots, n_i\}
\end{gathered}
\end{equation}
In particular, unlike standard games, the two-player MSD system remains bilinear in the scaled strategies and is therefore not a linear complementarity program of the Lemke-Howson type.

\paragraph{Mean-deviation} Before proving Lemma~\ref{lem:msdMLCP_main}, we note that the mean-deviation game has a very similar setup and result as above, starting with an auxiliary variable $z_{i,k} = - |\mu^{\mathbb{P}}_i(\mathbf{x}_i, \mathbf{x}_{-i}) - u_i(\mathbf{x}_i, \mathbf{x}_{-i}\mid \xi_k)|$ and constraints $z_{i,k} \leq  \mu^{\mathbb{P}}_i(\mathbf{x}_i, \mathbf{x}_{-i}) - u_i(\mathbf{x}_i, \mathbf{x}_{-i}\mid \xi_k)$ and $z_{i,k} \leq  u_i(\mathbf{x}_i, \mathbf{x}_{-i}\mid \xi_k) -\mu^{\mathbb{P}}_i(\mathbf{x}_i, \mathbf{x}_{-i}) $. The rest of the steps are tedious but very similar to the mean-semideviation game and therefore omitted.

We prove the result of Lemma~\ref{lem:msdMLCP_main} by splitting it into a lemma and its corollary below. The corollary is the same as Lemma~\ref{lem:msdMLCP_main}, which proves the MSD result.

\begin{lemma} \label{lem:msdMCP}
    Consider the following set of constraints for player $i$
\begin{subequations}
\begin{align}
    & 0 \leq \lambda_{i,k} \perp  \left(d_{i,k}(\mathbf{x}_i, \mathbf{x}_{-i}) - z_{i,k} \right) \geq 0, \quad \forall k \in \{1, \dots, K\} \\
    & 0 \leq (\gamma_s \mathbb{P}(\xi_k) - \lambda_{i,k} ) \perp (-z_{i,k}) \geq 0, \quad \forall k \in \{1, \dots, K\} \\
    & 0 \leq (\alpha_i - v_{i,l}(\mathbf{x}_{-i})) \perp x_i(a_l) \geq 0, \quad \forall l \in \{1, \dots, n_i\} \\
    & \mathbf{1}^\top\mathbf{x}_i = 1, \quad \alpha_i \in \mathbb{R} 
\end{align}
\end{subequations}
Stack the above for all players to get a mixed complementarity program whose solutions correspond to DREs of the mean-semideviation game. 
Also, $\sum_l^{n_i} x_i(a_l)v_{i,l}(\mathbf{x}_{-i}) = \alpha_i$, that is, $\alpha_i$ is the value of the game for player $i$. 
\end{lemma}
Note that the above is a \emph{mixed} complementarity program with an equality constraint, whereas in standard games one generally obtains a multilinear complementarity program (MLCP) in the standard form without any equality constraint. This construction is possible here:
\begin{corollary} \label{cor:msdMLCP}
Consider the following definition of $\tau$ and substitutions for $\mathbf{x}_i$ and $\mathbf{z}_i$,
\[
    \tau \coloneq \Big( \prod_{j = 1}^m \alpha_j \Big)^{\frac{1}{m-1}} \qquad \tilde{\mathbf{x}}_i \coloneq \frac{\alpha_i}{\tau}\mathbf{x}_i, \quad \forall i  \qquad \tilde{z}_{i,k} \coloneq \frac{z_{i,k}}{\tau}, \quad \forall i,k \;, 
\]
and the multilinear complementarity program formed by stacking the constraints below for all players
\begin{subequations}
\begin{align}
    & 0 \leq \lambda_{i,k} \perp  \left(d_{i,k}(\tilde{\mathbf{x}}_i, \tilde{\mathbf{x}}_{-i}) - \tilde{z}_{i,k} \right) \geq 0, \quad \forall k \in \{1, \dots, K\}   \\
    & 0 \leq (\gamma_s \mathbb{P}(\xi_k) - \lambda_{i,k})  \perp (-\tilde{z}_{i,k}) \geq 0, \quad \forall k \in \{1, \dots, K\}  \\
    & 0 \leq (1 - v_{i,l}(\tilde{\mathbf{x}}_{-i})) \perp \tilde{x}_i(a_l) \geq 0, \quad \forall l \in \{1, \dots, n_i\}
\end{align}
\end{subequations}
    Any non-zero $\tilde{\mathbf{x}}_i$ (forall $i$) solution of the above program gives a solution to the program of Lemma~\ref{lem:msdMCP}, 
with the mixed strategies and game values retrieved by
\[
    \mathbf{x}_i = \frac{\tilde{\mathbf{x}}_i}{\sum_{l=1}^{n_i} \tilde{x}_i(a_l)} \qquad \alpha_i = \frac{\sum_{l=1}^{n_i} \tilde{x}_i(a_l)}{\prod_{i=1}^m\sum_{l=1}^{n_i} \tilde{x}_i(a_l)}
\]
\end{corollary}

\paragraph{Proof of Lemma~\ref{lem:msdMCP} and Corollary~\ref{cor:msdMLCP}}
\begin{proof}

\textbf{Lemma Proof}: The optimality conditions for the linear program $\eqref{eq:MSDLP}$ yield dual multipliers $\lambda_{i,k}$ for constraints in \eqref{MSDLP::cons1}, $\mu_{i,k}$ for \eqref{MSDLP::cons2}, $\alpha_i$ for \eqref{MSDLP::cons3}, and $\beta_{i,l}$ for \eqref{MSDLP::cons4}, yielding:
\begin{subequations}
\begin{align}
& \underbrace{\mu^{\mathbb{P}}_i(a_l, \mathbf{x}_{-i}) + \sum_{k=1}^K\lambda_{i,k}\left( u_i(a_l, \mathbf{x}_{-i} \mid \xi_k) - \mu^{\mathbb{P}}_i(a_l, \mathbf{x}_{-i}) \right)}_{v_{i,l}(\mathbf{x}_{-i})} - \alpha_i + \beta_{i,l} = 0, \quad \forall l \in \{1,\dots,n_i \} \label{mcp::x} \\
& \gamma_s \mathbb{P}(\xi_k) - \lambda_{i,k} - \mu_{i,k} = 0, \quad  \forall k \in \{1, \dots, K\} \label{mcp::z} \\
& 0 \leq \lambda_{i,k} \perp  \left(\sum_{l=1}^{n_i}x_i(a_l) u_i(a_l, \mathbf{x}_{-i} \mid \xi_k) - \sum_{l=1}^{n_i}x_i(a_l) \mu^{\mathbb{P}}_i(a_l, \mathbf{x}_{-i}) - z_{i,k} \right) \geq 0, \quad  \forall k \in \{1, \dots, K\} \\
& 0 \leq \mu_{i,k}  \perp (-z_{i,k}) \geq 0, \quad  \forall k \in \{1, \dots, K\} \\
& 0 \leq \beta_{i,l} \perp x_i(a_l) \geq 0, \quad  \forall l \in \{1, \dots, n_i\} \label{mcp::suppConst} \\
& \mathbf{1}^\top\mathbf{x}_i = 1, \quad  \alpha_i \in \mathbb{R}
\end{align}
\end{subequations}
Note the notation $v_{i,l}$ above, which we call as value for action $l$. \eqref{mcp::suppConst} implies that actions in the support of MSD game equilibrium strategy have equal risk-adjusted reward $v_{i,l} (\mathbf{x}_{-i})=\alpha_i$ for all actions $l$. Simplifying by substituting $\beta_{i,l}$ and $\mu_{i,k}$ using \eqref{mcp::x} and \eqref{mcp::z}, we get the program:

\begin{subequations}\label{opt::newMax2}
\begin{align}
    & 0 \leq \lambda_{i,k} \perp  \left(\underbrace{\sum_{l=1}^{n_i}x_i(a_l) u_i(a_l, \mathbf{x}_{-i} \mid \xi_k) - \sum_{l=1}^{n_i}x_i(a_l) \mu^{\mathbb{P}}_i(a_l, \mathbf{x}_{-i})}_{d_{i,k}(\mathbf{x}_i,\mathbf{x}_{-i})} - z_{i,k} \right) \geq 0, \quad \forall k \in \{1, \dots, K\} \label{newMax2::msdCond} \\
    & 0 \leq (\gamma_s \mathbb{P}(\xi_k) - \lambda_{i,k} ) \perp (-z_{i,k}) \geq 0, \quad \forall k \in \{1, \dots, K\} \label{newMax2::slackCond} \\
    & 0 \leq (\alpha_i - v_{i,l}(\mathbf{x}_{-i})) \perp x_i(a_l) \geq 0, \quad \forall l \in \{1, \dots, n_i\} \label{newMax2::value}\\
    & \mathbf{1}^\top\mathbf{x}_i = 1, \quad \alpha_i \in \mathbb{R} \label{newMax2::equality}
\end{align}
\end{subequations}

Eq.~\eqref{newMax2::value} implies:
\begin{align}
    (\alpha_i - v_{i,l}(\mathbf{x}_{-i}))x_i(a_l) = 0 \implies \sum_l^{n_i} x_i(a_l)v_{i,l}(\mathbf{x}_{-i}) = \alpha_i\label{valueConst2}
\end{align}

In other words, $\alpha_i$ is the value of the game for player $i$. Note that at equilibrium $v_{i,l}(\mathbf{x}_{-i})$ also depends on $\mathbf x_i$ via the dual variables $\lambda_{i,k}$.

Additionally, observe the following interesting side-result: complementarity conditions \eqref{newMax2::msdCond} and \eqref{newMax2::slackCond} show that
\begin{align*}
        d_{i,k}(\mathbf{x}_i,\mathbf{x}_{-i}) > 0 &\implies (z_{i,k}, \lambda_{i,k}) = (0, 0)\\
        d_{i,k}(\mathbf{x}_i,\mathbf{x}_{-i}) < 0 &\implies (z_{i,k}, \lambda_{i,k}) = (d_{i,k}(\mathbf{x}_i,\mathbf{x}_{-i}), \gamma_s\mathbb{P}(\xi_k))\\
        d_{i,k}(\mathbf{x}_i,\mathbf{x}_{-i}) = 0 &\implies z_{i,k} = 0, \ \lambda_{i,k} \in [0, \gamma_s\mathbb{P}(\xi_k)]
\end{align*}

\textbf{Corollary Proof}: First, we show that $\alpha_i \geq 0$.
 We know that $0 \leq \gamma_s \leq 1$. From Eq.~\ref{mcp::z} we know that $\lambda_{i,k} \leq \gamma_s \mathbb{P}(\xi_k)$. 
Then $\sum_{k=1}^K \lambda_{i,k}  \leq \gamma_s \leq 1$ and, since $0 \leq u^i(\mathbf{a}\mid\xi) \leq 1,\forall \mathbf{a}$, we have $ \mu^{\mathbb{P}}_i(a_l, \mathbf{x}_{-i}) $ and $ \{ \mu^{\mathbb{P}}_i(a_l, \mathbf{x}_{-i} \mid \xi_k  \}_{k=1}^K$ are all between $0$ and $1$.
$v_{i,l}(\mathbf{x}_{-i})$ is a convex combination of $\{ \mu^{\mathbb{P}}_i(a_l, \mathbf{x}_{-i} \}$ and $ \{ \mu^{\mathbb{P}}_i(a_l, \mathbf{x}_{-i} \mid \xi_k  \}_{k=1}^K$ as the equation below shows. Hence, $0 \leq v_{i,l}(\mathbf{x}_{-i}) \leq 1$ for all $l$, and by Eq.~\ref{valueConst2} $\alpha_i$ is a convex combination of $\{ v_{i,l}(\mathbf{x}_{-i}) \}_{l=1}^{n_i}$ and hence, $0 \leq \alpha_i \leq 1$.
\begin{align*}
    v_{i,l}(\mathbf{x}_{-i}) & = \mu^{\mathbb{P}}_i(a_l, \mathbf{x}_{-i}) + \sum_{k=1}^K\lambda_{i,k}\left( u_i(a_l, \mathbf{x}_{-i} \mid \xi_k) - \mu^{\mathbb{P}}_i(a_l, \mathbf{x}_{-i}) \right) \\
    & = (1 - \sum_{k=1}^K\lambda_{i,k})\mu^{\mathbb{P}}_i(a_l, \mathbf{x}_{-i}) + \sum_{k=1}^K\lambda_{i,k} u_i(a_l, \mathbf{x}_{-i} \mid \xi_k)
\end{align*}
To obtain a multilinear complementarity program, we need to remove the equality constraint \eqref{newMax2::equality}. We start with a non-zero solution of the MLCP stated in corollary statement, which yields $\alpha_i > 0$. Then, we consider the following definitions for $\mathbf{x}_i$ and $z_{i,k}$,
\begin{equation}
    \tau \coloneq \left( \prod_{j = 1}^m \alpha_j \right)^{\frac{1}{m-1}} \qquad \tilde{\mathbf{x}}_i \coloneq \frac{\alpha_i}{\tau}\mathbf{x}_i, \quad \forall i  \qquad \tilde{z}_{i,k} \coloneq \frac{z_{i,k}}{\tau}, \quad \forall i,k \label{eq:transformMSD}
\end{equation}
With this observe that 
\begin{align*}
    \sum_{l=1}^{n_i}x_i(a_l) u_i(a_l, \mathbf{x}_{-i} \mid \xi_k) & = \sum_{l=1}^{n_i}x_i(a_l) \sum_{\mathbf{a}_{-i} \in \mathbf{A}_{-i}} u_i(a_l, \mathbf{a}_{-i} \mid \xi_k) \prod_{j=1, j \neq i}^m x_j (a_{-i, j})\\
    & = \sum_{l=1}^{n_i}x_i(a_l) \sum_{\mathbf{a}_{-i} \in \mathbf{A}_{-i}} u_i(a_l, \mathbf{a}_{-i} \mid \xi_k) \tau^{m-1} \prod_{j=1, j \neq i}^m \tilde{x}_j (a_{-i, j}) / \alpha_j\\
    & = \sum_{l=1}^{n_i}x_i(a_l) \sum_{\mathbf{a}_{-i} \in \mathbf{A}_{-i}} u_i(a_l, \mathbf{a}_{-i} \mid \xi_k)  \prod_{j = 1}^m \alpha_j \prod_{j=1, j \neq i}^m \tilde{x}_j (a_{-i, j}) / \alpha_j\\
    & = \sum_{l=1}^{n_i}x_i(a_l) \sum_{\mathbf{a}_{-i} \in \mathbf{A}_{-i}} u_i(a_l, \mathbf{a}_{-i} \mid \xi_k)  \alpha_i \prod_{j=1, j \neq i}^m \tilde{x}_j (a_{-i, j}) \\
    & = \sum_{l=1}^{n_i} \alpha_i x_i(a_l) \sum_{\mathbf{a}_{-i} \in \mathbf{A}_{-i}} u_i(a_l, \mathbf{a}_{-i} \mid \xi_k)   \prod_{j=1, j \neq i}^m \tilde{x}_j (a_{-i, j}) \\
    & = \tau \sum_{l=1}^{n_i}\tilde{x}_i(a_l) u_i(a_l, \tilde{\mathbf{x}}_{-i} \mid \xi_k)
\end{align*}
Similarly, we get $\sum_{l=1}^{n_i}x_i(a_l) \mu^{\mathbb{P}}_i(a_l, \mathbf{x}_{-i}) = \tau \sum_{l=1}^{n_i} \tilde{x}_i(a_l) \mu^{\mathbb{P}}_i(a_l, \tilde{\mathbf{x}}_{-i})$ and $v_{i,l}(\mathbf{x}_{-i}) = \alpha_i v_{i,l}(\tilde{\mathbf{x}}_{-i})$. Note that complementarity constraints remain the same when multiplied by any constant, and we obtain MCP \eqref{opt::newMax2}, except for the equality constraint, by multiplying these common factors $\tau$ for Eq.~\ref{newMax3::msdCond},~\ref{newMax3::slackCond} and $\alpha_i$ and $\tau$ for Eq.~\ref{newMax3::value} in the following MLCP (which is same as in the corollary statement). 
\begin{subequations}\label{opt::newMax3}
\begin{align}
    & 0 \leq \lambda_{i,k} \perp  \left(\sum_{l=1}^{n_i}\tilde{x}_i(a_l) u_i(a_l, \tilde{\mathbf{x}}_{-i} \mid \xi_k) - \sum_{l=1}^{n_i}\tilde{x}_i(a_l) \mu^{\mathbb{P}}_i(a_l, \tilde{\mathbf{x}}_{-i}) - \tilde{z}_{i,k} \right) \geq 0, \quad \forall k \in \{1, \dots, K\}   \label{newMax3::msdCond} \\
    & 0 \leq (\gamma_s \mathbb{P}(\xi_k) - \lambda_{i,k})  \perp (-\tilde{z}_{i,k}) \geq 0, \quad \forall k \in \{1, \dots, K\} \label{newMax3::slackCond} \\
    & 0 \leq (1 - v_{i,l}(\tilde{\mathbf{x}}_{-i})) \perp \tilde{x}_i(a_l) \geq 0, \quad \forall l \in \{1, \dots, n_i\}\label{newMax3::value} 
\end{align}
\end{subequations}
As the MLCP is just a scaled version of the MCP, it is easy to check that any non-zero solution of the MLCP \eqref{opt::newMax3} scaled as 
$
    \mathbf{x}_i = \frac{\tilde{\mathbf{x}}_i}{\sum_{l=1}^{n_i} \tilde{x}_i(a_l)} \quad \alpha_i = \frac{\sum_{l=1}^{n_i} \tilde{x}_i(a_l)}{\prod_{i=1}^m\sum_{l=1}^{n_i} \tilde{x}_i(a_l)}
$
is a solution of the MCP \eqref{opt::newMax2} as these satisfies all constraints and also satisfies the equality $\mathbf 1^\top \mathbf x_i = 1$.

For completeness, let's instantiate the result for two players.
Setting the number of players $m=2$, let $\left(A^\xi,B^\xi\right)$ denote the payoff matrices contingent on $\xi$, $\big(A = \mathbb{E}^\mathbb{P}\left[A^\xi\right], B = \mathbb{E}^\mathbb{P}\left[B^\xi\right]\big)$ their expectations over the sample distribution, and $(\mathbf{x},\mathbf{y})$ the mixed strategies for players 1 and 2. We then extend program \eqref{opt::newMax2} to both players as:
\begin{equation}\label{opt::twoPlayer}
\begin{gathered}
    0 \leq \lambda_k^1 \perp \mathbf{x}^\top(A^{\xi_k} -  A)\mathbf{y} - z_k^1 \geq 0\\
    0 \leq \gamma_s \mathbb{P}(\xi_k) - \lambda^1_k \perp -z_k^1 \geq 0\\
    0 \leq \alpha^1 - v_{l_1}^1(\mathbf{y}) \perp x(a_{l_1}) \geq 0\\
    \mathbf{1}^\top\mathbf{x} = 1
\end{gathered}
\qquad
\begin{gathered}
    0 \leq \lambda_k^2 \perp \mathbf{x}^\top(B^{\xi_k} -  B)\mathbf{y} - z_k^2 \geq 0,\\
    0 \leq \gamma_s \mathbb{P}(\xi_k) - \lambda^2_k \perp -z_k^2 \geq 0,\\
    0 \leq \alpha^2 - v_{l_2}^2(\mathbf{x}) \perp y(a_{l_2}) \geq 0,\\
    \mathbf{1}^\top\mathbf{y} = 1,
\end{gathered}
\qquad
\begin{gathered}
    \forall k \in \{1, \dots, K\}\\
    \forall k \in \{1, \dots, K\}\\
    \forall l_i \in \{1, \dots, n_i\}\\
    \alpha_i \in \mathbb{R}
\end{gathered}
\end{equation}
where $v^1_{l_1} = A_{l_1}\mathbf{y} + \sum_{k=1}^K \lambda^1_k\left(A^{\xi_k}_{l_1} - A_{l_1}\right)\mathbf{y} \leq \alpha^1$ and likewise for player 2. Note that $A_{l_1}$ denotes the $l_1$ row of $A$. Assuming that $A^{\xi_k}, B^{\xi_k}$ are between 0 and 1, we define $\tilde{\mathbf{x}} = \mathbf{x}/\alpha^2, \tilde{\mathbf{y}} = \mathbf{y}/\alpha^1, \tilde{z}^i_k = z_{i,k}/(\alpha^1\alpha^2),\; i \in \{1, 2\}$, using which we obtain:
\begin{equation}\label{opt::twoPlayerFinal}
\begin{gathered}
    0 \leq \lambda_k^1 \perp \tilde{\mathbf{x}}^\top(A^{\xi_k} -  A)\tilde{\mathbf{y}} - \tilde{z}_k^1 \geq 0\\
    0 \leq \gamma_s \mathbb{P}(\xi_k) - \lambda^1_k \perp -\tilde{z}_k^1 \geq 0\\
    0 \leq 1 - \tilde{v}_{l_1}^1(\tilde{\mathbf{y}}) \perp \tilde{x}(a_{l_1}) \geq 0
\end{gathered}
\qquad
\begin{gathered}
    0 \leq \lambda_k^2 \perp \tilde{\mathbf{x}}^\top(B^{\xi_k} -  B)\tilde{\mathbf{y}} - \tilde{z}_k^2 \geq 0,\\
    0 \leq \gamma_s \mathbb{P}(\xi_k) - \lambda^2_k \perp -\tilde{z}_k^2 \geq 0,\\
    0 \leq 1 - \tilde{v}_{l_2}^2(\tilde{\mathbf{x}}) \perp \tilde{y}(a_{l_2}) \geq 0
\end{gathered}
\qquad
\begin{gathered}
    \forall k \in \{1, \dots, K\}\\
    \forall k \in \{1, \dots, K\}\\
    \forall l_i \in \{1, \dots, n_i\}
\end{gathered}
\end{equation}
where a non-zero solution $(\tilde{\mathbf{x}}, \tilde{\mathbf{y}})$ of $\eqref{opt::twoPlayerFinal}$ gives a solution $(\mathbf{x}, \mathbf{y}, \alpha^1, \alpha^2) = (\tilde{\mathbf{x}}(\mathbf{1}^\top \tilde{\mathbf{x}})^{-1}, \tilde{\mathbf{y}}(\mathbf{1}^\top \tilde{\mathbf{y}})^{-1}, (\mathbf{1}^\top \tilde{\mathbf{y}})^{-1},$ $(\mathbf{1}^\top \tilde{\mathbf{x}})^{-1})$ of \eqref{opt::twoPlayer}.
\end{proof} 

\subsubsection{Details of CVaR Complementarity Program}
We prove the result of Lemma~\ref{lem:CVaRMLCP_main} in the main paper by splitting it into two lemmas below.

\begin{lemma} \label{lem:CVaRMCP}
    Consider the following set of constraints for player $i$
\begin{subequations}
\begin{align}
& 0 \leq (\frac{\gamma_c}{\alpha}\mathbb{P}({\xi_k}) - \lambda_{i,k}) \perp (-\nu_{i,k}) \geq 0 \quad \forall k \in \{1, \dots, K \} \\
& 0 \leq \lambda_{i,k} \perp (\sum_{l=1}^{n_i}x_i(a_l)u_i(a_l, \mathbf{x}_{-i} \mid \xi_k) - z_i - \nu_{i,k}) \geq 0 \quad \forall k \in \{1, \dots, K \} \\
& 0 \leq (\alpha_i - v_{i,l}(\mathbf{x}_{-i})) \perp x_i(a_l) \geq 0, \quad  \forall l \in \{1, \dots, n_i\} \\
& \sum_{k=1}^K\lambda_{i,k} = \gamma_c, \quad \mathbf{1}^\top\mathbf{x}_i = 1, \quad  \alpha_i \in \mathbb{R}
\end{align}
\end{subequations}
Stack the above for all players to get a mixed complementarity program whose solutions correspond to DREs of the CVaR game. 
Also, $\sum_l^{n_i} x_i(a_l)v_{i,l}(\mathbf{x}_{-i}) = \alpha_i$, that is, $\alpha_i$ is the value of the game for player $i$. 
\end{lemma}

The above is a \emph{mixed} complementarity program with \emph{two equality constraints}, which is different from mean-semideviation and standard games that have only the mixed strategy probability equality. However, we can still obtain a multilinear complementarity program in the standard form without any equality constraint. This construction is somewhat involved and presented below:
\begin{lemma} \label{lem:CVaRMLCP}
Consider a new variable $t_i$, the following definition of $\tau ,\Theta$ and substitutions for $\mathbf{x}_i, \mathbf{\nu}_i, z_i,\lambda_{i,k}$
\begin{align*}
    & \tau \coloneq \left( \prod_{j = 1}^m \alpha_j \right)^{\frac{1}{m-1}} \;
    \Theta \coloneq \left( \prod_{j = 1}^m t_j \right)^{\frac{1}{m-1}} \; 
    & \tilde{\mathbf{x}}_i \coloneq \frac{\alpha_i t_i}{\tau\Theta}\mathbf{x_i} , \; \forall i  \quad
    \tilde{\mathbf{\nu}}_i \coloneq \frac{\mathbf{\nu}_i}{\tau\Theta} , \; \forall i  \quad
    \tilde{z}_i \coloneq \frac{z_i}{\tau\Theta}, \quad \tilde{\lambda}_{i,k} \coloneq t_i \lambda_{i,k} \; \forall i,k.   
\end{align*}
Consider the MLCP formed by stacking the constraints below for all players
\begin{subequations} 
\begin{align}
& 0 \leq (\frac{\gamma_c}{\alpha}\mathbb{P}({\xi_k}) t_i - \tilde{\lambda}_{i,k}) \perp (-\tilde{\nu}_{i,k}) \geq 0 \quad \forall k \in \{1, \dots, K \} \\
& 0 \leq \tilde{\lambda}_{i,k} \perp (\sum_{l=1}^{n_i}\tilde{x}_i(a_l)u_i(a_l, \tilde{\mathbf{x}}_{-i} \mid \xi_k) - \tilde{z}_i - \tilde{\nu}_{i,k}) \geq 0 \quad \forall k \in \{1, \dots, K \} \\
& 0 \leq \Big(1 - t_i (1 - \gamma_c)\mu^{\mathbb{P}}_i(a_l, \tilde{\mathbf{x}}_{-i}) - \sum_{k=1}^K\tilde{\lambda}_{i,k} u_i(a_l, \tilde{\mathbf{x}}_{-i} \mid \xi_k) \Big) \perp \tilde{x}_i(a_l) \geq 0, \quad  \forall l \in \{1, \dots, n_i\}  \\
& 0 \leq t_i
    \;\perp\;
    (\gamma_c t_i - \sum_{k=1}^K \tilde{\lambda}_{i,k})
    \geq 0.
\end{align}
\end{subequations}
    Any non-zero $\tilde{\mathbf{x}}_i$ (forall $i$) solution of the above program gives a solution to the program of Lemma~\ref{lem:CVaRMCP}, 
with the mixed strategies and game values retrieved by
\[
    \mathbf{x}_i = \frac{\tilde{\mathbf{x}}_i}{\sum_{l=1}^{n_i} \tilde{x}_i(a_l)} \qquad
    \Theta = \left( \prod_{j = 1}^m t_j \right)^{\frac{1}{m-1}} \qquad
    \tau\Theta = \frac{1}{\prod_{i=1}^m\sum_{l=1}^{n_i} \tilde{x}_i(a_l)} \qquad
    \alpha_i = \frac{\tau\Theta}{t_i}\sum_{l=1}^{n_i} \tilde{x}_i(a_l)
\]
\end{lemma}

\paragraph{Proof of Lemma~\ref{lem:CVaRMCP}}
\begin{proof}
The optimality conditions for the linear program $\eqref{eq:CVARLP}$ yield dual multipliers $\lambda_{i,k}, \mu_{i,k}, \alpha_i$ and $\beta_{i,l}$ for constraints $\eqref{cVarOpt:cons1}, \eqref{cVarOpt:cons2}, \eqref{cVarOpt:cons3}$, and $\eqref{cVarOpt:cons4}$, respectively. The optimality conditions are as follows:
\begin{subequations}
\begin{align}
& \underbrace{(1 - \gamma_c)\mu^{\mathbb{P}}_i(a_l, \mathbf{x}_{-i}) + \sum_{k=1}^K\lambda_{i,k} u_i(a_l, \mathbf{x}_{-i} \mid \xi_k) }_{v_{i,l}(\mathbf{x}_{-i})} - \alpha_i + \beta_{i,l} = 0, \quad \forall l \in \{1,\dots,n_i \} \label{cVarMC:value} \\
& \gamma_c - \sum_{k=1}^K\lambda_{i,k} = 0 \label{cVarMC:cVar}\\
& \frac{\gamma_c}{\alpha}\mathbb{P}({\xi_k}) - \lambda_{i,k} - \mu_{i,k} = 0 \quad \forall k \in \{1, \dots, K \} \label{cVarMC:Multiplier}\\ 
& 0 \leq \mu_{i,k} \perp (-\nu_{i,k}) \geq 0 \quad \forall k \in \{1, \dots, K \} \label{cVarMC:Nonlinear1}\\
& 0 \leq \lambda_{i,k} \perp (\sum_{l=1}^{n_i}x_i(a_l)u_i(a_l, \mathbf{x}_{-i} \mid \xi_k) - z_i - \nu_{i,k}) \geq 0 \quad \forall k \in \{1, \dots, K \} \label{cVarMC:Nonelinear2}\\
& 0 \leq \beta_{i,l} \perp x_i(a_l) \geq 0, \quad  \forall l \in \{1, \dots, n_i\}  \\
& \mathbf{1}^\top\mathbf{x}_i = 1, \quad  \alpha_i \in \mathbb{R}
\end{align}
\end{subequations}

Simplifying equations \eqref{cVarMC:value}, \eqref{cVarMC:cVar}, and \eqref{cVarMC:Multiplier}, we obtain the mixed complementarity program for the CVaR game for player $i$. We can stack these for all players to solve the game.

\begin{subequations}\label{eq:CVaRMCP}
\begin{align}
&   \sum_{k=1}^K\lambda_{i,k} = \gamma_c \\
& 0 \leq (\frac{\gamma_c}{\alpha}\mathbb{P}({\xi_k}) - \lambda_{i,k}) \perp (-\nu_{i,k}) \geq 0 \quad \forall k \in \{1, \dots, K \} \label{eq:CVaROrig2}\\
& 0 \leq \lambda_{i,k} \perp (\sum_{l=1}^{n_i}x_i(a_l)u_i(a_l, \mathbf{x}_{-i} \mid \xi_k) - z_i - \nu_{i,k}) \geq 0 \quad \forall k \in \{1, \dots, K \} \label{eq:CVaROrig3}\\
& 0 \leq (\alpha_i - v_{i,l}(\mathbf{x}_{-i})) \perp x_i(a_l) \geq 0, \quad  \forall l \in \{1, \dots, n_i\} \label{eq:CVaROrig4} \\
& \mathbf{1}^\top\mathbf{x}_i = 1, \quad  \alpha_i \in \mathbb{R}
\end{align}
\end{subequations}

Eq.~\eqref{eq:CVaROrig4} implies:
\begin{align}
    (\alpha_i - v_{i,l}(\mathbf{x}_{-i}))x_i(a_l) = 0 \implies \sum_l^{n_i} x_i(a_l)v_{i,l}(\mathbf{x}_{-i}) = \alpha_i\label{valueConst2CVaR}
\end{align}

In other words, $\alpha_i$ is the value of the game for player $i$. Note that at equilibrium $v_{i,l}(\mathbf{x}_{-i})$ also depends on $\mathbf x_i$ via the dual variables $\lambda_{i,k}$.
\end{proof}

\paragraph{Proof of Lemma~\ref{lem:CVaRMLCP}}
\begin{proof}
First, we show that $\alpha_i \geq 0$.
 We know that $\lambda_{i,k} \geq 0$.  Since $0 \leq u^i(\mathbf{a}\mid\xi) \leq 1,\forall \mathbf{a}$, we have $ \mu^{\mathbb{P}}_i(a_l, \mathbf{x}_{-i} )$ is between $0$ and $1$. Then, from Eq~\ref{cVarMC:value} we can see that $0 \leq v_{i,l}(\mathbf{x}_{-i}) $ for all $l$. By Eq.~\ref{valueConst2CVaR} $\alpha_i$ is a convex combination of $\{ v_{i,l}(\mathbf{x}_{-i}) \}_{l=1}^{n_i}$ and hence, $0 \leq \alpha_i $.

The mixed complementarity program \eqref{eq:CVaRMCP} in Lemma~\ref{lem:CVaRMCP} has two equality constraints, which we need to eliminate. We start with a non-zero solution of the MLCP in the lemma statement. Observe that if $t_i = 0$ then $\sum_k \tilde{\lambda}_{i,k} \leq 0$, which implies $\tilde{\lambda}_{i,k} = 0$ for all $k$. Then, from Eq.~\ref{eq:xzero} $\tilde{x}_i(a_l) = 0 $ for all $l$, which is a zero solution. From our assumption, we are in the non-zero solution space (and one such exists from equilibrium existence) and thus, we focus on $t_i > 0$, which implies $\sum_k \tilde{\lambda}_{i,k} = t_i \gamma_c$, which implies $\sum_k \lambda_{i,k} = \gamma_c$ recovering one of the equality in the original mixed complementarity program.

Then,
along with the variables $t_i$ for all $i$, 
we define these variables 
\begin{align}
    & \tau \coloneq \left( \prod_{j = 1}^m \alpha_j \right)^{\frac{1}{m-1}} \quad
    \Theta \coloneq \left( \prod_{j = 1}^m t_j \right)^{\frac{1}{m-1}} \; 
     \tilde{\mathbf{x}}_i \coloneq \frac{\alpha_i t_i}{\tau\Theta}\mathbf{x_i} , \; \forall i  \quad
    \tilde{\mathbf{\nu}}_i \coloneq \frac{\mathbf{\nu}_i}{\tau\Theta} , \; \forall i  \quad
    \tilde{z}_i \coloneq \frac{z_i}{\tau\Theta}, \quad \tilde{\lambda}_{i,k} \coloneq t_i \lambda_{i,k} \; \forall i,k.   \label{eq:scaleCVaR}
\end{align}

With this, we observe the following fact: 
\begin{align*}
    \sum_{l=1}^{n_i}x_i(a_l) u_i(a_l, \mathbf{x}_{-i} \mid \xi_k) & = \sum_{l=1}^{n_i}x_i(a_l) \sum_{\mathbf{a}_{-i} \in \mathbf{A}_{-i}} u_i(a_l, \mathbf{a}_{-i} \mid \xi_k) \prod_{j=1, j \neq i}^m x_j (a_{-i, j})\\
    & = \sum_{l=1}^{n_i}x_i(a_l) \sum_{\mathbf{a}_{-i} \in \mathbf{A}_{-i}} u_i(a_l, \mathbf{a}_{-i} \mid \xi_k) (\tau \Theta)^{m-1} \prod_{j=1, j \neq i}^m \tilde{x}_j (a_{-i, j}) / (\alpha_j t_j)\\
    & = \sum_{l=1}^{n_i}x_i(a_l) \sum_{\mathbf{a}_{-i} \in \mathbf{A}_{-i}} u_i(a_l, \mathbf{a}_{-i} \mid \xi_k)  \prod_{j = 1}^m \alpha_j t_j \prod_{j=1, j \neq i}^m \tilde{x}_j (a_{-i, j}) / (\alpha_j t_j)\\
    & = \sum_{l=1}^{n_i}x_i(a_l) \sum_{\mathbf{a}_{-i} \in \mathbf{A}_{-i}} u_i(a_l, \mathbf{a}_{-i} \mid \xi_k)  \alpha_i t_i \prod_{j=1, j \neq i}^m \tilde{x}_j (a_{-i, j}) \\
    & = \sum_{l=1}^{n_i} \alpha_i t_i x_i(a_l) \sum_{\mathbf{a}_{-i} \in \mathbf{A}_{-i}} u_i(a_l, \mathbf{a}_{-i} \mid \xi_k)   \prod_{j=1, j \neq i}^m \tilde{x}_j (a_{-i, j}) \\
    & = \tau \Theta \sum_{l=1}^{n_i}\tilde{x}_i(a_l) u_i(a_l, \tilde{\mathbf{x}}_{-i} \mid \xi_k)
\end{align*}
Note that complementarity constraints remain the same when multiplied by any constant. We use this and multiply that constraint in Eq.~\ref{eq:CVaRMLCP2} (right one) by $\tau \Theta$ and multiply the left side of Eq.~\ref{eq:CVaRMLCP2} by $1/t_i$ to obtain the MCP constraint \eqref{eq:CVaROrig3}. Similarly,  we can multiply left side of Eq.~\ref{eq:CVaRMLCP1} by $1/t_i$ and right side by $\tau \Theta$ to obtain MCP constraint~\eqref{eq:CVaROrig2}. The final constraint in \eqref{eq:CVaRMLCP} is an extra constraint that enforces that non-zero $t_i$ implies $\sum_k \tilde{\lambda}_{i,k} = t_i \gamma_c$, which in turn implies $\sum_k \lambda_{i,k} = \gamma_c$. We show that the third constraint Eq.~\ref{eq:xzero} is also a scaled version of Eq.~\ref{eq:CVaROrig4}.



For Eq.~\ref{eq:CVaROrig4}, we show that
\begin{align*}
v_{i,l}(\mathbf{x}_{-i}) & = (1 - \gamma_c) \sum_{k=1}^K \mathbb{P}(\xi_k) u_i(a_l, \mathbf{x}_{-i}\mid \xi_k) + \sum_{k=1}^K\lambda_{i,k} u_i(a_l, \mathbf{x}_{-i} \mid \xi_k)  \\
& =\sum_{k=1}^K \big((1 - \gamma_c)  \mathbb{P}(\xi_k) + \lambda_{i,k} \big ) u_i(a_l, \mathbf{x}_{-i} \mid \xi_k)\\
    & = \sum_{k=1}^K \big((1 - \gamma_c)  \mathbb{P}(\xi_k) + \lambda_{i,k} \big ) \sum_{\mathbf{a}_{-i} \in \mathbf{A}_{-i}} u_i(a_l, \mathbf{a}_{-i} \mid \xi_k) \prod_{j=1, j \neq i}^m x_j (a_{-i, j})\\
    & = \sum_{k=1}^K \big((1 - \gamma_c)  \mathbb{P}(\xi_k) + \lambda_{i,k} \big ) \sum_{\mathbf{a}_{-i} \in \mathbf{A}_{-i}} u_i(a_l, \mathbf{a}_{-i} \mid \xi_k) (\tau \Theta)^{m-1} \prod_{j=1, j \neq i}^m \tilde{x}_j (a_{-i, j}) / (\alpha_j t_j)\\
    & = \sum_{k=1}^K \big((1 - \gamma_c)  \mathbb{P}(\xi_k) + \lambda_{i,k} \big ) \sum_{\mathbf{a}_{-i} \in \mathbf{A}_{-i}} u_i(a_l, \mathbf{a}_{-i} \mid \xi_k)  \prod_{j = 1}^m \alpha_j t_j \prod_{j=1, j \neq i}^m \tilde{x}_j (a_{-i, j}) / (\alpha_j t_j)\\
    & = \sum_{k=1}^K \big((1 - \gamma_c)  \mathbb{P}(\xi_k) + \lambda_{i,k} \big ) \sum_{\mathbf{a}_{-i} \in \mathbf{A}_{-i}} u_i(a_l, \mathbf{a}_{-i} \mid \xi_k)  \alpha_i t_i \prod_{j=1, j \neq i}^m \tilde{x}_j (a_{-i, j}) \\
    & = \sum_{k=1}^K \alpha_i t_i \big((1 - \gamma_c)  \mathbb{P}(\xi_k) + \lambda_{i,k} \big ) \sum_{\mathbf{a}_{-i} \in \mathbf{A}_{-i}} u_i(a_l, \mathbf{a}_{-i} \mid \xi_k)   \prod_{j=1, j \neq i}^m \tilde{x}_j (a_{-i, j}) \\
    & = \sum_{k=1}^K \alpha_i t_i \big((1 - \gamma_c)  \mathbb{P}(\xi_k) + \lambda_{i,k} \big ) u_i(a_l, \tilde{\mathbf{x}}_{-i} \mid \xi_k) \\
    & = \alpha_i \Big( \sum_{k=1}^K  t_i (1 - \gamma_c)  \mathbb{P}(\xi_k) u_i(a_l, \tilde{\mathbf{x}}_{-i} \mid \xi_k) + \sum_{k=1}^K t_i \lambda_{i,k}  u_i(a_l, \tilde{\mathbf{x}}_{-i} \mid \xi_k) \Big) \\
    & = \alpha_i \Big( \sum_{k=1}^K  t_i (1 - \gamma_c)  \mathbb{P}(\xi_k) u_i(a_l, \tilde{\mathbf{x}}_{-i} \mid \xi_k) + \sum_{k=1}^K \tilde{\lambda}_{i,k}  u_i(a_l, \tilde{\mathbf{x}}_{-i} \mid \xi_k) \Big )\\
    & = \alpha_i \Big( t_i (1 - \gamma_c)  \mu_i^{\mathbb{P}}(a_l, \tilde{\mathbf{x}}_{-i}) + \sum_{k=1}^K \tilde{\lambda}_{i,k}  u_i(a_l, \tilde{\mathbf{x}}_{-i} \mid \xi_k) \Big ) 
\end{align*}
Multiplying the left side of Eq.~\ref{eq:xzero} by $\alpha_i$ gives Eq.~\ref{eq:CVaROrig4}.

\begin{subequations} \label{eq:CVaRMLCP}
\begin{align}
& 0 \leq (\frac{\gamma_c}{\alpha}\mathbb{P}({\xi_k}) t_i - \tilde{\lambda}_{i,k}) \perp (-\tilde{\nu}_{i,k}) \geq 0 \quad \forall k \in \{1, \dots, K \} \label{eq:CVaRMLCP1} \\
& 0 \leq \tilde{\lambda}_{i,k} \perp (\sum_{l=1}^{n_i}\tilde{x}_i(a_l)u_i(a_l, \tilde{\mathbf{x}}_{-i} \mid \xi_k) - \tilde{z}_i - \tilde{\nu}_{i,k}) \geq 0 \quad \forall k \in \{1, \dots, K \} \label{eq:CVaRMLCP2} \\
& 0 \leq \Big(1 - t_i (1 - \gamma_c)\mu^{\mathbb{P}}_i(a_l, \tilde{\mathbf{x}}_{-i}) - \sum_{k=1}^K\tilde{\lambda}_{i,k} u_i(a_l, \tilde{\mathbf{x}}_{-i} \mid \xi_k) \Big) \perp \tilde{x}_i(a_l) \geq 0, \quad  \forall l \in \{1, \dots, n_i\}  \label{eq:xzero}\\
& 0 \leq t_i
    \;\perp\;
    (\gamma_c t_i - \sum_{k=1}^K \tilde{\lambda}_{i,k})
    \geq 0.
\end{align}
\end{subequations}
Then, the scaling below allows us to claim that non-zero solutions also satisfy the original mixed complementarity program \eqref{eq:CVaRMCP}:
\[
    \mathbf{x}_i = \frac{\tilde{\mathbf{x}}_i}{\sum_{l=1}^{n_i} \tilde{x}_i(a_l)} \qquad
    \Theta = \left( \prod_{j = 1}^m t_j \right)^{\frac{1}{m-1}} \qquad
    \tau\Theta = \frac{1}{\prod_{i=1}^m\sum_{l=1}^{n_i} \tilde{x}_i(a_l)} \qquad
    \alpha_i = \frac{\tau\Theta}{t_i}\sum_{l=1}^{n_i} \tilde{x}_i(a_l)
\]

Note that from above, we obtain $\sum_{l=1}^{n_i}x_i(a_l) = 1$, further $ (\prod_j \alpha_j )^{(1/m-1)} = \tau \Theta / (\prod_j t_j )^{(1/m-1)} = \tau \Theta / \Theta = \tau$. It can be easily checked that the original equations of  \eqref{eq:CVaRMCP} are all satisfied by scaling the other tilde variables $\tilde{\nu}, \tilde{z}, \tilde{\lambda}$ back to the original variables according to Eq.~\ref{eq:scaleCVaR}.
\end{proof}

\subsection{Details and Proofs in Section~\ref*{subsec:scalable-action-cumg}~\nameref*{subsec:scalable-action-cumg}}



\subsubsection{Restricted-profile Gap for MSD and CVaR Games}
\label{subsec:restrictedgap}

\paragraph{Restricted-profile gap for MSD} Recall that, for MSD,
$
\mu_i^{\mathbb P}(\mathbf x)=\sum_{k=1}^K \mathbb P(\xi_k)u_i(\mathbf x\mid \xi_k)
$
and
$
d_{i,k}(\mathbf x_i,\mathbf x_{-i})
=
u_i(\mathbf x\mid \xi_k)-\mu_i^{\mathbb P}(\mathbf x),
$
so that $z_{i,k}$ represents the negative part $\min\{0,d_{i,k}(\mathbf x_i,\mathbf x_{-i})\}$.

For a candidate action support profile $S=(S_1,\ldots,S_m)$ and data-sample support profile $T=(T_1,\ldots,T_m)$, define
\[
\begin{aligned}
\widehat\eta^{\mathrm{MSD}}_{S,T}
:=
\min_{\mathbf x,\mathbf q,\mathbf z,\eta}\quad
& \eta\\
\text{s.t.}\quad
& \mathbf x_i\in\Delta(\mathbf A_i),\qquad
x_i(a_l)=0 \quad \forall a_l\notin S_i,
&& \forall i,\\
& \mathbf q_i\in\Delta_K,\qquad
q_{i,k}=0 \quad \forall k\notin T_i,
&& \forall i,\\
& \sum_{k=1}^K q_{i,k}u_i(a_l,\mathbf x_{-i}\mid \xi_k) \le \mu_i^{\mathbb P}(\mathbf x) +
\gamma_s\sum_{k=1}^K \mathbb P(\xi_k)z_{i,k}
+\eta,
&& \forall i,\ \forall a_l\in\mathbf A_i ,\\
& z_{i,k}\le d_{i,k}(\mathbf x_i,\mathbf x_{-i}),\qquad
z_{i,k}\le 0,
&& \forall i,\ \forall k.
\end{aligned}
\]


\paragraph{Restricted-profile gap for CVaR}: Recall that, for CVaR, the robust utility is
\[
\rho_i^{\mathrm{CVaR}}(\mathbf x)
=
(1-\gamma_c)\mu_i^{\mathbb P}(\mathbf x) + \gamma_c
\max_{z_i\in\mathbb R}
\left[
z_i+\frac{1}{\alpha}\sum_{k=1}^K \mathbb P(\xi_k)\nu_{i,k}
\right],
\]
where $\nu_{i,k}$ represents
$
\min\{0,u_i(\mathbf x\mid \xi_k)-z_i\}
$
through the constraints
$
\nu_{i,k}\le u_i(\mathbf x\mid \xi_k)-z_i
$
and
$
\nu_{i,k}\le 0.
$
For a candidate action support profile $S=(S_1,\ldots,S_m)$ and data-sample support profile $T=(T_1,\ldots,T_m)$, define
\[
\begin{aligned}
\widehat\eta^{\mathrm{CVaR}}_{S,T}
:=
\min_{\mathbf x,\mathbf q,\mathbf z,\boldsymbol\nu,\eta}\quad
& \eta\\
\text{s.t.}\quad
& \mathbf x_i\in\Delta(\mathbf A_i),\qquad
x_i(a_l)=0 \quad \forall a_l\notin S_i,
&& \forall i,\\
& \mathbf q_i\in\Delta_K,\qquad
q_{i,k}=0 \quad \forall k\notin T_i,
&& \forall i,\\
& \nu_{i,k}\le u_i(\mathbf x\mid \xi_k)-z_i,\qquad
\nu_{i,k}\le 0,
&& \forall i,\ \forall k,\\
& \sum_{k=1}^K q_{i,k}u_i(a_l,\mathbf x_{-i}\mid \xi_k)
\le
(1-\gamma_c)\mu_i^{\mathbb P}(\mathbf x)
+
\gamma_c
\left[
z_i+\frac{1}{\alpha}\sum_{k=1}^K\mathbb P(\xi_k)\nu_{i,k}
\right]
+\eta,
&& \forall i,\ \forall a_l\in\mathbf A_i .
\end{aligned}
\]

\subsection{Details and Proofs in Section~\ref*{subsec:scalable-samples-cumg}~\nameref*{subsec:scalable-samples-cumg}}
\begin{proposition}[Interior strategy maximizers under entropy regularization]
\label{prop:entropy-interiority}
Fix a player $i$, $\kappa>0$, $\tau>0$, 
$\mathbf{x}_{-i}\in \prod_{j\neq i}\mathbf X_j$, and 
$\boldsymbol{\theta}_i\in\Theta_i$. Suppose that the mapping
$
\mathbf{x}_i \mapsto 
\rho_{i}^{\mathrm{aux},\tau}
(\mathbf{x}_i,\mathbf{x}_{-i},\boldsymbol{\theta}_i)
$
has bounded gradients on $\mathbf X_i$. Then every maximizer of
$
\max_{\mathbf{x}_i\in\mathbf X_i}
\left\{
\rho_{i}^{\mathrm{aux},\tau}
(\mathbf{x}_i,\mathbf{x}_{-i},\boldsymbol{\theta}_i)
+
\kappa H_i(\mathbf{x}_i)
\right\}
$
lies in the relative interior of $\mathbf X_i$. Equivalently, if
$\mathbf{x}_i^\star$ is a maximizer, then
$
x_i^\star(a_l)>0,\; \forall l\in\{1,\dots,n_i\}.
$
\end{proposition}

\begin{proof}
Let
$
f_i(\mathbf{x}_i)
\coloneqq
\rho_{i}^{\mathrm{aux},\tau}
(\mathbf{x}_i,\mathbf{x}_{-i},\boldsymbol{\theta}_i).
$
Since $f_i$ has bounded gradients on $\mathbf X_i$, it is Lipschitz on
$\mathbf X_i$. Hence, there exists $L_i<\infty$ such that, for all
$\mathbf{x}_i,\mathbf{y}_i\in\mathbf X_i$,
$
|f_i(\mathbf{y}_i)-f_i(\mathbf{x}_i)|
\leq L_i\|\mathbf{y}_i-\mathbf{x}_i\|_1.
$

Suppose, for contradiction, that a maximizer
$\mathbf{x}_i^\star$ lies on the boundary of $\mathbf X_i$. Then there exists
an action $a_r$ such that $x_i^\star(a_r)=0$. Since
$\mathbf{x}_i^\star\in\mathbf X_i$, there must also exist an action $a_s$
such that $x_i^\star(a_s)>0$. For sufficiently small
$\epsilon\in(0,x_i^\star(a_s))$, define the feasible perturbation
$
\mathbf{x}_i^\epsilon = \mathbf{x}_i^\star + \epsilon \mathbf e_r - \epsilon \mathbf e_s,
$
where $\mathbf e_r$ and $\mathbf e_s$ are the unit vectors corresponding to
actions $a_r$ and $a_s$.
By Lipschitzness,
$
f_i(\mathbf{x}_i^\epsilon)-f_i(\mathbf{x}_i^\star)
\geq
- L_i\|\mathbf{x}_i^\epsilon-\mathbf{x}_i^\star\|_1
=
-2L_i\epsilon.
$
On the other hand, using the convention $0\log 0=0$, the entropy gain from
moving $\epsilon$ probability mass to the previously unused action satisfies
\[
H_i(\mathbf{x}_i^\epsilon)-H_i(\mathbf{x}_i^\star)
=
-\epsilon\log\epsilon
-
\big(x_i^\star(a_s)-\epsilon\big)
\log\big(x_i^\star(a_s)-\epsilon\big)
+
x_i^\star(a_s)\log x_i^\star(a_s).
\]
As $\epsilon$ goes to $0$,
$
H_i(\mathbf{x}_i^\epsilon)-H_i(\mathbf{x}_i^\star)
=
\epsilon\log(1/\epsilon)+O(\epsilon).
$
Therefore,
\[
\begin{aligned}
&
\left[ f_i(\mathbf{x}_i^\epsilon)
+ \kappa H_i(\mathbf{x}_i^\epsilon) \right] - \left[ f_i(\mathbf{x}_i^\star) + \kappa H_i(\mathbf{x}_i^\star)
\right] \geq
-2L_i\epsilon +
\kappa\epsilon\log(1/\epsilon) + O(\epsilon).
\end{aligned}
\]
For sufficiently small $\epsilon>0$, the positive term
$\kappa\epsilon\log(1/\epsilon)$ dominates the linear terms in $\epsilon$.
Hence the entropy-regularized objective strictly increases at
$\mathbf{x}_i^\epsilon$, contradicting the optimality of
$\mathbf{x}_i^\star$.
Thus, no maximizer can assign zero probability to any action, and every
maximizer lies in the relative interior of $\mathbf X_i$.
\end{proof}

\subsubsection{Instantiation of Lemma~\ref{lem:smoothedResidualRegret} for MSD and CVaR}
\label{subsec:instantiation}
\paragraph{Mean-semideviation.}
The MSD utility is
$
\rho_{i,\mathrm{MSD}}(\mathbf x)=\mu_i^{\mathbb P}(\mathbf x)- \gamma_s \sum_{k=1}^K \mathbb P(\xi_k) \bigl(\mu_i^{\mathbb P}(\mathbf x)-u_i(\mathbf{x}_i, \mathbf{x}_{-i}\mid \xi_k) \bigr)_+ .
$
Its smoothed version is
\[
\rho_{i,\mathrm{MSD}}^{\tau}(\mathbf x)=\mu_i^{\mathbb P}(\mathbf x)- \gamma_s \sum_{k=1}^K \mathbb P(\xi_k) \varphi_\tau\bigl(\mu_i^{\mathbb P}(\mathbf x)-u_i(\mathbf{x}_i, \mathbf{x}_{-i}\mid \xi_k) \bigr),
\]
so
$
\sup_{\mathbf x}\left|\rho_{i,\mathrm{MSD}}^{\tau}(\mathbf x)-\rho_{i,\mathrm{MSD}}(\mathbf x)\right|\le \gamma_s c_\varphi\tau.
$
Here $\boldsymbol{\theta}_i$ is absent, so $\Delta_i^{\mathrm{aux},\tau}(\mathbf x)=\sup_{\mathbf y_i\in\mathbf X_i}\left[\rho_{i,\mathrm{MSD}}^\tau(\mathbf y_i,\mathbf x_{-i})-\rho_{i,\mathrm{MSD}}^\tau(\mathbf x)\right]$; by the second part of Lemma~\ref{lem:smoothedResidualRegret}, $\Delta_i^{\mathrm{aux},\tau}(\mathbf x)\le\sqrt{2}\|\mathbf r_i^\tau\|_2+\kappa\log n_i$. Define
$
\beta_{i,k}^{\tau}(\mathbf x)=\varphi_\tau'\bigl(\mu_i^{\mathbb P}(\mathbf x)-u_i(\mathbf{x}_i, \mathbf{x}_{-i}\mid \xi_k) \bigr).
$
Then, the gradient of $\rho_{i,\mathrm{MSD}}^{\tau}$ has components
\[
g_{i,l}^{\tau}(\mathbf x)=\mu_i^{\mathbb P}(a_l,\mathbf x_{-i})+ \gamma_s\sum_{k=1}^K \mathbb P(\xi_k) \beta_{i,k}^{\tau}(\mathbf x)\left[u_i(a_l,\mathbf x_{-i}\mid \xi_k)-\mu_i^{\mathbb P}(a_l,\mathbf x_{-i})\right].
\]
$\mu_i^{\mathbb P}(a_l,\mathbf x_{-i})$ can be computed exactly from the empirical distribution averaged payoff matrix: $\mu_i^{\mathbb P}(\mathbf a)= \sum_{k=1}^K \mathbb P(\xi_k) u_i(\mathbf a\mid \xi_k)$, which needs to be computed only once. For a mini-batch $B$ formed by sampling indices $k\in[K]$ i.i.d. according to
$\mathbb P(\xi_k)$,
\[
\widehat g_{i,l,B}^{\tau}(\mathbf x)=\mu_i^{\mathbb P}(a_l,\mathbf x_{-i})+\frac{\gamma_s}{|B|}\sum_{k\in B}\beta_{i,k}^{\tau}(\mathbf x)\left[u_i(a_l,\mathbf x_{-i}\mid \xi_k)-\mu_i^{\mathbb P}(a_l,\mathbf x_{-i})\right].
\]
Thus \(\mathbb E_B[\widehat{\mathbf g}_{i,B}^{\tau}(\mathbf x)]=\mathbf g_i^\tau(\mathbf x)\). Since the smoothed residual is differentiable and sample-additive,
$
\mathbb E_B[D_{\mathbf z}\widehat{\mathcal R}_{\kappa,\tau,B}(\mathbf z)]=D_{\mathbf z}\mathcal R_{\kappa,\tau}(\mathbf z),
$
and the double-sampling estimator \(\widehat J^{(1)\top}\widehat{\mathcal R}^{(2)}\) is unbiased for \(\nabla M_{\kappa,\tau}\). By Lemma~\ref{lem:smoothedResidualRegret},
\[
\sup_{\mathbf y_i\in\mathbf X_i}\left[\rho_{i,\mathrm{MSD}}(\mathbf y_i,\mathbf x_{-i})-\rho_{i,\mathrm{MSD}}(\mathbf x)\right]\le \sqrt{2}\|\mathbf r_i^\tau\|_2+\kappa\log n_i+2\gamma_s c_\varphi\tau.
\]

\paragraph{CVaR}
The CVaR-based coherent utility is
\[
\rho_{i,\mathrm{CVaR}}(\mathbf x)
=
(1-\gamma_c)\mu_i^{\mathbb P}(\mathbf x)
+
\gamma_c\sup_{z_i\in\mathbb R}
\left[
z_i-\frac{1}{\alpha}
\sum_{k=1}^K
\mathbb P(\xi_k)
\bigl(z_i-u_i(\mathbf x_i,\mathbf x_{-i}\mid\xi_k)\bigr)_+
\right].
\]
Equivalently, define the nonsmoothed auxiliary objective
\[
\rho_{i,\mathrm{CVaR}}^{\mathrm{aux}}(\mathbf x,z_i)
=
(1-\gamma_c)\mu_i^{\mathbb P}(\mathbf x)
+
\gamma_c
\left[
z_i-\frac{1}{\alpha}
\sum_{k=1}^K
\mathbb P(\xi_k)
\bigl(z_i-u_i(\mathbf x_i,\mathbf x_{-i}\mid\xi_k)\bigr)_+
\right],
\]
so that
$\rho_{i,\mathrm{CVaR}}(\mathbf x)
=
\sup_{z_i\in\mathbb R}
\rho_{i,\mathrm{CVaR}}^{\mathrm{aux}}(\mathbf x,z_i)$.
The smoothed auxiliary objective is
\[
\rho_{i,\mathrm{CVaR}}^{\mathrm{aux},\tau}(\mathbf x,z_i)
=
(1-\gamma_c)\mu_i^{\mathbb P}(\mathbf x)
+
\gamma_c
\left[
z_i-\frac{1}{\alpha}
\sum_{k=1}^K
\mathbb P(\xi_k)
\varphi_\tau
\bigl(z_i-u_i(\mathbf x_i,\mathbf x_{-i}\mid\xi_k)\bigr)
\right].
\]
Since
$|\varphi_\tau(s)-s_+|\leq c_\varphi\tau$
uniformly in $s$, the smoothing error satisfies
\[
\sup_{\mathbf x\in\mathbf X,\,z_i\in\mathbb R}
\left|
\rho_{i,\mathrm{CVaR}}^{\mathrm{aux},\tau}(\mathbf x,z_i)
-
\rho_{i,\mathrm{CVaR}}^{\mathrm{aux}}(\mathbf x,z_i)
\right|
\leq
\frac{\gamma_c c_\varphi\tau}{\alpha}.
\]

Define
$
\beta_{i,k}^\tau(\mathbf x,z_i) = \varphi_\tau' \bigl(z_i-u_i(\mathbf x_i,\mathbf x_{-i}\mid\xi_k)\bigr).
$
The strategy-gradient components and the auxiliary derivative are, respectively,
\[
g_{i,l}^\tau(\mathbf x,z_i)
=
(1-\gamma_c)\mu_i^{\mathbb P}(a_l,\mathbf x_{-i})
+
\frac{\gamma_c}{\alpha}
\sum_{k=1}^K
\mathbb P(\xi_k)
\beta_{i,k}^\tau(\mathbf x,z_i)
u_i(a_l,\mathbf x_{-i}\mid\xi_k),
\]
and
\[
h_i^\tau(\mathbf x,z_i)
=
\gamma_c
\left[
1-\frac{1}{\alpha}
\sum_{k=1}^K
\mathbb P(\xi_k)
\beta_{i,k}^\tau(\mathbf x,z_i)
\right].
\]

We next bound the range required for the current smoothed auxiliary variable. Since $\varphi_\tau'$ is increasing and
$u_i(\mathbf x\mid\xi_k)\in[0,1]$,
\[
\varphi_\tau'(z_i-1)
\leq
\sum_{k=1}^K
\mathbb P(\xi_k)
\beta_{i,k}^\tau(\mathbf x,z_i)
\leq
\varphi_\tau'(z_i).
\]
When $\gamma_c>0$, any zero of $h_i^\tau$ satisfies
$
\sum_{k=1}^K
\mathbb P(\xi_k)
\beta_{i,k}^\tau(\mathbf x,z_i) = \alpha.
$
Since
$\varphi_\tau'(s)=1/(1+\exp(-s/\tau))$,
it follows that
\[
\tau\log\left(\frac{\alpha}{1-\alpha}\right)
\leq
z_i
\leq
1+\tau\log\left(\frac{\alpha}{1-\alpha}\right).
\]
Accordingly, for CVaR we choose the feasible set
$\widetilde{\mathcal Z}$
in Algorithm~\ref{alg:gempStyleCUMG} so that every iterate satisfies
\[
z_i\in
\left[
\tau\log\left(\frac{\alpha}{1-\alpha}\right),
1+\tau\log\left(\frac{\alpha}{1-\alpha}\right)
\right].
\]
When $\gamma_c>0$, this restriction does not exclude any $z_i$ satisfying the smoothed auxiliary first-order condition $h_i^\tau(\mathbf x,z_i)=\partial\rho_{i,\mathrm{CVaR}}^{\mathrm{aux},\tau}(\mathbf x,z_i)/\partial z_i=0$, for any $\mathbf x\in\mathbf X$. When $\gamma_c=0$, the auxiliary variable is immaterial and the same restriction may be imposed without loss.

For a mini-batch $B$ formed by sampling indices $k\in[K]$ independently according to $\mathbb P(\xi_k)$, define
\[
\widehat g_{i,l,B}^\tau(\mathbf x,z_i)
=
(1-\gamma_c)\mu_i^{\mathbb P}(a_l,\mathbf x_{-i})
+
\frac{\gamma_c}{\alpha|B|}
\sum_{k\in B}
\beta_{i,k}^\tau(\mathbf x,z_i)
u_i(a_l,\mathbf x_{-i}\mid\xi_k),
\]
and
\[
\widehat h_{i,B}^\tau(\mathbf x,z_i)
=
\gamma_c
\left[
1-\frac{1}{\alpha|B|}
\sum_{k\in B}
\beta_{i,k}^\tau(\mathbf x,z_i)
\right].
\]
Since $z_i$ is maintained as an explicit auxiliary variable and is not recomputed as a mini-batch quantile,
\[
\mathbb E_B
\left[
\widehat{\mathbf g}_{i,B}^\tau(\mathbf x,z_i)
\right]
=
\mathbf g_i^\tau(\mathbf x,z_i),
\qquad
\mathbb E_B
\left[
\widehat h_{i,B}^\tau(\mathbf x,z_i)
\right]
=
h_i^\tau(\mathbf x,z_i).
\]
Because the smoothed auxiliary objective is differentiable and the sample sum is finite,
\[
\mathbb E_B
\left[
D_{\mathbf z}
\widehat{\mathcal R}_{\kappa,\tau,B}(\mathbf z)
\right]
=
D_{\mathbf z}
\mathcal R_{\kappa,\tau}(\mathbf z).
\]
Therefore, using independent mini-batches for the Jacobian and residual factors,
$\widehat J^{(1)\top}\widehat{\mathcal R}^{(2)}$
is unbiased for
$\nabla M_{\kappa,\tau}(\mathbf z)$.

For the residual-to-regret certificate, utilities bounded in $[0,1]$ imply that the supremum defining nonsmoothed reduced CVaR may be taken over a comparison variable $\zeta_i\in[0,1]$; hence, $\bar\Theta_i=[0,1]$. This comparison set is distinct from the interval above used for the current smoothed auxiliary variable $z_i$. The smoothed augmented gap is
\[
\Delta_i^{\mathrm{aux},\tau}(\mathbf x,z_i)
=
\sup_{\mathbf y_i\in\mathbf X_i,\,\zeta_i\in[0,1]}
\left[
\rho_{i,\mathrm{CVaR}}^{\mathrm{aux},\tau}
(\mathbf y_i,\mathbf x_{-i},\zeta_i)
-
\rho_{i,\mathrm{CVaR}}^{\mathrm{aux},\tau}
(\mathbf x,z_i)
\right].
\]
Since
$\rho_{i,\mathrm{CVaR}}^{\mathrm{aux},\tau}$
is concave in $(\mathbf x_i,z_i)$ for fixed $\mathbf x_{-i}$, and since
$\rho_i^{\kappa,\tau}
=
\rho_{i,\mathrm{CVaR}}^{\mathrm{aux},\tau}
+\kappa H_i$
has strategy gradient $\mathbf q_i^\tau$, for every
$\mathbf y_i\in\mathbf X_i$
and $\zeta_i\in[0,1]$,
\[
\begin{aligned}
\rho_{i,\mathrm{CVaR}}^{\mathrm{aux},\tau}
(\mathbf y_i,\mathbf x_{-i},\zeta_i) - \rho_{i,\mathrm{CVaR}}^{\mathrm{aux},\tau}
(\mathbf x,z_i) \leq
\left\langle \mathbf q_i^\tau,\mathbf y_i-\mathbf x_i \right\rangle + h_i^\tau(\zeta_i-z_i) + \kappa \left( H_i(\mathbf x_i)-H_i(\mathbf y_i) \right).
\end{aligned}
\]
For every feasible $z_i$ and every $\zeta_i\in[0,1]$,
$
|\zeta_i-z_i|
\leq
1+\tau
\left|
\log\left(\frac{\alpha}{1-\alpha}\right)
\right|.
$
Therefore, using centering, the Euclidean diameter of the simplex, and
$H_i(\mathbf x_i)-H_i(\mathbf y_i)\leq\log n_i$ gives
\[
\Delta_i^{\mathrm{aux},\tau}(\mathbf x,z_i)
\leq
\sqrt{2}\|\mathbf r_i^\tau\|_2
+
\left(
1+\tau
\left|
\log\left(\frac{\alpha}{1-\alpha}\right)
\right|
\right)
|h_i^\tau|
+
\kappa\log n_i.
\]
Consequently, Lemma~\ref{lem:smoothedResidualRegret} gives
\[
\begin{aligned}
\sup_{\mathbf y_i\in\mathbf X_i}
\left[
\rho_{i,\mathrm{CVaR}}(\mathbf y_i,\mathbf x_{-i})
-
\rho_{i,\mathrm{CVaR}}(\mathbf x)
\right]
\leq{}
\sqrt{2}\|\mathbf r_i^\tau\|_2
+
\left(
1+\tau
\left|
\log\left(\frac{\alpha}{1-\alpha}\right)
\right|
\right)
|h_i^\tau| +
\kappa\log n_i +
\frac{2\gamma_c c_\varphi\tau}{\alpha}.
\end{aligned}
\]

\subsection{Details for Section~\ref*{sec:experiments}~\nameref*{sec:experiments}}
\label{sec:experimentdetails}

$\texttt{PATH}$ options were configured as follows
\begin{center}
   \begin{tabular}{llllll}
  major\_iteration\_limit: & 50,000,000 & restart\_limit: & 100 & nms\_memory\_size: & 50 \\
  minor\_iteration\_limit: & 50,000,000 & time\_limit: & 300 & convergence\_tolerance: & 1e-8 \\
  cumulative\_iteration\_limit: & 100,000,000
\end{tabular}
\end{center}

\begin{figure}[t]
   \centering
   \includegraphics[width=1\linewidth]{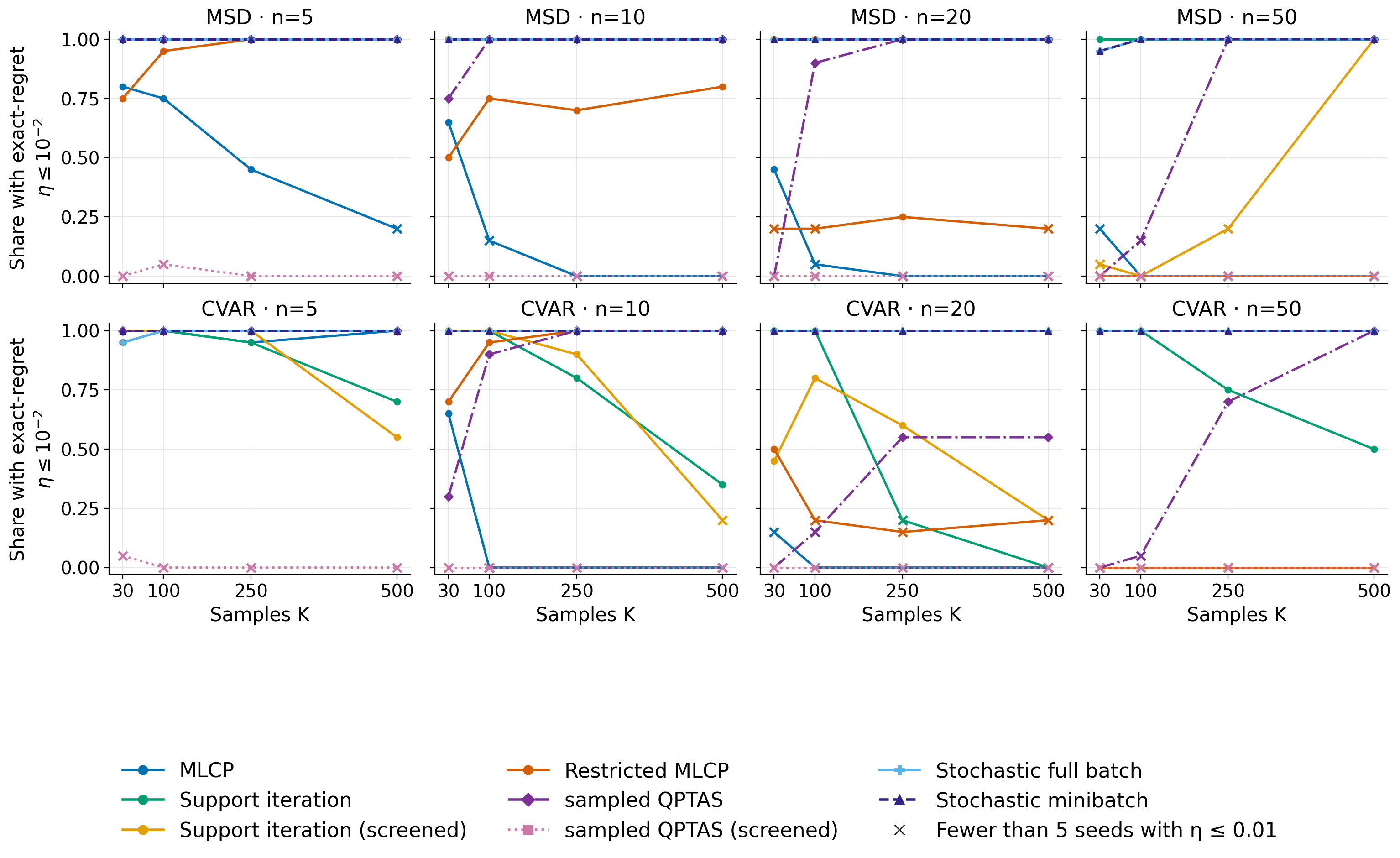}
   \caption{Fraction of 20 games for which the solved approximate equilibrium strategies generated a regret $\leq 0.01$ by different algorithms.}
   \label{fig:algorithmSuccessRates}
\end{figure}

Figure~\ref{fig:algorithmSuccessRates} shows the success rates for the various algorithms in the MSD and CVaR games in the setting of uniformly distributed payoffs, measured by the number of games each of them successfully solved out of 20, as the size of the game increases. Hyperparameters for the stochastic methods were calibrated using pilot runs on different seeds for the heterogenous payoff population games. Through the pilot experiments, entropy weight $\kappa$ was set to 0.01, smoothening parameter $\tau$ was set to 0.002, the learning rate at iteration $t$ was $\frac{1000}{\sqrt{t}}$ for the MSD games and $\frac{500}{\sqrt{t}}$ for the CVaR games. The minibatch size was chosen as $\lceil\sqrt{K}\rceil$, and the iteration budget per random start was set to 2,000.

\section{Dual derivations for specific coherent utility measures}
\label{sec:dualSpecificCUMG}

This appendix derives the complementarity formulations for specific coherent utility measures from their dual representations. The derivations show how the auxiliary variables in the complementarity programs can be interpreted as variables selecting a worst-case dual distribution.

\subsection{Mean-semideviation}
Following the dual representation of mean-semideviation in~\cite{ruszczynski2006optimization}, the order-$p$ mean-semideviation utility admits a risk-envelope representation with density
$\frac{d\mathbb{Q}}{d\mathbb{P}}=1+\gamma_s(\eta-\mathbb{E}^{\mathbb{P}}[\eta])$, where $\eta\geq 0$ and $\|\eta\|_q\leq 1$, with $1/p+1/q=1$. The mean-semideviation used here corresponds to the case $p=1$, so $q=\infty$ and the constraint becomes $0\leq\eta\leq 1$. Hence, for
$\rho_{\mathrm{MSD}}(X)=\mathbb{E}^{\mathbb{P}}[X]-\gamma_s\mathbb{E}^{\mathbb{P}}[\max(0,\mathbb{E}^{\mathbb{P}}[X]-X)]$, the corresponding dual density is
\[
\frac{d\mathbb{Q}}{d\mathbb{P}}
=
1+\gamma_s\left(\eta-\mathbb{E}^{\mathbb{P}}[\eta]\right),
\qquad 0\leq\eta\leq 1.
\]
In the finite-sample setting, let $\eta_k=\eta(\xi_k)$ and define $\lambda_{i,k}=\gamma_s\mathbb{P}(\xi_k)\eta_k$. Then $0\leq\lambda_{i,k}\leq\gamma_s\mathbb{P}(\xi_k)$ and
$
\mathbb{Q}_i(\xi_k)=\mathbb{P}(\xi_k)+\lambda_{i,k}-\mathbb{P}(\xi_k)\sum_{r=1}^K\lambda_{i,r}.
$
Let $X_k=u_i(\mathbf{x}_i,\mathbf{x}_{-i}\mid \xi_k)$ and recall 
$\mu_i^{\mathbb{P}}(\mathbf{x}_i,\mathbf{x}_{-i})=\sum_{k=1}^K\mathbb{P}(\xi_k)X_k$.
Therefore,
\begin{align*}
\mathbb{E}^{\mathbb{Q}_i}[X]
&=
\sum_{k=1}^K
\left(
\mathbb{P}(\xi_k)+\lambda_{i,k}
-
\mathbb{P}(\xi_k)\sum_{r=1}^K\lambda_{i,r}
\right)X_k\\
&=
\sum_{k=1}^K\mathbb{P}(\xi_k)X_k
+
\sum_{k=1}^K\lambda_{i,k}X_k
-
\left(\sum_{r=1}^K\lambda_{i,r}\right)
\sum_{k=1}^K\mathbb{P}(\xi_k)X_k\\
&=
\mu_i^{\mathbb{P}}(\mathbf{x}_i,\mathbf{x}_{-i})
+
\sum_{k=1}^K\lambda_{i,k}X_k
-
\left(\sum_{k=1}^K\lambda_{i,k}\right)
\mu_i^{\mathbb{P}}(\mathbf{x}_i,\mathbf{x}_{-i})\\
&=
\mu_i^{\mathbb{P}}(\mathbf{x}_i,\mathbf{x}_{-i})
+
\sum_{k=1}^K\lambda_{i,k}
\left(
X_k-\mu_i^{\mathbb{P}}(\mathbf{x}_i,\mathbf{x}_{-i})
\right)\\
&=
\mu_i^{\mathbb{P}}(\mathbf{x}_i,\mathbf{x}_{-i})
+
\sum_{k=1}^K\lambda_{i,k}
\left(
u_i(\mathbf{x}_i,\mathbf{x}_{-i}\mid \xi_k)
-
\mu_i^{\mathbb{P}}(\mathbf{x}_i,\mathbf{x}_{-i})
\right).
\end{align*}
Since
$
u_i(\mathbf{x}_i,\mathbf{x}_{-i}\mid \xi_k)
=
\sum_{l=1}^{n_i}x_i(a_l)u_i(a_l,\mathbf{x}_{-i}\mid \xi_k)
$
and
$
\mu_i^{\mathbb{P}}(\mathbf{x}_i,\mathbf{x}_{-i})
=
\sum_{l=1}^{n_i}x_i(a_l)\mu_i^{\mathbb{P}}(a_l,\mathbf{x}_{-i}),
$
we obtain
\[
u_i(\mathbf{x}_i,\mathbf{x}_{-i}\mid \xi_k)
-
\mu_i^{\mathbb{P}}(\mathbf{x}_i,\mathbf{x}_{-i})
=
d_{i,k}(\mathbf{x}_i,\mathbf{x}_{-i}).
\]
Hence,
\begin{equation}
\label{eq:appMSDdual}
\rho_{\text{MSD}}(u_i(\mathbf{x}_i,\mathbf{x}_{-i}\mid \xi))
=
\min_{0\leq \lambda_{i,k}\leq \gamma_s\mathbb{P}(\xi_k)}
\left[
\mu_i^{\mathbb{P}}(\mathbf{x}_i,\mathbf{x}_{-i})
+
\sum_{k=1}^K\lambda_{i,k}
d_{i,k}(\mathbf{x}_i,\mathbf{x}_{-i})
\right].
\end{equation}
For fixed $\mathbf{x}_i$ and $\mathbf{x}_{-i}$, the only variables in the inner dual problem are $\lambda_{i,k}$. Since
\[
\min_{0\leq \lambda_{i,k}\leq \gamma_s\mathbb{P}(\xi_k)}
\left[
\mu_i^{\mathbb{P}}(\mathbf{x}_i,\mathbf{x}_{-i})
+
\sum_{k=1}^K\lambda_{i,k}d_{i,k}(\mathbf{x}_i,\mathbf{x}_{-i})
\right]
\]
has a constant first term with respect to $\lambda_i$, its inner optimality conditions are exactly those of the separable linear problem
$
\min_{0\leq \lambda_{i,k}\leq \gamma_s\mathbb{P}(\xi_k)}
\sum_{k=1}^K\lambda_{i,k}d_{i,k}(\mathbf{x}_i,\mathbf{x}_{-i}).
$
Thus, for each $k$, the optimal $\lambda_{i,k}$ solves
$\min_{0\leq \lambda_{i,k}\leq \gamma_s\mathbb{P}(\xi_k)}\lambda_{i,k}d_{i,k}(\mathbf{x}_i,\mathbf{x}_{-i})$. Hence, if $d_{i,k}(\mathbf{x}_i,\mathbf{x}_{-i})>0$, the minimizer is $\lambda_{i,k}=0$; if $d_{i,k}(\mathbf{x}_i,\mathbf{x}_{-i})<0$, the minimizer is $\lambda_{i,k}=\gamma_s\mathbb{P}(\xi_k)$; and if $d_{i,k}(\mathbf{x}_i,\mathbf{x}_{-i})=0$, any value in the interval is optimal.

Introducing $z_{i,k}=\min(0,d_{i,k}(\mathbf{x}_i,\mathbf{x}_{-i}))$, these three cases can be written equivalently as
\begin{subequations}
\label{eq:appMSDdualInner}
\begin{align}
    &0\leq \lambda_{i,k}\perp d_{i,k}(\mathbf{x}_i,\mathbf{x}_{-i})-z_{i,k}\geq 0,
    \quad \forall k\in\{1,\dots,K\},\\
    &0\leq \gamma_s\mathbb{P}(\xi_k)-\lambda_{i,k}\perp -z_{i,k}\geq 0,
    \quad \forall k\in\{1,\dots,K\}.
\end{align}
\end{subequations}
Indeed, if $d_{i,k}>0$, then $z_{i,k}=0$ and the first complementarity condition forces $\lambda_{i,k}=0$. If $d_{i,k}<0$, then $z_{i,k}=d_{i,k}$ and the second complementarity condition forces $\lambda_{i,k}=\gamma_s\mathbb{P}(\xi_k)$. If $d_{i,k}=0$, then $z_{i,k}=0$ and any $\lambda_{i,k}\in[0,\gamma_s\mathbb{P}(\xi_k)]$ satisfies the conditions.

Now define, for fixed $\mathbf{x}_{-i}$,
\[
\phi_i(\mathbf{x}_i,\lambda_i)
=
\mu_i^{\mathbb{P}}(\mathbf{x}_i,\mathbf{x}_{-i})
+
\sum_{k=1}^K\lambda_{i,k}d_{i,k}(\mathbf{x}_i,\mathbf{x}_{-i}).
\]
The MSD utility is $\rho_{\mathrm{MSD}}(u_i(\mathbf{x}_i,\mathbf{x}_{-i}\mid\xi))=\min_{\lambda_i}\phi_i(\mathbf{x}_i,\lambda_i)$ over the box constraints above. Since $\phi_i(\cdot,\lambda_i)$ is affine in $\mathbf{x}_i$, Danskin's theorem implies that the gradient of $\phi_i$ with respect to $\mathbf{x}_i$ at any optimal $\lambda_i$ is a supergradient of the MSD utility. Therefore, the $l$-th action risk-adjusted action value is
\[
v_{i,l}(\mathbf{x}_{-i})
=
\frac{\partial \phi_i(\mathbf{x}_i,\lambda_i)}{\partial x_i(a_l)}
=
\mu_i^{\mathbb{P}}(a_l,\mathbf{x}_{-i})
+
\sum_{k=1}^K\lambda_{i,k}
\left(
u_i(a_l,\mathbf{x}_{-i}\mid \xi_k)
-
\mu_i^{\mathbb{P}}(a_l,\mathbf{x}_{-i})
\right).
\]
If the inner minimizer $\lambda_i$ is not unique, then Danskin's theorem gives the full superdifferential as the convex hull of the corresponding gradients over all optimal choices of $\lambda_i$. Thus, any such gradient can be used in the outer best-response KKT conditions, and the general case allows any convex combination of them.

Combining the inner optimality conditions for the MSD dual variables with the outer best-response KKT conditions over $\mathbf{X}_i$ gives the following system for player $i$:
\begin{subequations}
\label{eq:appMSDdualMCP}
\begin{align}
    &0\leq \lambda_{i,k}\perp d_{i,k}(\mathbf{x}_i,\mathbf{x}_{-i})-z_{i,k}\geq 0,
    \quad \forall k\in\{1,\dots,K\},\\
    &0\leq \gamma_s\mathbb{P}(\xi_k)-\lambda_{i,k}\perp -z_{i,k}\geq 0,
    \quad \forall k\in\{1,\dots,K\},\\
    &0\leq \alpha_i-v_{i,l}(\mathbf{x}_{-i})\perp x_i(a_l)\geq 0,
    \quad \forall l\in\{1,\dots,n_i\},\\
    &\mathbf{1}^\top\mathbf{x}_i=1,\qquad \alpha_i\in\mathbb{R}.
\end{align}
\end{subequations}
The first two lines encode optimality of the worst-case MSD dual variables for the fixed profile $(\mathbf{x}_i,\mathbf{x}_{-i})$. The third and fourth lines encode optimality of $\mathbf{x}_i$ in player $i$'s best-response problem. Stacking \eqref{eq:appMSDdualMCP} over all players yields the MSD mixed complementarity formulation.

\subsection{Order-$p$ mean-semideviation}
This proof follows the same style of proof as the one with $p=1$; the reader is advised to read that before reading this terse proof.
For $1<p<\infty$, let $q$ denote the conjugate exponent, so $1/p+1/q=1$. The order-$p$ mean-semideviation utility is
\[
\rho_{\text{MSD},p}(X)
=
\mathbb{E}^{\mathbb{P}}[X]
-
\gamma_s
\left\|
\max(0,\mathbb{E}^{\mathbb{P}}[X]-X)
\right\|_p.\;,
\]
where $\|Y\|_p=(\mathbb{E}^{\mathbb{P}}[|Y|^p])^{1/p}$. In the finite-sample setting, this norm is
$\|Y\|_p=(\sum_{k=1}^K\mathbb{P}(\xi_k)|Y(\xi_k)|^p)^{1/p}$.
Its dual density can be written as $\frac{d\mathbb{Q}}{d\mathbb{P}}=1+\gamma_s(\eta-\mathbb{E}^{\mathbb{P}}[\eta])$, where $\eta\geq 0$ and $\|\eta\|_q\leq 1$. Hence, in the finite-sample setting, player $i$'s dual problem is
\[
\rho_{\text{MSD},p}(u_i(\mathbf{x}_i,\mathbf{x}_{-i}\mid \xi))
=
\min_{\eta_{i,k}\geq 0,\ \sum_{k=1}^K\mathbb{P}(\xi_k)\eta_{i,k}^q\leq 1}
\left[
\mu_i^{\mathbb{P}}(\mathbf{x}_i,\mathbf{x}_{-i})
+
\gamma_s\sum_{k=1}^K\mathbb{P}(\xi_k)\eta_{i,k}
d_{i,k}(\mathbf{x}_i,\mathbf{x}_{-i})
\right].
\]
Unlike the $p=1$ case, the inner minimization for $1<p<\infty$ is not separable across samples because the weighted $\ell_q$ constraint $\sum_{k=1}^K\mathbb{P}(\xi_k)\eta_{i,k}^q\leq 1$ couples the variables $\eta_{i,k}$.
Let $\theta_i\geq 0$ be the multiplier for $\sum_{k=1}^K\mathbb{P}(\xi_k)\eta_{i,k}^q\leq 1$. The inner KKT conditions are
\begin{subequations}
\label{eq:appMSDpInner}
\begin{align}
    &0\leq \eta_{i,k}\perp
    \gamma_s\mathbb{P}(\xi_k)d_{i,k}(\mathbf{x}_i,\mathbf{x}_{-i})
    +
    q\theta_i\mathbb{P}(\xi_k)\eta_{i,k}^{q-1}
    \geq 0,
    \quad \forall k\in\{1,\dots,K\},\\
    &0\leq \theta_i\perp
    1-\sum_{k=1}^K\mathbb{P}(\xi_k)\eta_{i,k}^q
    \geq 0.
\end{align}
\end{subequations}
By Danskin's theorem use (similar to $p=1$ case), an optimal $\eta_i$ induces
\[
v_{i,l}(\mathbf{x}_{-i})
=
\mu_i^{\mathbb{P}}(a_l,\mathbf{x}_{-i})
+
\gamma_s\sum_{k=1}^K\mathbb{P}(\xi_k)\eta_{i,k}
\left(
u_i(a_l,\mathbf{x}_{-i}\mid \xi_k)
-
\mu_i^{\mathbb{P}}(a_l,\mathbf{x}_{-i})
\right).
\]
Therefore, the order-$p$ MSD equilibrium conditions are
\begin{subequations}
\label{eq:appMSDpMCP}
\begin{align}
    &0\leq \eta_{i,k}\perp
    \gamma_s\mathbb{P}(\xi_k)d_{i,k}(\mathbf{x}_i,\mathbf{x}_{-i})
    +
    q\theta_i\mathbb{P}(\xi_k)\eta_{i,k}^{q-1}
    \geq 0,
    \quad \forall k\in\{1,\dots,K\},\\
    &0\leq \theta_i\perp
    1-\sum_{k=1}^K\mathbb{P}(\xi_k)\eta_{i,k}^q
    \geq 0,\\
    &0\leq \alpha_i-v_{i,l}(\mathbf{x}_{-i})\perp x_i(a_l)\geq 0,
    \quad \forall l\in\{1,\dots,n_i\},\\
    &\mathbf{1}^\top\mathbf{x}_i=1,\qquad \alpha_i\in\mathbb{R}.
\end{align}
\end{subequations}
Unlike the $p=1$ case, \eqref{eq:appMSDpMCP} is a mixed nonlinear complementarity formulation because of the $\ell_q$ constraint in the dual envelope.

\subsection{CVaR}
Let $X_k=u_i(\mathbf{x}_i,\mathbf{x}_{-i}\mid \xi_k)$ and recall 
$\mu_i^{\mathbb{P}}(\mathbf{x}_i,\mathbf{x}_{-i})=\sum_{k=1}^K\mathbb{P}(\xi_k)X_k$.
Using the dual representation of CVaR, the risk envelope consists of densities satisfying $0\leq d\mathbb{Q}/d\mathbb{P}\leq 1/\alpha$ and $\mathbb{E}^{\mathbb{P}}[d\mathbb{Q}/d\mathbb{P}]=1$. In the finite-sample setting, with $r_k=\mathbb{Q}(\xi_k)$, this gives
\[
\mathrm{CVaR}_\alpha[X]
=
\min_{\mathbf{r}}
\left\{
\sum_{k=1}^K r_k X_k:
\sum_{k=1}^K r_k=1,\ 
0\leq r_k\leq \frac{\mathbb{P}(\xi_k)}{\alpha}
\right\}.
\]

Since $\rho_{\text{CVaR}}(X)=(1-\gamma_c)\mathbb{E}^{\mathbb{P}}[X]+\gamma_c\mathrm{CVaR}_\alpha[X]$, define $\lambda_{i,k}=\gamma_c r_{i,k}$. Following the Ruszczyński--Shapiro dual representation, the mixed mean-CVaR utility has a risk-envelope representation with density bounded as
\[
1-\gamma_c
\leq
\frac{d\mathbb{Q}}{d\mathbb{P}}
\leq
1-\gamma_c+\frac{\gamma_c}{\alpha}.
\]
In the finite-sample setting, this means
\[
(1-\gamma_c)\mathbb{P}(\xi_k)
\leq
q_{i,k}
\leq
\left(1-\gamma_c+\frac{\gamma_c}{\alpha}\right)\mathbb{P}(\xi_k),
\qquad
\sum_{k=1}^K q_{i,k}=1.
\]
Writing $q_{i,k}=(1-\gamma_c)\mathbb{P}(\xi_k)+\lambda_{i,k}$ gives
$0\leq\lambda_{i,k}\leq \gamma_c\mathbb{P}(\xi_k)/\alpha$ and
$\sum_{k=1}^K\lambda_{i,k}=\gamma_c$, which yields the formulation below.
\begin{equation}
\label{eq:appCVaRdual}
\rho_{\text{CVaR}}(u_i(\mathbf{x}_i,\mathbf{x}_{-i}\mid \xi))
=
\min_{\substack{0\leq\lambda_{i,k}\leq \frac{\gamma_c}{\alpha}\mathbb{P}(\xi_k)\\ \sum_{k=1}^K\lambda_{i,k}=\gamma_c}}
\left[
(1-\gamma_c)\mu_i^{\mathbb{P}}(\mathbf{x}_i,\mathbf{x}_{-i})
+
\sum_{k=1}^K\lambda_{i,k}u_i(\mathbf{x}_i,\mathbf{x}_{-i}\mid \xi_k)
\right].
\end{equation}
For fixed $\mathbf{x}_i$ and $\mathbf{x}_{-i}$, this is a linear program in $\lambda_i$. Let $z_i$ be the multiplier for $\sum_{k=1}^K\lambda_{i,k}=\gamma_c$, and let $-\nu_{i,k}\geq 0$ be the multiplier for the upper bound $\lambda_{i,k}\leq \frac{\gamma_c}{\alpha}\mathbb{P}(\xi_k)$. The KKT conditions of the inner dual problem are
\begin{subequations}
\label{eq:appCVaRdualInner}
\begin{align}
    &0\leq \frac{\gamma_c}{\alpha}\mathbb{P}(\xi_k)-\lambda_{i,k}\perp -\nu_{i,k}\geq 0,
    \quad \forall k\in\{1,\dots,K\},\\
    &0\leq \lambda_{i,k}\perp
    u_i(\mathbf{x}_i,\mathbf{x}_{-i}\mid \xi_k)-z_i-\nu_{i,k}
    \geq 0,
    \quad \forall k\in\{1,\dots,K\},\\
    &\sum_{k=1}^K\lambda_{i,k}=\gamma_c.
\end{align}
\end{subequations}
Using $u_i(\mathbf{x}_i,\mathbf{x}_{-i}\mid\xi_k)=\sum_{l=1}^{n_i}x_i(a_l)u_i(a_l,\mathbf{x}_{-i}\mid\xi_k)$, For fixed $\mathbf{x}_{-i}$, define
\[
\phi_i(\mathbf{x}_i,\lambda_i)
=
(1-\gamma_c)\mu_i^{\mathbb{P}}(\mathbf{x}_i,\mathbf{x}_{-i})
+
\sum_{k=1}^K\lambda_{i,k}u_i(\mathbf{x}_i,\mathbf{x}_{-i}\mid\xi_k),
\]
where $\lambda_i$ solves the inner CVaR dual problem. Since $\phi_i(\cdot,\lambda_i)$ is affine in $\mathbf{x}_i$, Danskin's theorem implies that the derivative of $\phi_i$ with respect to $x_i(a_l)$ at an optimal $\lambda_i$ gives a supergradient component of $\rho_{\mathrm{CVaR}}(u_i(\mathbf{x}_i,\mathbf{x}_{-i}\mid\xi))$. Using $\frac{\partial}{\partial x_i(a_l)}\mu_i^{\mathbb{P}}(\mathbf{x}_i,\mathbf{x}_{-i})=\mu_i^{\mathbb{P}}(a_l,\mathbf{x}_{-i})$ and $\frac{\partial}{\partial x_i(a_l)}u_i(\mathbf{x}_i,\mathbf{x}_{-i}\mid\xi_k)=u_i(a_l,\mathbf{x}_{-i}\mid\xi_k)$, we obtain
\[
v_{i,l}(\mathbf{x}_{-i})
=
\frac{\partial \phi_i(\mathbf{x}_i,\lambda_i)}{\partial x_i(a_l)}
=
(1-\gamma_c)\mu_i^{\mathbb{P}}(a_l,\mathbf{x}_{-i})
+
\sum_{k=1}^K\lambda_{i,k}u_i(a_l,\mathbf{x}_{-i}\mid\xi_k).
\]
If the inner CVaR dual minimizer $\lambda_i$ is not unique, Danskin's theorem gives the superdifferential as the convex hull of these gradients over all optimal choices of $\lambda_i$.
Therefore, the CVaR equilibrium conditions are
\begin{subequations}
\label{eq:appCVaRdualMCP}
\begin{align}
    &0\leq \frac{\gamma_c}{\alpha}\mathbb{P}(\xi_k)-\lambda_{i,k}\perp -\nu_{i,k}\geq 0,
    \quad \forall k\in\{1,\dots,K\},\\
    &0\leq \lambda_{i,k}\perp
    \sum_{l=1}^{n_i}x_i(a_l)u_i(a_l,\mathbf{x}_{-i}\mid \xi_k)-z_i-\nu_{i,k}
    \geq 0,
    \quad \forall k\in\{1,\dots,K\},\\
    &0\leq \alpha_i-v_{i,l}(\mathbf{x}_{-i})\perp x_i(a_l)\geq 0,
    \quad \forall l\in\{1,\dots,n_i\},\\
    &\sum_{k=1}^K\lambda_{i,k}=\gamma_c,\qquad
    \mathbf{1}^\top\mathbf{x}_i=1,\qquad \alpha_i\in\mathbb{R}.
\end{align}
\end{subequations}
Stacking \eqref{eq:appCVaRdualMCP} over all players recovers the CVaR mixed complementarity formulation. In this derivation, $z_i$ arises as the equality multiplier in the dual CVaR problem and coincides with the usual CVaR threshold variable in the primal representation.

\end{document}